\documentclass[11pt]{article}

\usepackage[letterpaper,margin=1in]{geometry}
\usepackage[T1]{fontenc}
\usepackage[utf8]{inputenc}
\usepackage{lmodern}
\IfFileExists{microtype.sty}{\usepackage{microtype}}{}
\usepackage{graphicx}
\usepackage{amssymb}
\usepackage{amsmath}
\usepackage{mathtools}
\usepackage{bm}
\usepackage{enumerate}
\usepackage{url}
\usepackage{amsfonts}
\usepackage{centernot}
\usepackage[bottom]{footmisc}
\usepackage{indentfirst}
\usepackage{endnotes}
\usepackage{rotating}
\usepackage{booktabs}
\usepackage{algorithm}
\usepackage{array}
\usepackage{multirow}
\usepackage{algpseudocode}
\usepackage{enumerate}
\usepackage{bbm}
\usepackage{natbib}
\usepackage{dsfont}
\usepackage{verbatim}
\usepackage{subfig}
\usepackage{xr}
\usepackage{relsize}
\usepackage{xcolor}
\colorlet{red}{black}
\newcommand{\rev}[1]{\textcolor{red}{#1}}
\newenvironment{revblock}{\begingroup\color{red}}{\endgroup}

\usepackage{tabularx}
    \newcolumntype{L}{>{\raggedright\arraybackslash}X}
\usepackage[
  colorlinks=true,
  linkcolor=blue,
  citecolor=blue,
  urlcolor=blue,
  pdftitle={High-dimensional Change-point Detection Using Generalized Homogeneity Metrics},
  pdfauthor={Shubhadeep Chakraborty, Runmin Wang, Xianyang Zhang},
  pdfkeywords={High Dimensionality, Multiple Change-Point Detection, Seeded Narrowest-Over-Threshold, Two Sample Test}
]{hyperref}

\newenvironment{keywords}
  {\par\small\noindent\textbf{Keywords:} }
  {\par\normalsize\vskip 0.25ex}

\newcommand{\BlackBox}{\rule{1.5ex}{1.5ex}}
\newenvironment{proof}
  {\par\noindent\textbf{Proof}\ }
  {\hfill\BlackBox\\[2mm]}

\newcommand{\argmax}{\operatornamewithlimits{argmax}}

\newcommand{\argmin}{\mathop{\rm argmin~}}
\newcommand{\Var}{\mathop{\rm Var}}

\newcommand{\cF}{\mathcal{F}}

\makeatletter
\newcases{nocases}
  {\quad}
  {$\m@th\displaystyle{##}$\hfil}
  {$\m@th\displaystyle{##}$\hfil}
  {.}{.}
\makeatother

\newcommand{\field}[1]{\mathbb{#1}}

\newcommand{\E}{\field{E}}
\newcommand{\tran}{^{\top\kern-\scriptspace}}

\newtheorem{assumption}{Assumption}[section]
\newtheorem{theorem}{Theorem}[section]

\newtheorem{proposition}{Proposition}[section]
\newtheorem{lemma}{Lemma}[section]

\newtheorem{example}{Example}[section]
\newtheorem{remark}{Remark}[section]

\newtheorem{definition}{Definition}[section]

\newcommand{\Rmnum}[1]{\expandafter\romannumeral #1}

\def\lf{\lfloor}
\def\rf{\rfloor}

\def\bb{\mathbb}
\def\cal{\mathcal}
\def\dis{\displaystyle}
\font\n=cmcsc10 scaled \magstep1
\def\cov{{\mbox{cov}}}
\def\var{{\mbox{var}}}

\newcommand{\bigCI}{\mathrel{\text{\scalebox{1.07}{$\perp\mkern-10mu\perp$}}}}

\newcommand{\floor}[1]{\lfloor #1 \rfloor}

\DeclareMathOperator*{\plim}{plim}

\newcounter{alphasect}
\def\alphainsection{0}

\let\oldsection=\section
\def\section{%
  \ifnum\alphainsection=1%
    \addtocounter{alphasect}{1}
  \fi%
\oldsection}%

\renewcommand\thesection{%
 \ifnum\alphainsection=1%
   \Alph{alphasect}%
 \else
   \arabic{section}%
 \fi%
}%

\newenvironment{alphasection}{%
  \ifnum\alphainsection=1%
    \errhelp={Let other blocks end at the beginning of the next block.}
    \errmessage{Nested Alpha section not allowed}
  \fi%
  \setcounter{alphasect}{0}
  \def\alphainsection{1}
}{%
  \setcounter{alphasect}{0}
  \def\alphainsection{0}
}%

\begin{document}

\title{High-dimensional Change-point Detection Using Generalized Homogeneity Metrics}

\author{%
\begin{tabular}{c}
\textbf{Shubhadeep Chakraborty}\textsuperscript{*}\quad
\href{mailto:shubhadeep.chakraborty@bms.com}{\texttt{shubhadeep.chakraborty@bms.com}}\\
\small Bristol Myers Squibb Company\\
\small Lawrenceville, NJ 08648, USA\\[0.65em]
\textbf{Runmin Wang}\textsuperscript{*}\quad
\href{mailto:runminw@tamu.edu}{\texttt{runminw@tamu.edu}}\\
\small Department of Statistics\\
\small Texas A\&M University\\
\small College Station, TX 77843, USA\\[0.65em]
\textbf{Xianyang Zhang}\quad
\href{mailto:zhangxiany@stat.tamu.edu}{\texttt{zhangxiany@stat.tamu.edu}}\\
\small Department of Statistics\\
\small Texas A\&M University\\
\small College Station, TX 77843, USA\\[0.5em]
\footnotesize \textsuperscript{*}These authors contributed equally to this work.
\end{tabular}}

\date{}
\maketitle
 
\begin{abstract}
Change-point detection is a classical yet vibrant field of research in statistics. In this work, we address the problem of detecting abrupt changes in the data-generating distributions of a sequence of high-dimensional observations beyond the first two moments. This problem has remained substantially less explored in the existing literature, especially in the high-dimensional context, compared to detecting changes in the mean or the covariance structure. To the best of our knowledge, this is one of the first attempts to detect and localize general types of distributional changes in the high-dimensional regime. We develop a distance-based method to (i) test for the existence of a change-point, and (ii) identify the change-point locations in an independent sequence of high-dimensional observations. Our approach rests upon recent distance-based tests for the homogeneity of two high-dimensional distributions. We construct a single change-point test statistic based on a cumulative sum process in an embedded Hilbert space and rigorously derive its limiting null distribution and prove asymptotic consistency under the high-dimensional medium sample size (HDMSS) framework. Subsequently, we combine our statistics with the Narrowest-Over-Threshold (NOT) strategy to recursively estimate and test for multiple change-point locations. \rev{We also study a componentwise monotone-invariant, rank-based extension; because its pseudo-observations are pooled empirical mid-ranks and are therefore dependent, we present this version as a practically useful heuristic extension supported by simulation evidence rather than as a fully proved analogue of the original statistic.} The superior performance of our methodology compared to existing procedures is illustrated via extensive simulation studies and an application to U.S. stock return data during the global financial crisis. The proposed method is implemented in the R package \texttt{KDist}, available at \url{https://github.com/zhangxiany-tamu/KDist}.
\end{abstract}


\section{Introduction}
Change-point detection is a well-established and active area of research in statistics that aims to identify a lack of homogeneity in a sequence of time-ordered observations. It finds an abundance of applications in a wide variety of fields, including bioinformatics \citep{picard, curtis}, neuroscience \citep{park}, digital speech processing \citep{rabiner}, and social network analysis \citep{mcculloh}. We refer the readers to \cite{aue}, \cite{jan}, and \cite{T2020} for some recent reviews on this topic. A critical and fundamental problem in detecting structural breaks in multivariate data is the detection of changes in the mean vector. The mean change problem has been extensively studied when the dimension is low compared to the sample size. However, high-dimensional data is frequently encountered in many scientific areas in the big data era. Recent works that tackle the detection of mean change for high-dimensional data include \cite{enik}, \cite{jirak}, \cite{CF}, \cite{WS}, \cite{yu2021}, \cite{zhang2022}, \cite{WVS} and \cite{wangfeng2023}. Detecting changes in the covariance structure in a sequence of high-dimensional observations is also an important problem. Recent works in this area include \cite{avanesov}, \cite{dette}, \cite{wang2021optimal}, \cite{kaul2023inference}, \cite{li2023online}, and \cite{ryan2023detecting}, among others.

A substantial part of the existing literature on change-point detection has historically focused on changes in specific parameters like the mean or covariance structure. Beyond these specific parametric shifts, a growing body of literature, to which our work contributes, focuses on detecting and localizing more general types of changes in the data-generating distribution. Within the non-parametric change-point detection paradigm, it is crucial to distinguish between methods that are truly invariant to monotone transformations of the data and operate on ranks, thereby avoiding moment assumptions and other approaches that, while not assuming specific parametric forms for the underlying distributions, may rely on different principles and not fully share these invariance properties.

Pioneering work in `true' non-parametric change-point detection includes \cite{carlstein} and, notably, \cite{dumbgen}. The latter provides a fundamental and remarkably general theoretical framework, permitting data to reside in any measurable space. This generality is achieved by quantifying a signed measure of change with a seminorm whose suitable behavior for enabling detection is controlled via a Vapnik-Chervonenskis type assumption (Assumption 2.1 in \citealp{dumbgen}). Methodologies aligned with this `true' non-parametric spirit include the non-parametric maximum likelihood approach of \cite{zou}, which uses BIC for detecting multiple change-points in real-valued data, and the work of \cite{lung} based on marginal rank statistics.

Alongside these rank-based and distribution-free methods, another significant stream of research employs distance or kernel-based statistics to detect distributional changes. While these methods are often termed non-parametric because they do not assume a specific parametric family for the distributions, they may differ from rank-based tests in their invariance properties and can implicitly or explicitly involve moment assumptions. For instance, \cite{MJ} proposed the E-Divisive procedure based on energy distance for estimating multiple change-point locations in multivariate observations of arbitrary (but fixed) dimensions. \cite{biau} rigorously derived the asymptotic distribution of the statistic proposed by \cite{MJ}, adding theoretical justifications. Other kernel-based approaches include \cite{hc}, who proposed a kernel-based procedure using a segment neighborhood algorithm without providing theoretical guarantees for their method. Built on the idea of \cite{hc}, \cite{arlot} developed a kernel-based multiple change-point detection algorithm and studied its theoretical properties. Graph-based tests have been proposed by \cite{CZ} and \cite{CC} to detect structural breaks. While the approach by \cite{CZ} is noted to be more effective for detecting location alternatives rather than scale alternatives, it has lower detection power when changes occur away from the center of the sequence. \cite{CC} addressed these limitations by introducing more robust tests that are effective for both location and scale alternatives. However, our numerical studies indicate that these graph-based tests may not be very effective in detecting changes in higher-order moments for high-dimensional data.

Energy distance, originally proposed by \cite{sr2004, sr2005} and \cite{bf}, is a classical distance-based measure of equality of two multivariate distributions, taking the value zero if and only if the two random vectors are identically distributed. Such a complete characterization of the homogeneity of distributions lends itself to reasonable use in one-sample goodness-of-fit testing and two-sample testing for equality of distributions. In the high dimension low sample size framework, \cite{us} and \cite{zs2019} showed a striking result that energy distance based on the usual Euclidean distance could not completely characterize the homogeneity of the two high-dimensional distributions in the sense that it can only detect the {\it equality of means and the traces of covariance matrices} of the two high-dimensional random vectors when the sample sizes are fixed, and the dimension grows. In other words, the Euclidean energy distance fails to detect inhomogeneity between two high-dimensional distributions beyond the first two moments. A recent study by \cite{yz} reveals a more delicate interplay between the moment discrepancy that the energy distance can detect and the dimension-and-sample orders. To overcome such a limitation, \cite{us} proposed a new class of homogeneity metrics that inherits the desirable properties of energy distance in the low-dimensional setting. And more importantly, in the high-dimensional setup, the new class of homogeneity metrics is capable of detecting the {\it pairwise homogeneity of the low-dimensional marginal distributions}, going beyond the scope of the Euclidean energy distance. The proposed class of homogeneity metrics can capture a wider range of inhomogeneity between distributions compared to the classical Euclidean energy distance in the high-dimensional framework. The core of their methodology is a new way of defining the distance between sample points (interpoint distance) in high-dimensional Euclidean spaces.

This paper focuses on detecting and estimating an unknown number of multiple change-point locations in an independent sequence of ${\mathbb R}^p$-valued observations of sample size $n$, where $p$ can by far exceed $n$. The main contributions of the paper are summarized as follows:
\begin{enumerate}
    \item The majority of the research in this field concentrates on detecting changes in the mean or covariance matrix for high-dimensional data. To the best of our knowledge, we make one of the first attempts in the literature to detect and localize general types of changes in the underlying distribution beyond the first two moments in the high-dimensional regime. 
    
    \item Motivated by the cumulative sum process in an embedded space, we introduce a new change-point detection statistic and rigorously derive its limiting null distribution and asymptotic consistency under alternatives. We propose algorithms for single and multiple change-point detection and estimation. A unique advantage of the proposed method is that it is sensitive to changes in the mean, covariance structure, and higher-order moments.
    
    \item  Moreover, we propose a recursive estimation procedure using the Seeded Narrowest-Over-Threshold (Seeded NOT) strategy \citep{baranowski2019narrowest} to recursively estimate and test for the significance of (an unknown number of) multiple change-point locations. This strategy improves upon standard binary segmentation by effectively handling frequent changes and avoiding the masking problem.

    \item \rev{Addressing the trade-off between moment-sensitivity and invariance, we explore a componentwise monotone-invariant extension based on empirical marginal ranks. Since the resulting pseudo-observations are no longer independent, this extension is advertised as a heuristic, simulation-supported procedure that improves robustness and invariance, while the formal HDMSS theory in this paper is established for the original statistic.}

    \item \rev{To address the computational burden of high-dimensional distance-based methods, we propose two practical approximations, Surrogate A and Surrogate B. They can reduce the standard $O(n^2p)$ cost in the regimes specified in Section~\ref{sec:surrogates}, but they introduce approximation error. We therefore do not claim consistency or unconditional computational gains for either surrogate without additional choices of the sketch size or pair-subsampling rate.}
\end{enumerate}

Our approach rests upon distance-based two-sample tests for the homogeneity of two high-dimensional distributions. We first construct a single change-point detection statistic $M_n$ based on the homogeneity metrics proposed by \cite{us} by defining a cumulative sum process in an embedded Hilbert space. It generalizes the single change-point statistic developed by \cite{MJ} and \cite{biau} to the high-dimensional setup, providing a unifying framework. Testing for the statistical significance of the estimated candidate change-point location necessitates determining the quantiles of the distribution of $M_n$. One theoretical innovation of this paper is to rigorously derive the limiting null distribution of $M_n$ as the dimension $p$ and the sample size $n$ diverge to infinity. Such a setup is typically known in the literature as the high-dimensional medium sample size (HDMSS) framework. The intrinsic difficulty is establishing the uniform weak convergence of an underlying stochastic process under certain moment assumptions, which has been non-trivial and challenging. Because of the pivotal nature of the limiting null distribution, its quantiles can be approximated using a large number of Monte Carlo simulations. To further improve the finite sample performance, we propose an algorithm for single change-point detection based on a permutation procedure to approximate the quantiles of the distribution of $M_n$. \rev{Under the single change-point alternative, we separately prove divergence of the studentized scan statistic $M_n$ and a consistency rate for the unstudentized location estimator $\widehat\nu^*$ in \eqref{def-nu}; we do not identify these as one common result for the practical studentized locator.}

Subsequently, we combine the idea of Seeded NOT proposed by \cite{baranowski2019narrowest} to recursively estimate and test for the significance of multiple change-point locations. The superior performance of our procedure compared to the existing methodologies is illustrated over extensive simulated datasets. When applied to the stock return data observed during the global financial crisis in the United States, our method furnishes more reasonable and meaningful estimates of significant change-point locations given the historical sequence of eventualities compared to the other existing methods. Finally, we briefly illustrate an extension of our methodology to incorporate directed and undirected graph information. Further research along this line is well underway. 

We emphasize that the change-point detection problem addressed in this work is much more challenging than the two-sample problem \citep{us} for several reasons: (i) As the locations of the change-points are unknown, the change-point detection procedure requires an extra layer of complication to search for the optimal change-point locations that divide the data into potentially homogeneous groups for comparison (in two sample problem, the two groups are pre-determined); (ii) The technical analysis requires demonstrating that the underlying stochastic process associated with the two-sample test statistic converges weakly to a limit in a functional space, necessitating a more involved technical analysis as can be seen from our proofs; (iii) We conduct power analysis of our testing procedure under HDMSS, while \cite{us} only examined the asymptotic behavior of the two-sample test under the null hypothesis; (iv) \rev{We explore, as a heuristic extension, rank-based measures using componentwise monotone transformations, which enhance robustness to outliers and ensure invariance under monotonic transformations of the data;} (v) We also discuss the incorporation of external graph information to fully characterize the discrepancy between two high dimensional distributions, which is a new aspect not explored in the literature.

\emph{Notation}. Denote by $\Vert\cdot\Vert$ the Euclidean norm in $\mathbb{R}^p$. Let $0_p$ be the origin of $\mathbb{R}^p$. For a set $\mathcal{S}\subseteq [p]:=\{1,2,\dots,p\}$ and $z=(z_1,\dots,z_p)\in\mathbb{R}^p$, we let $\text{card}(\mathcal{S})$ denote the cardinality of $\mathcal{S}$ and $z_\mathcal{S}=(z_i:i\in \mathcal{S})$ be the subvector of $z$ containing the components whose indices are in $\mathcal{S}$. We use `$X \overset{d}{=} Y$' to indicate that $X$ and $Y$ are identically distributed. Let $X'$ be an independent copy of $X$. `O' and `o' stand for the usual notations in mathematics: `is no larger than' and `is ultimately smaller than,' respectively. We use the symbol `$a \lesssim b$' to indicate that $a \leq C \,b$\, for some constant $C>0$. We utilize the order in probability notations, such as stochastic boundedness $O_p$ (big O in probability), convergence in probability $o_p$ (small o in probability), and equivalent order $\asymp_p$, which is defined as follows: for a sequence of random variables $\{Z_n\}_{n=1}^{\infty}$ and a sequence of real numbers $\{a_n\}_{n=1}^{\infty}$, $Z_n \asymp_p a_n$ if and only if $Z_n/a_n = O_p(1)$ and $a_n/Z_n = O_p(1)$ as $n \to \infty$. If $Z_n \overset{P}{\rightarrow} Z$ as $n \to \infty$, then we say $\text{plim}_{n \to \infty} Z_n = Z$. For a metric space $(\mathcal{X}, \rho )$, let $\mathcal{M}(\mathcal{X})$ and $\mathcal{M}_1(\mathcal{X})$ denote the set of all finite signed Borel measures on $\mathcal{X}$ and all probability measures on $\mathcal{X}$, respectively. Define $\mathcal{M}^1_{\rho }(\mathcal{X}):= \{v \in \mathcal{M}(\mathcal{X}) \,:\, \exists\, x_0 \in \mathcal{X} \; \text{s.t.}\;\int_{\mathcal{X}} \rho (x,x_0)\, d|v|(x) < \infty\}$. Let $\mathbbm{1}(A)$ denote the indicator function associated with a set $A$. For a compact set $\mathcal{T}$, define $L^{\infty}(\mathcal{T}):= \{ f:\mathcal{T}\to \mathbb{R} \; ;\, \Vert f \Vert_{\infty} = \sup_{t \in \mathcal{T}} |f(t)| < \infty \}$. Weak convergence in $L^{\infty}(\mathcal{T})$ is denoted by `$\rightsquigarrow$'. Let ${\bf 1}_n = (1,\dots, 1) \in \mathbb{R}^n$. Write $a\vee b=\max\{a,b\}$ and $a\wedge b=\min\{a,b\}$. Finally, denote by $\lfloor a \rfloor$ the integer part of $a\in\mathbb{R}$.

\section{Distance-based homogeneity tests}\label{sec:overview}
\subsection{Generalized energy distance}\label{sec:overview:ed}
The energy distance \citep{sr2005, bf}, or the Euclidean energy distance, between two random vectors $X,Y\in \mathbb{R}^p$ and $X \perp \!\!\! \perp Y$ with $\mathbb{E} \Vert X \Vert < \infty$ and $\mathbb{E} \Vert Y \Vert < \infty$, is defined as
\begin{align}\label{ed def 0}
	E(X,Y)=\frac{1}{c_p}\int_{\mathbb{R}^{p}}\frac{|f_{X}(t)-f_Y(t)|^2}{\Vert t \Vert^{1+p}}\,dt \,,
\end{align}
where $f_X$ and $f_Y$ are the characteristic functions of $X$ and $Y$ respectively, and $c_{p}=\pi^{(1+p)/2}/\,\Gamma((1+p)/2)$
is a constant with $\Gamma(\cdot)$ being the complete gamma function. Theorem 1 in \cite{sr2005} shows that \,$E(X,Y) \geq 0$ and the equality holds if and only if $X \overset{d}{=} Y$. In other words, energy distance can completely characterize the homogeneity between two multivariate distributions. An equivalent expression for $E(X,Y)$ is given by
\begin{equation}\label{ed def}
	E(X,Y)=2\,\mathbb{E} \Vert X-Y \Vert - \mathbb{E} \Vert X-X'\Vert - \mathbb{E} \Vert Y-Y'\Vert \; ,
\end{equation}
where $(X',Y')$ is an independent copy of $(X,Y)$. 

\begin{definition}[Generalized energy distance]
	For an arbitrary metric space $(\mathcal{X}, \rho)$, the generalized energy distance between $X \sim P_X$\, and \,$Y \sim P_Y$ where $P_X, P_Y \in \mathcal{M}_1(\mathcal{X}) \cap \mathcal{M}^1_{\rho}(\mathcal{X})$ is defined as
	\begin{equation}\label{ed def general}
		E_\rho(X,Y) \;=\; 2\,\mathbb{E} \,\rho(X,Y) - \mathbb{E} \,\rho(X,X') - \mathbb{E} \,\rho(Y,Y') \;.
	\end{equation}
\end{definition}

\begin{definition}[Spaces of negative type]\label{negative type}
	The metric space $(\mathcal{X},\rho)$ is said to have negative type if for any $n\ge 2$, $x_{1},\dots,x_{n}\in\mathcal{X}$, and $\alpha_{1},\cdots,\alpha_{n}\in\mathbb{R}$ with $\sum_{i=1}^{n}\alpha_{i}=0$,  
	$\sum_{i=1}^{n}\sum_{j=1}^{n}\alpha_{i}\alpha_{j}\rho(x_{i},x_{j})\le 0.$
	Suppose $P,Q\in \mathcal{M}_1(\mathcal{X}) \cap \mathcal{M}^1_{\rho}(\mathcal{X})$. When $(\mathcal{X},\rho)$ has negative type, 
	\begin{align}\label{supp-eq1}
		\int\rho(x_{1},x_{2})d(P-Q)^{2}(x_{1},x_{2})\le 0.     
	\end{align}
	We say that $(\mathcal{X},\rho)$ has strong negative type if it has negative type and the equality in (\ref{supp-eq1}) holds only when $P=Q$. 
\end{definition}

Below, we provide some examples of spaces of strong negative type. 
\begin{itemize}
    \item When $\mathcal{X}=\mathbb{R}^p$ and $\rho$ is the Euclidean distance, $(\mathcal{X},\rho)$ is of strong negative type. More generally, according to Theorem 3.16 in \cite{lyons}, every separable Hilbert space (with the corresponding inner product induced distance) is of strong negative type.
    \item Consider $\mathcal{X}=\mathbb{R}^p$ and $\rho(z,z')=\mathcal{K}(z,z)+\mathcal{K}(z',z')-2\mathcal{K}(z,z')$ for some kernel function $\mathcal{K}$. By Proposition 29 in \cite{ssgf}, if $\mathcal{K}$ is a characteristic kernel, then $(\mathcal{X},\rho)$ is of strong negative type. This equivalence highlights the connection between distance-based methods in metric spaces of strong negative type and kernel-based methods using characteristic kernels (e.g., Gaussian or Laplace kernels).
    \item If $(\mathcal{X},\rho)$ has negative type, then $(\mathcal{X},\rho^a)$ is of strong negative type for any $0<a<1$; see Remark 3.19 of \cite{lyons}.
\end{itemize}

If $(\mathcal{X}, \rho)$ has a strong negative type, then $E_\rho(X,Y)=0$ if and only if $X \overset{d}{=} Y$, or in other words, the complete characterization of the homogeneity of two distributions holds in any metric spaces of strong negative type \citep{lyons, ssgf}. Thus, the quantification of homogeneity of distributions by the Euclidean energy distance given in (\ref{ed def}) is just a special case when $\rho$ is the Euclidean distance on $\mathcal{X} = \mathbb{R}^p$. Suppose $\mathbf{X}_n=\{X_i\}^{n}_{i=1}$ and $\mathbf{Y}_m=\{Y_i\}^{m}_{i=1}$ are two independent i.i.d samples on $X$ and $Y$ taking values in $(\mathcal{X},\rho)$. A U-statistic type estimator of the generalized energy distance between $X$ and $Y$ is defined as
\begin{align*}
	\widehat{E}_\rho(\mathbf{X}_n,\mathbf{Y}_m)=&\frac{2}{n m}\sum_{i=1}^{n}\sum_{j=1}^{m}\rho(X_{i}, Y_{j})-\frac{1}{n(n-1)}\sum_{1 \leq i\neq j \leq n} \rho(X_{i},X_{j})\\&-\frac{1}{m(m-1)}\sum_{1\leq i\neq j \leq m} \rho(Y_{i},Y_{j}).
\end{align*}
We refer the readers to Section A.1 in the Supplementary Materials of \cite{us} for a comprehensive overview of the properties and asymptotic behavior of the U-statistic type estimator of $E_\rho(X,Y)$ in the low-dimensional setting.

\subsection{Generalized energy distance in high dimensions}\label{sec:overview:ours}
The question of interest is how the classical distance-based homogeneity metrics like energy distance behave in the high-dimensional framework. Consider two $\mathbb{R}^p$-valued random vectors $X=(x_1,\dots,x_p)$ and $Y=(y_1,\dots,y_p)$. \cite{us} showed a striking result: when the dimension grows high, the Euclidean energy distance between $X$ and $Y$ can only capture the equality of the means and the first spectral means, i.e., $\mu_X = \mu_Y$ and $\text{tr}(\Sigma_X) = \text{tr}(\Sigma_Y)$, where $\mu_X$ and $\mu_Y$, and $\Sigma_X$ and $\Sigma_Y$ are the mean vectors and the covariance matrices of $X$ and $Y$, respectively.

To illustrate, consider the case where $X \sim N(\mu, \mathbf{I}_p)$ with $\mu = \mathbf{1}_p \in\mathbb{R}^p$\, and the components of $Y$ independently follow Exponential\,$(1)$ for $1\leq i \leq p$. That is, $\mu_X = \mu_Y$ and $\text{tr}(\Sigma_X) = \text{tr}(\Sigma_Y)$ although $X$ and $Y$ have different distributions. The homogeneity test based on the Euclidean energy distance has trivial power in this case. Such a limitation of the classical Euclidean energy distance arises essentially due to the use of Euclidean distance. \cite{us} proposed a new class of homogeneity metrics to overcome such a limitation, which is based on a new way of defining the distance between sample points (interpoint distance) in high-dimensional Euclidean spaces. Here, we present a slightly generalized version of their distance by allowing the groups (i.e., $\mathcal{S}_i$'s below) to overlap.

\begin{definition}[Generalized Euclidean distance]\label{def_GED}
	{\rm 
		Consider a collection of subsets $\{\mathcal{S}_i:1\leq i\leq g\}$ with $\mathcal{S}_i\subseteq [p]:=\{1,2,\dots,p\}$ and $\text{card}(\mathcal{S}_i)=d_i$.  Suppose $\cup_{i=1}^{g}\mathcal{S}_i=[p]$ and $\rho_i$ is a distance of strong negative type on $\mathbb{R}^{d_i}$
		for $1\leq i\leq g$. For $z,z'\in\mathbb{R}^p$, we define the generalized Euclidean distance as
		$$\gamma(z,z'):= \sqrt{\rho_1 (z_{\mathcal{S}_1}, z_{\mathcal{S}_1}') + \cdots+\rho_g (z_{\mathcal{S}_g}, z_{\mathcal{S}_g}')},$$ 
		which can be shown to be a valid metric on $\mathbb{R}^{p}$.
	}
\end{definition}

For illustration, consider the case where \( g = p \) and \( d_i = 1 \) for all \( 1 \leq i \leq g \).
\begin{itemize}
    \item When $\rho_i$ is the Euclidean distance on $\mathbb{R}$, the metric boils down to
    $$\gamma(z,z')=\| z-z' \|_1^{1/2}=\sqrt{\sum_{j=1}^p |z_j - z'_j|},$$
    where $\| z \|_{1} = \sum_{j=1}^p |z_j|$ is the $l_1$ or the absolute norm on $\mathbb{R}^p$. 
    \item When $\rho_i$ is the Laplace kernel induced distance on $\mathbb{R}$, the metric becomes
    $$\gamma_L(z,z')=\sqrt{\sum_{j=1}^p\left(2 - 2e^{-|z_j - z'_j|/h_j}\right)},$$
    where $h_j$ represents the bandwidth parameter for the $j$-th dimension, controlling the characteristic length scale of the Laplace kernel for each feature.
    \item When $\rho_i$ is the Gaussian kernel induced distance on $\mathbb{R}$, the metric takes the form of
    $$\gamma_G(z,z')=\sqrt{\sum_{j=1}^p\left(2 - 2e^{-(z_j - z'_j)^2/(2h_j^2)}\right)},$$
    where $h_j$ represents the bandwidth parameter for the $j$-th dimension, controlling the characteristic length scale of the Gaussian kernel for each feature.
\end{itemize}

The new class of distance-based homogeneity metrics replaces the Euclidean distance in the definition of energy distance with this proposed distance. For fixed $p$, $(\mathbb{R}^p, \gamma)$ is shown to have a strong negative type and hence $E_\gamma(X,Y)=0$ if and only if $X \overset{d}{=} Y$. In other words, $E_\gamma(X,Y)$ completely characterizes the homogeneity of the distributions of $X$ and $Y$ in the low-dimensional setting. Theorem 4.1 and Lemma 4.1 of \cite{us} show that when $p$ grows high, and the dimensions of the sub-vectors remain fixed, $E_\gamma(X,Y)$ can capture the pairwise homogeneity of the marginal distributions of $X_{\mathcal{S}_i}$ and $Y_{\mathcal{S}_i}$. Clearly $X_{\mathcal{S}_i} \overset{d}{=} Y_{\mathcal{S}_i}$ for $1\leq i\leq g$ implies $\mu_X = \mu_Y$ and $\text{tr} (\Sigma_X) = \text{tr} (\Sigma_Y)$, and therefore the proposed class of homogeneity metrics can capture a wider range of inhomogeneity of distributions compared to the Euclidean energy distance in the high-dimensional framework.

\subsection{Two-sample t-test}\label{sec:t-test}
\cite{us} introduced a two-sample t-test for high-dimensional inference based on the generalized homogeneity metrics. Given the samples $\mathbf{X}_n$ and $\mathbf{Y}_m$, we first define the double-centered distance matrices. To avoid confusion with the scaling factor $a_{nm}$ defined later, we denote the entries of these matrices by capital letters:
\begin{align*}
\widetilde{A}_{k,k'}:=&\gamma(X_{k}, X_{k'}) -\frac{1}{n-2} \sum_{j=1}^n \gamma (X_{k}, X_j)- \frac{1}{n-2} \sum_{i=1}^n \gamma (X_{i}, X_{k'})\\
&+ \frac{1}{(n-1)(n-2)} \sum_{i,j=1}^n \gamma (X_i, X_j),\\
\widetilde{B}_{l,l'}:=&\gamma(Y_{l}, Y_{l'}) -\frac{1}{m-2} \sum_{j=1}^m \gamma (Y_{l}, Y_j)- \frac{1}{m-2} \sum_{i=1}^m \gamma (Y_{i}, Y_{l'})
\\&+ \frac{1}{(m-1)(m-2)} \sum_{i,j=1}^m \gamma (Y_i, Y_j),\\
\widetilde{D}_{k,l}:=&\gamma(X_{k},Y_{l})-\frac{1}{m}\sum^{m}_{j=1}\gamma(X_{k},Y_{j})-\frac{1}{n}\sum^{n}_{i=1}\gamma(X_{i},Y_{l})+\frac{1}{nm}\sum_{i=1}^{n}\sum^{m}_{j=1}\gamma(X_{i},Y_{j}),
\end{align*}
where $1\leq k,k' \leq n$ and $1\leq l,l' \leq m$. Define the pooled variance estimator
\begin{align*}
	&\widehat{S}^2(\mathbf{X}_n,\mathbf{Y}_m)=\frac{
		4v_n \widehat{\mathcal{D}}^2(\mathbf{X}_n)+4v_m \widehat{\mathcal{D}}^2(\mathbf{Y}_m)
		+4(n-1)(m-1)\widehat{\mathcal{C}}(\mathbf{X}_n,\mathbf{Y}_m)}{v_n+v_m+(n-1)(m-1)},    
\end{align*}
where $\widehat{\mathcal{D}}^2$ and $\widehat{\mathcal{C}}$ are the sample distance variance and the cross distance covariance, defined respectively by
\begin{align*}
	&\widehat{\mathcal{D}}^2(\mathbf{X}_n) = \frac{1}{n(n-3)} \sum_{1\leq k\neq k' \leq n} \widetilde{A}^2_{k,k'},\quad \widehat{\mathcal{D}}^2(\mathbf{Y}_m)=\frac{1}{m(m-3)} \sum_{1\leq l\neq l' \leq m} \widetilde{B}^2_{l,l'},\\
 &\widehat{\mathcal{C}}(\mathbf{X}_n,\mathbf{Y}_m)=\frac{1}{(n-1)(m-1)}\sum^{n}_{k=1}\sum_{l=1}^{m}\widetilde{D}_{k,l}^2,
\end{align*}
and $v_k = k(k-3)/2$\, for\, $k=n,m$. The two-sample t-statistic is defined as
\begin{align*}
	T(\mathbf{X}_n,\mathbf{Y}_m)&=\frac{\widehat{E}_{\gamma}(\mathbf{X}_n,\mathbf{Y}_m)}{a_{nm}\,\widehat{S}(\mathbf{X}_n,\mathbf{Y}_m)}\quad \text{ where } \quad a_{nm}^2 = \frac{1}{nm}  + \frac{1}{2n(n-1)}+ \frac{1}{2m(m-1)}.
\end{align*}
Note that the construction of the pooled variance estimator and hence the two-sample statistic requires $n,m\geq 4.$ The computational complexity of calculating $T(\mathbf{X}_n,\mathbf{Y}_m)$ is nominally $O((m\vee n)^2p)$; however, as discussed in Section \ref{sec:surrogates}, efficient surrogates can be employed to reduce this burden in ultra-high dimensions.

Under the moment assumptions detailed in \cite{us} (analogous to Assumptions \ref{ass1_new} and \ref{ass2_new}), Theorem B.1 in \cite{us} shows that under $H_0: X\overset{d}{=}Y$, 
$$T(\mathbf{X}_n,\mathbf{Y}_m) \overset{d}{\rightarrow} N(0,1)$$ 
as $n,m,p \to \infty$. The proposed change-point detection statistic will be constructed by recursively calculating the two-sample t-statistic for the split data sequence at all potential candidate change-point locations.

\section{High-dimensional change-point detection}\label{methods}

\subsection{Problem statement}
With the above background knowledge, we now turn to the change-point detection problem. Consider an independent sequence of $\mathbb{R}^p$-valued observations $\{X_t\}_{t=1}^n$ with $X_t=(X_{t,1},\dots,X_{t,p})$, where the dimension $p$ can be much higher than the sample size $n$. We are concerned with testing the null hypothesis $H_0: X_t \sim F_1$\, for\, $t=1,\dots,n$\, against the alternative
\begin{align}\label{alt:mult}
	H_1: \exists\, N \in \mathbb{Z}^+,\quad 1\leq \nu_1 < \dots < \nu_{N} < n, \quad
	X_t \sim \begin{cases}
		F_1, & 1\leq t \leq \nu_1, \\
		F_2, & \nu_1 + 1 \leq t \leq \nu_2,\\
		\vdots \\
		F_{N+1}, & \nu_{N} + 1 \leq t \leq n,
	\end{cases}
\end{align}
where two consecutive probability distributions $F_i$ and $F_{i+1}$ differ on a set with non-zero measure for $1\leq i\leq N$. Note that if there is no change point, we let $N = 0$ and follow the convention by setting $\nu_{N+1} = n$. For $i = 1,2,\dots,N$, let $\zeta_i = \lim_{n \rightarrow \infty} \nu_i/n$ (assuming its existence).

\subsection{A self-normalized U-statistic approach for single change-point detection} 
A recent approach based on distance methods, for example \cite{MJ}, involves a statistic built upon the Euclidean energy distance \citep{sr2004, sr2005, bf}:
$$ \mathcal{Q}_n=\max_{2\leq k\leq n-2}\frac{k(n-k)}{n}\,\widehat{E}_\rho(\mathbf{X}_{1:k},\mathbf{X}_{k+1:n}) $$
where $\rho$ is the usual Euclidean distance and $\mathbf{X}_{a:b}=(X_a,X_{a+1},\dots,X_b)$. While this method is effective in low-dimensional scenarios, it encounters limitations with high-dimensional data. Specifically, it may fail to detect structural changes in sequences of high-dimensional observations that extend beyond differences in the first two moments. Our numerical studies, presented in Section \ref{sec:num}, corroborate this observation. Furthermore, a practical challenge with this statistic is that its limiting distribution under the null hypothesis is non-pivotal.

To overcome these limitations, our proposed methodology is built upon the generalized homogeneity metric (as discussed in Section \ref{sec:overview:ours}), which is adept at comparing two high-dimensional distributions. The natural strategy for estimating the generalized energy distance is to employ a U-statistic estimator, which is unbiased. Following this, we introduce our test statistic $M_n$:
\begin{align}\label{our_statistic_revised} 
	M_n\;:=\;\max_{4\leq k\leq n-4}\,\frac{k(n-k) }{n^2}\,T_n(k)
\end{align}
where $T_n(k):=T(\mathbf{X}_{1:k},\mathbf{X}_{k+1:n})$ is the two-sample t-statistic introduced in Section \ref{sec:t-test}. The term $T_n(k)$ can be viewed as a self-normalized U-statistic that estimates a normalized version of the generalized energy distance. Consequently, a candidate for the change-point location can be estimated as:
\begin{align}\label{single cp candidate location} 
	\widehat{\nu}\;:=\;\argmax_{4\leq k\leq n-4}\,\frac{k(n-k) }{n^2}\,T_n(k)\,.
\end{align}

\subsubsection{A componentwise monotone-invariant version}\label{sec:mono_invariant}
A key desideratum for nonparametric change-point analysis is invariance to transformations that preserve the relative ordering of observations. To address this, we consider a simple \emph{componentwise} rank/CDF transformation, which yields an estimator that is invariant to strict componentwise monotone transformations and alleviates heavy-tail concerns. \rev{Because the transformation uses pooled empirical mid-ranks, the pseudo-observations introduced below are mutually dependent even when the original observations are independent. Thus, the monotone-invariant version should be viewed as a heuristic extension of the original method; its practical value is assessed through the simulation studies rather than through a complete HDMSS limiting theory.}

\begin{revblock}
For each coordinate $j\in\{1,\dots,p\}$, define the pooled empirical mid-distribution function
\[
\widehat F^{\mathrm{mid}}_j(x)
:=\frac{1}{n}\sum_{t=1}^n\left\{\mathbbm{1}(X_{t,j}<x)
+\frac{1}{2}\mathbbm{1}(X_{t,j}=x)\right\}.
\]
If $R_{t,j}$ is the average rank of $X_{t,j}$ among $X_{1,j},\dots,X_{n,j}$, the pseudo-observations are
\[
U_{t,j}:=\widehat F^{\mathrm{mid}}_j(X_{t,j})
=\frac{R_{t,j}-1/2}{n},\qquad t=1,\dots,n.
\]
Thus the manuscript and implementation use the same mid-rank convention, including in the presence of ties. The heuristic HDMSS discussion below assumes continuous marginals. Strict componentwise increasing transformations preserve both order and ties, so Proposition~\ref{prop:componentwise_invariance} remains valid under this convention.
\end{revblock}
Let $\mathbf U_t := (U_{t,1},\dots,U_{t,p})^\top\in[0,1]^p$. We define the monotone-invariant statistic and the associated single change-point locator by applying the same procedure to $\{\mathbf U_t\}_{t=1}^n$:
\begin{align}\label{our_statistic_MI}
	M_n^{\mathrm{MI}} &:= \max_{4\leq k\leq n-4}\,\frac{k(n-k)}{n^2}\,T(\mathbf{U}_{1:k},\mathbf{U}_{k+1:n}),\\
	\widehat{\nu}^{\mathrm{MI}} &:= \argmax_{4\leq k\leq n-4}\,\frac{k(n-k)}{n^2}\,T(\mathbf{U}_{1:k},\mathbf{U}_{k+1:n}). \label{single_cp_MI}
\end{align}
In practice, $T(\cdot,\cdot)$ can be taken to be the same self-normalized statistic used in \eqref{our_statistic_revised}.

\begin{proposition}\label{prop:componentwise_invariance}
Let $g=(g_1,\dots,g_p)$ where each $g_j:\mathbb{R}\to\mathbb{R}$ is strictly increasing, and define the transformed data
$Y_t := (g_1(X_{t,1}),\dots,g_p(X_{t,p}))^\top$. Then the pseudo-observations computed from $\{Y_t\}$ equal those computed from $\{X_t\}$, and consequently
\[
M_n^{\mathrm{MI}}(\mathbf{Y}_{1:n}) = M_n^{\mathrm{MI}}(\mathbf{X}_{1:n})
\quad\text{and}\quad
\widehat{\nu}^{\mathrm{MI}}(\mathbf{Y}_{1:n}) = \widehat{\nu}^{\mathrm{MI}}(\mathbf{X}_{1:n}).
\]
Here $M_n^{\mathrm{MI}}(\mathbf{X}_{1:n})$ and $\widehat{\nu}^{\mathrm{MI}}(\mathbf{X}_{1:n})$ denote the test statistic and estimator computed using the pseudo-observations derived from the original sequence $\mathbf{X}_{1:n}$, while $M_n^{\mathrm{MI}}(\mathbf{Y}_{1:n})$ and $\widehat{\nu}^{\mathrm{MI}}(\mathbf{Y}_{1:n})$ denote those computed from the transformed sequence $\mathbf{Y}_{1:n}$.

\end{proposition}

\begin{remark}
{\rm 
The transformation \eqref{our_statistic_MI} is the multivariate/high-dimensional analogue of classical rank-based invariance in univariate change-point analysis: it is invariant to strict componentwise monotone re-parameterizations and reduces sensitivity to marginal tail behavior. While more sophisticated multivariate rank notions (e.g., depth-, spatial-, or transport-induced ranks) are possible, the componentwise version is simple, computationally cheap, and invariant under strict componentwise monotone transformations.
}
\end{remark}

This approach is designed to detect a broader array of changes in high-dimensional distributions, not merely shifts in mean or covariance. Importantly, our proposed statistic $M_n$ possesses a pivotal limiting distribution under the null hypothesis and demonstrates greater power than the test of \cite{MJ} for a wider range of structural breaks. \rev{Here and in the formal asymptotic results below, $M_n$ refers to the statistic computed from the original independent observations. For $M_n^{\mathrm{MI}}$, the same critical values are used as a heuristic calibration, motivated by the double-centering structure and validated empirically in Section~\ref{sec:num}.}

\begin{remark}
{\rm 
    Our theoretical results for $M_n$ and $\widehat{\nu}$ rely on the independence of the raw observations across time, which allows the partial sum process to be analyzed via standard martingale central limit theorems. In contrast, the theoretical analysis for the componentwise monotone-invariant version ($M_n^{\mathrm{MI}}$) is substantially more involved. Because the empirical pseudo-observations are computed using the marginal empirical CDFs evaluated over the \emph{entire pooled sample}, the transformed data points become inherently coupled. Structurally, this substitution elevates the underlying test statistic from a standard degree-2 $U$-statistic to a higher-order generalized $U$-statistic, thereby breaking the standard martingale structure. 
    
    While substituting a $\sqrt{n}$-consistent estimator into a $U$-statistic typically introduces a leading-order perturbation that alters the limiting distribution, we conjecture that the strict double-centering constraints of the generalized energy distance act to difference out the main effects of this marginal estimation error. This heuristic is strongly supported by our extensive numerical studies, which show that the asymptotic critical values derived for $M_n$ provide excellent size control and minimal power loss when applied to $M_n^{\mathrm{MI}}$. However, formally establishing this uniform convergence in the high-dimensional medium sample size (HDMSS) framework---where $p$ diverges---is technically highly non-trivial. Therefore, we defer a fully rigorous theoretical justification to future work, while providing a more detailed heuristic discussion in Section \ref{sec:theory_MI} of the Supplementary Materials.
}
\end{remark}

\subsection{Theoretical framework: Embedding in Hilbert space}\label{sec:theoretical_framework}
The theoretical foundation of our method incorporates the cumulative sum (CUSUM) process, constructed within an embedded Hilbert space. This construction leverages Proposition \ref{equiv_revised}, a known result concerning the characterization of spaces of negative type \citep[see Section 3 in][]{lyons}.

\begin{proposition}\label{equiv_revised}
A metric space $(\mathcal{X},\rho)$ has a negative type if and only if there is a Hilbert space $(\mathcal{H},\langle\cdot,\cdot\rangle_{\mathcal{H}})$ and an embedding map $\phi: \mathcal{X} \to \mathcal{H}$ such that $\rho(x,x') = \Vert \phi(x) - \phi(x')\Vert_{\mathcal{H}}^2$ \, for all $x,x' \in \mathcal{X}$, where $\Vert \cdot \Vert_{\mathcal{H}}=\langle\cdot,\cdot\rangle_{\mathcal{H}}^{1/2}$ is the norm associated with $\mathcal{H}$.
\end{proposition}

Given that $(\mathbb{R}^p,\gamma)$ possesses a strong negative type, Proposition \ref{equiv_revised} guarantees the existence of an embedding map $\phi: \mathbb{R}^p \to \mathcal{H}$ for some Hilbert space $\mathcal{H}$, such that $\gamma(x,x')=\Vert \phi(x) - \phi(x') \Vert_{\mathcal{H}}^2$, for all $x, x' \in \mathbb{R}^p$. From this, we derive:
\begin{align}\label{embed eqn 2} 
	\langle\phi(x),\phi(x')\rangle_{\mathcal{H}}=2^{-1}(\langle\phi(x),\phi(x)\rangle_{\mathcal{H}}+\langle\phi(x'),\phi(x')\rangle_{\mathcal{H}}-\gamma(x,x')).
\end{align}
Assume that $X,X'\stackrel{\mathrm{i.i.d.}}{\sim}P$ and $Y,Y'\stackrel{\mathrm{i.i.d.}}{\sim}Q$, and that $(X,X')$ is independent of $(Y,Y')$. If $\mathbb{E}\|\phi(X)\|^2_{\mathcal{H}}+\mathbb{E}\|\phi(Y)\|^2_{\mathcal{H}}<\infty$, then
\begin{align*}
E_{\gamma}(X,Y) = & 2\mathbb{E}[\langle\phi(X), \phi(X')\rangle_{\mathcal{H}} +\langle\phi(Y), \phi(Y')\rangle_{\mathcal{H}} -2\langle\phi(X), \phi(Y)\rangle_{\mathcal{H}}]
\\= & 2\,\big\|\mathbb{E}\{\phi(X)\}-\mathbb{E}\{\phi(Y)\}\big\|^2_{\mathcal{H}}.
\end{align*}
In particular, when $\phi$ is induced by a reproducing kernel $\mathcal K$ and $\gamma$ is the corresponding kernel-induced distance, $E_{\gamma}$ coincides (up to a constant) with the squared maximum mean discrepancy between $P$ and $Q$. In this setting, the strong negative type property is closely related to the kernel being \emph{characteristic} (i.e., the mean embedding $P\mapsto \mathbb{E}\{\phi(X)\}$ is injective on the relevant class of measures). The generalized energy distance is twice the squared norm of the difference between the means of the embedded data. Therefore, detecting structural breaks for distributional changes can be viewed as detecting changes in the means in the embedded space, when the homogeneity between two distributions is characterized by the generalized energy distance.

The CUSUM process for detecting changes in means within this embedded space is defined as:
\begin{align*}
	S_k := \frac{1}{\sqrt{n}} \sum_{t=1}^k\left(\phi(X_t) - \frac{1}{n} \sum_{j=1}^n \phi(X_j)\right)
\end{align*} 
for $1\leq k \leq n$. Some fundamental properties of $S_k$ are presented in Lemma \ref{lemma_cusum}.

\begin{lemma}\label{lemma_cusum} 
	The cumulative sum process $S_k$ can be expressed as \vspace{-0.1in} $$S_k = \frac{k(n-k)}{n^{3/2}} \;\left(\frac{1}{k} \sum_{t=1}^k \phi(X_t) \,-\, \frac{1}{n-k} \sum_{t=k+1}^n \phi(X_t) \right)$$ for $1\leq k \leq n$. Further, the squared norm of $S_k$ is given by
	\small
	\begin{align*}
		\Vert S_k \Vert^2_{\mathcal{H}} \;=\;
\frac{k^2\,(n-k)^2}{2n^3}\,\Bigg(&\frac{2}{k(n-k)} \sum_{t=1}^k \sum_{t'=k+1}^n \gamma(X_t,X_{t'}) \;-\; \frac{1}{k^2} \sum_{t,t'=1}^k \gamma(X_t,X_{t'})
  \\ &- \;\frac{1}{(n-k)^2} \sum_{t,t'=k+1}^n \gamma(X_t,X_{t'})\Bigg)\,.
\end{align*}
\normalsize
\end{lemma}
The proof of Lemma \ref{lemma_cusum} is provided in Section \ref{tech appx} of the supplementary material. The first part of the lemma shows that when there is a single change-point at $\nu$, the quantity $n^{3/2}S_{\nu}/\{\nu(n-\nu)\}$ is an unbiased estimator of $\mathbb{E}[\phi(X)] - \mathbb{E}[\phi(Y)]$, where $X \sim F_1$, $Y \sim F_2$, and $X \perp \!\!\! \perp Y$. Accordingly, $2n^3\|S_{\nu}\|_{\mathcal{H}}^2/\{\nu^2(n-\nu)^2\}$, which is a V-statistic, serves as a natural, though biased, plug-in estimator of $E_{\gamma}(X, Y)$. In high-dimensional settings, it is well-known that the bias of V-statistics is often non-negligible. While $\|S_k\|_{\mathcal{H}}$ is still expected to attain its maximum at the true change-point $\nu$ after suitable normalization, the finite sample performance of the estimator based on the V-statistic is inferior to that of a U-statistic-based estimator, which mitigates bias by excluding diagonal terms. For this reason, we construct our test statistic $M_n$ using the U-statistic estimator $\widehat{E}_{\gamma}(\mathbf{X}_{1:k}, \mathbf{X}_{k+1:n})$.

For the practical application of $M_n$ in hypothesis testing—that is, to assess the statistical significance of an estimated change-point location $\widehat{\nu}$—it is necessary to determine its null distribution. A primary challenge is the derivation of the limiting distribution of $M_n$ under the null hypothesis. A key theoretical contribution of this work is the rigorous derivation of the asymptotic null distribution of $M_n$ as both the sample size $n$ and the dimension $p$ tend to infinity. This derivation necessitates a uniform weak convergence result for the stochastic process $\{T_n(\lfloor nr \rfloor): r\in[0,1]\}$, as pointwise weak convergence alone is insufficient. Finally, while determining the limiting distribution is a significant theoretical challenge, the practical implementation of the test statistic $M_n$ also involves computational considerations. A detailed discussion of the computational cost associated with our proposed test statistic is provided in Section \ref{sec:computational}.

\subsection{Assumptions}\label{sec:assumptions}
Before presenting the main theoretical results, we provide some technical assumptions for the change-point model described in (\ref{alt:mult}). Assume $\{Z_i\}_{i = 1}^{N+1}$  is a sequence of independent random vectors such $Z_i \sim F_i$ for all $i = 1,2,..,N+1$. We define $\tau^{(i,j)} = \sqrt{\mathbb{E}(\gamma^2(Z_i,Z_j))}$ for $i \neq j$ and $\tau^{(i,i)} = \sqrt{\mathbb{E}(\gamma^2(Z_i,Z_i'))}$, where $Z_i'$ is an independent copy of $Z_i.$ Recall that for a random vector $X \in \mathbb{R}^p$, $X_{\mathcal{S}_k}=(X_{j}:j\in \mathcal{S}_k)\in\mathbb{R}^{d_k}$.

\begin{assumption}\label{ass0_new}
    There exist constants $0 \leq c \leq C < \infty$ such that uniformly over $p$, for any $X \sim F_i$ and $Y\sim F_j$ that are independent, \[c \leq \inf_{1 \leq k \leq g}\mathbb{E}\rho_k(X_{\mathcal{S}_k}, Y_{\mathcal{S}_k}) \leq \sup_{1 \leq k \leq g}\mathbb{E}\rho_k(X_{\mathcal{S}_k}, Y_{\mathcal{S}_k}) \leq C,\]
    for all $i,j = 1,2,\dots,N+1.$
\end{assumption}

When \( g = p \), \( d_i = 1 \) for \( 1 \leq i \leq g \), and \( \rho_i \) represents the Euclidean distance on \( \mathbb{R} \), Assumption \ref{ass0_new} simplifies to 
\[
c \leq \mathbb{E}|x_i - x_i'| \leq C
\]
for all \( 1 \leq i \leq p \), where \( x_i \) and $x_i'$ denote the \( i \)-th components of \( X \) and $X'$ respectively. It is important to note that for each \( i \), since \( \mathbb{E}|x_i - x_i'| \leq 2 \mathbb{E}|x_i| \), the upper bound holds if \( \mathbb{E}|x_i| \leq C/2 \), i.e., the components of $X$ have uniformly bounded first moments. The lower bound indicates that \( x_i \) is not equal to a constant with probability one, which is a mild requirement. This condition is primarily intended to rule out certain pathological cases and can be relaxed. Under Assumption \ref{ass0_new}, it is easy to see that $\tau^{(i,j)} \asymp p^{1/2}$. The following proposition presents an expansion formula for the distance metric $\gamma$ when the dimension is high, which plays a key role in our theoretical analysis.

\begin{proposition}\label{Prop 4.1 in CZ}
Under Assumption \ref{ass0_new}, for any $X \sim F_i$ and $Y\sim F_j$ that are independent, we have $$\frac{\gamma(X,Y)}{\tau^{(i,j)}} = 1 + \frac{1}{2} L(X,Y) + R(X,Y)\,,$$ where $L(X,Y)= \frac{\gamma^2(X,Y) - [\tau^{(i,j)}]^2}{[\tau^{(i,j)}]^2}$ is the leading term and $R(X,Y)$ is the remainder term. In addition, if $L(X,Y)$  is an  $o_p(1)$ random variable as $p\to\infty$, then  $R(X,Y) = O_p\left(L^2(X,Y)\right)$.
\end{proposition}  

The proposition above suggests that \( L(X,Y) \) is the leading term of $\gamma(X,Y)/\tau^{(i,j)}$. We introduce another technical quantity that is closely related to \( L(X,Y) \) and plays an important role in the theoretical justification. For any two random vectors \( X, Y \in \mathbb{R}^p \), we define
\begin{align*}
	H(X, Y)=\frac{1}{\sqrt{\mathbb{E}[\gamma^2(X,Y)]}}\sum_{k=1}^g&\left(\rho_k(X_{\mathcal{S}_k},Y_{\mathcal{S}_k})-\mathbb{E}\left[\rho_k(X_{\mathcal{S}_k},Y_{\mathcal{S}_k})\big\vert X_{\mathcal{S}_k}\right]\right.\\&\qquad\left. - \mathbb{E}\left[\rho_k(X_{\mathcal{S}_k},Y_{\mathcal{S}_k})\big\vert Y_{\mathcal{S}_k}\right]+ \mathbb{E}\left[\rho_k(X_{\mathcal{S}_k},Y_{\mathcal{S}_k})\right]\right),
\end{align*}
where the summand can be viewed as the double-centered distance between $X_{\mathcal{S}_k}$ and $Y_{\mathcal{S}_k}$.

\begin{assumption}\label{ass1_new}
	As $n, p \to \infty$, for any $X \sim F_i$ and $Y\sim F_j$ that are independent,
	\begin{align*}
   \frac{1}{n}\frac{\mathbb{E}[H^4(X,Y)]}{(\mathbb{E}[H(X,Y)^2])^2} = o(1),\quad \frac{\mathbb{E}[H(X,Y)H(X',Y)H(X,Y')H(X',Y')]}{(\mathbb{E}[H(X,Y)^2])^2} = o(1),
	\end{align*}
for all $i,j = 1,2,\dots,N+1$.
\end{assumption}

Assumption \ref{ass1_new} imposes some moment restrictions for $F_1,\dots,F_{N+1}$, similar to those in \cite{zys2018}; see Section 2.2 therein for a more detailed discussion.

\begin{assumption}\label{ass2_new}
	As $n, p \to \infty$, for any $X,Y$ that are independent such that $X \sim F_i$ and $Y \sim F_j$, $$\frac{n^4 [\tau^{(i,j)}]^4 \mathbb{E} \left[ R^4(X,Y)\right]}{\left( \mathbb{E} \left[ H^2(X,Y)\right] \right)^2} =  o(1),$$
    for all $i,j = 1,2,\dots,N+1.$
\end{assumption}

We refer the readers to Remark 4.1 in \cite{us}, which illustrates some sufficient conditions under which  $\alpha_p = O(p^{-1/2})$ and consequently $\tau^{(i,j)}\alpha_p^2  = o(1)$ holds, as $\tau^{(i,j)}\asymp p^{1/2}$. In similar lines of Remark D.1 in the Supplementary Materials of their paper, it can be argued that\, $\mathbb{E}\left[ R^4(X,Y)\right] = O\left(p^{-4}\right)$. Furthermore, with a mild assumption that $\sigma^2 := \lim_{p \to \infty}\mathbb{E} \left[H^2(X, Y)\right]$, we can show that $\mathbb{E}\left[ H^2(X,Y)\right] \asymp 1$. Combining all these results, it is not hard to verify that $n^4 [\tau^{(i,j)}]^4\mathbb{E}\left[ R^4(X,Y)\right]/\left( \mathbb{E}\left[ H^2(X,Y)\right] \right)^2= o(1)$ holds provided that $n = o(p^{1/2})$.

\subsection{Asymptotic analysis under the null}\label{sec:methods:test stat}
The subsequent theorem establishes the limiting process of $T_n(k)$ under the null hypothesis, which is pivotal for deriving the limiting null distribution of $M_n$.

\begin{theorem}\label{theorem1}
	Under the null hypothesis, Assumptions \ref{ass0_new}, \ref{ass1_new} and \ref{ass2_new}, as $n,p \to \infty$,
	\begin{align*}
		&\left\{ \frac{\lfloor nr \rfloor  (n-\lfloor nr \rfloor) }{n^2}\, T_{n}(\lfloor nr \rfloor) \right\}_{r \in [0,1]} \;\rightsquigarrow \; G_0 \qquad \textrm{in} \;\; \mathcal{L}^{\infty}\left( [0,1] \right)\,,
	\end{align*}
	where $G_0(r) :=  r(1-r)\,Q(0,1) - (1-r)\,Q(0,r) - r\,Q(r,1)$\, for \,$r \in (0,1)$ and zero otherwise. Here, $Q$ is a centered Gaussian process with the covariance function given by 
	\begin{align}\label{eq-cov}
		\cov\,\big( Q(a_1,b_1)\,,\, Q(a_2,b_2) \big) \;=\; \big(b_1 \land b_2 \,-\, a_1 \lor a_2\big)^2 \,\, \mathbbm{1}\big(b_1 \land b_2 > a_1 \lor a_2 \big)\,.
	\end{align}
	In particular, $\var\,\big(Q(a,b)\big) = (b-a)^2\, \mathbbm{1}(b>a)$.
\end{theorem}

The proof of this theorem is non-trivial, requiring the establishment of finite-dimensional weak convergence and the stochastic equicontinuity of the random process $\{ n^{-2}\lfloor nr \rfloor  (n-\lfloor nr \rfloor)\,T_{n}(\lfloor nr \rfloor)\}_{r \in [0,1]}$; see Theorem 10.2 in \cite{pollard}. Due to its technical complexity, we relegate the proof to the Supplementary Materials. It is worth mentioning that the limiting Gaussian process $Q$ coincides with the one derived in Theorem 3.4 of \cite{WVS}.

\begin{remark}\label{rem_theorem1}
	{\rm 
	Theorem B.1 in the Supplementary Materials of \cite{us} proves that for fixed $r \in (0,1)$, $T_{n}(\lfloor nr \rfloor) \overset{d}{\rightarrow} N(0,1)$ as $n,p \to \infty$, which implies that $n^{-2}\lfloor nr \rfloor  (n-\lfloor nr \rfloor)\,T_{n}(\lfloor nr \rfloor) \overset{d}{\rightarrow} N(0,r^2(1-r)^2)$ as $n,p \to \infty$. By Theorem \ref{theorem1}, for a fixed $r \in (0,1)$, $G_0(r)$ has a Gaussian distribution with zero mean and $\var \,\left(G_0(r)\right) = r^2(1-r)^2$. This illustrates that the uniform weak convergence result established in Theorem \ref{theorem1} generalizes the pointwise weak convergence result in \cite{us}.
	}
\end{remark}

As a consequence of Theorem \ref{theorem1}, we derive the limiting null distribution of $M_n$.

\begin{theorem}\label{theorem2}
	Under the null hypothesis, Assumptions \ref{ass0_new}, \ref{ass1_new} and \ref{ass2_new}, as\, $n,p \to \infty$, 
	$$M_n \,\overset{d}{\rightarrow} \, \sup_{r \in (0,1)} \,  G_0(r).$$
\end{theorem}
Theorem \ref{theorem2} follows from Theorem \ref{theorem1} and the continuous mapping theorem. 
%
Note that the limiting null distribution is pivotal in nature. Consequently, the quantiles of the limiting distribution can be approximated using a large number of Monte Carlo simulations.

\begin{remark}\label{remark quantiles}
	{\rm 
		Table \ref{table:sim_quantiles} below provides the simulated quantiles of the limiting null distribution of $M_n$ based on 10,000 Monte Carlo replications with $\{X_t\}_{t=1}^n$ generated from the $N(0, \mathbf{I}_p)$ distribution with $n=1000$ and $p=1000$. 
		\begin{table}[!ht]\small
			\centering
			\caption{Simulated quantiles of the limiting distribution of $M_n$. Here $Q_{\alpha}$ denotes the $100(1-\alpha)$th quantile of this distribution.}
			\label{table:sim_quantiles}
			\begin{tabular}{cccc}
				\toprule
				$100(1-\alpha)\%$ & $90\%$ & $95\%$ & $99\%$ \\
				\hline
				$Q_{\alpha}$ & 0.568 & 0.642 & 0.812  \\
				\bottomrule
			\end{tabular}
		\end{table}
	}
\end{remark}

The change-point detection method based on the limiting null distribution of $M_n$ is computationally efficient but may sometimes lead to a slight increase in Type I error. For numerical evidence, please refer to Table \ref{tab1_asymp_quantiles1}. Alternatively, one can use a permutation procedure to approximate the quantiles of the distribution of $M_n$ for more accurate results. 
\rev{Algorithm \ref{alg1} presents the pseudocode of the permutation procedure used to test $H_0$ against the single change-point alternative. We denote the number of random permutations by $B_{\mathrm{perm}}$ to distinguish it from the pair-subsampling parameter used later for incomplete $U$-statistics. For the monotone-invariant statistic, the ranks may be computed once from the pooled sample and then permuted as rows, because row permutations do not change the pooled marginal ranks.}

\begin{algorithm}[!ht]
	\caption{Single change-point detection via permutation}\label{alg1}
	\begin{algorithmic}[1]
	\State \textbf{Input:} $\mathbb{R}^p$-valued observations $\{X_1, \dots, X_n\}$; level of significance $\alpha \in(0,1)$; \rev{number of permutation replicates $B_{\mathrm{perm}}$.}
	\State Compute the observed test statistic $M_n$ and the candidate change-point $\widehat{\nu}$.
		\For {\rev{$j=1,2,\dots, B_{\mathrm{perm}}$}}
			\State \rev{Generate an independent uniform random permutation of the indices $\{1, \dots, n\}$ to obtain permuted data $\{X_1^*, \dots, X_n^*\}$.}
			\State Compute the statistic for the permuted data, denote it $M_n^{(j)}$.
		\EndFor
		\State \rev{Compute $p_{\mathrm{perm}}=\{1+\sum_{j=1}^{B_{\mathrm{perm}}}\mathbbm{1}(M_n^{(j)}\ge M_n)\}/(B_{\mathrm{perm}}+1)$.}
	\If{\rev{$p_{\mathrm{perm}}\le \alpha$}}
	 \State Reject $H_0$ at level $\alpha$.
	 \State Return $\widehat{\nu}$ as the estimated change-point.
	 \EndIf
	 \end{algorithmic}
\end{algorithm}

\begin{theorem}\label{thm:perm_validity}
Let $\mathbf{X} = \{X_1, \dots, X_n\}$ be a sequence of independent random vectors in $\mathbb{R}^p$. Let $\mathcal{T}_n(\mathbf{X})$ be any test statistic.
\rev{Consider the permutation testing procedure (Algorithm~\ref{alg1}) which rejects the null hypothesis if the Monte Carlo permutation $p$-value satisfies $p_{\mathrm{perm}}\le \alpha$, where $p_{\mathrm{perm}}$ is computed from $B_{\mathrm{perm}}$ independently and uniformly sampled random permutations and the original statistic.}
Under the null hypothesis $H_0: X_1, \dots, X_n \overset{i.i.d.}{\sim} F$, the test controls the Type I error at level $\alpha$:
\[
P_{H_0}(\text{Reject } H_0) \le \alpha.
\]
\rev{This holds for any sample size $n$, dimension $p$, and number of replicates $B_{\mathrm{perm}}$, regardless of the underlying distribution $F$.}
\end{theorem}

Theorem~\ref{thm:perm_validity} immediately validates Algorithm~\ref{alg1} for both the standard statistic $M_n$ defined in \eqref{our_statistic_revised} and the monotone-invariant statistic $M_n^{\mathrm{MI}}$ defined in \eqref{our_statistic_MI}, as both are well-defined functionals of the sample $\mathbf{X}$. \rev{This statement concerns only the finite-sample permutation level guarantee; it does not assert the unproved pivotal HDMSS limit for $M_n^{\mathrm{MI}}$.}

\rev{\begin{remark}
{\rm 
The number $B_{\mathrm{perm}}$ controls the granularity of the Monte Carlo permutation test. Since the rank of the observed statistic among the $B_{\mathrm{perm}}+1$ values is discrete, the actual rejection probability is bounded by
\[
\frac{\lfloor \alpha(B_{\mathrm{perm}}+1)\rfloor}{B_{\mathrm{perm}}+1}\le \alpha,
\]
The inequality is strict whenever $\alpha(B_{\mathrm{perm}}+1)$ is not an integer, regardless of whether $B_{\mathrm{perm}}$ is small or large. Relatedly, the smallest attainable permutation $p$-value is $1/(B_{\mathrm{perm}}+1)$. The level gap is less than $1/(B_{\mathrm{perm}}+1)$, so a small number of permutations may make the discreteness practically important and reduce power, while a larger $B_{\mathrm{perm}}$ gives finer calibration at additional computational cost.
}
\end{remark}}

\subsection{Power analysis}
In this section, we establish the consistency of our testing procedure under the alternative hypothesis of a single change-point located at \(\nu\). Specifically, we assume the marginal distribution of the sequence follows \(F_1\) for the first \(\nu\) observations (\(X_1, \dots, X_{\nu}\)) and \(F_2\) for the remaining observations (\(X_{\nu + 1}, \dots, X_n\)), with \(F_1 \neq F_2\).

To simplify notation, let \(\tau_1 := \tau^{(1,1)}\), \(\tau_2 := \tau^{(2,2)}\), and \(\tau_3 := \tau^{(1,2)}\). Additionally, we define the variance components \(V_1 = \mathbb{E}[H(X_1, X_1')^2]\), \(V_2 = \mathbb{E}[H(X_n, X_n')^2]\), and \(V_3 = \mathbb{E}[H(X_1, X_n)^2]\). The following theorem establishes the behavior of the test statistic under the alternative hypothesis.

\begin{theorem}\label{alt:fix_alternative}
Define $V := \max\{V_1, V_2, V_3, n\Gamma_1, n\Gamma_2\}$, where 
\begin{align*}
    \Gamma_1 &= \operatorname{var}\Big(\tau_3\mathbb{E}[L(X_1,X_n)\mid X_1] - \tau_1\mathbb{E}[L(X_1,X_1')\mid X_1]\Big), \\
    \Gamma_2 &= \operatorname{var}\Big(\tau_3\mathbb{E}[L(X_1,X_n)\mid X_n] - \tau_2\mathbb{E}[L(X_n,X_n')\mid X_n]\Big).
\end{align*}
If Assumptions \ref{ass0_new}, \ref{ass1_new}, and \ref{ass2_new} hold, then under the single change-point alternative, as $n, p \to \infty$,
\[
E_{\gamma}(X_1, X_n) = 2\tau_3 - \tau_1 - \tau_2 + o(\sqrt{V}/n).
\]
Furthermore, if $n(2\tau_3 - \tau_1 - \tau_2)/\sqrt{V}\rightarrow \infty$, then $M_n \overset{P}{\rightarrow} \infty$. 
\end{theorem}

\begin{remark}\label{alt:illustration1}
{\rm 
Theorem \ref{alt:fix_alternative} demonstrates that under the alternative hypothesis, the expected generalized energy distance between the pre-change and post-change distributions is dominated by the term \(2\tau_3 - \tau_1 - \tau_2\). We can interpret \(2\tau_3 - \tau_1 - \tau_2\) as the ``energy'' of the change, while \(n(2\tau_3 - \tau_1 - \tau_2)/\sqrt{V}\) represents the signal-to-noise ratio. When this ratio diverges, the test statistic diverges in probability, thereby establishing the consistency of the test.

To provide further insight, consider the case where $\gamma$ is the standard Euclidean distance, i.e., $\gamma(X,X') = \|X - X'\|$. Simple calculations yield $\tau_1^2 = 2\operatorname{tr}(\Sigma_1)$, $\tau_2^2 = 2\operatorname{tr}(\Sigma_2)$, and $\tau_3^2 = \operatorname{tr}(\Sigma_1) + \operatorname{tr}(\Sigma_2) +  \|\mu_1 - \mu_2\|^2$, where $(\mu_1, \Sigma_1)$ and $(\mu_2, \Sigma_2)$ are the mean and covariance matrix of $F_1$ and $F_2$, respectively. Assumption \ref{ass0_new} ensures that $\tau_1,\tau_2, \tau_3 \asymp \sqrt{p}$, which implies $2\tau_3 - \tau_1 - \tau_2 = O(\sqrt{p})$. Specifically, direct calculation shows that 
\[
2\tau_3 - \tau_1 - \tau_2 \asymp \frac{\|\mu_1 - \mu_2\|^2}{\sqrt{p}} + \frac{(\operatorname{tr}(\Sigma_1) - \operatorname{tr}(\Sigma_2))^2}{p^{3/2}}.
\]
Under Assumptions \ref{ass0_new}-\ref{ass2_new}, if we further assume that the components of \(X_i\) are independent and have finite fourth moments, then the condition \(\sqrt{n}(2\tau_3 - \tau_1 - \tau_2) \rightarrow \infty\) is equivalent to requiring either \(\sqrt{n}\|\mu_1 - \mu_2\|^2/\sqrt{p} \rightarrow \infty\) or \(\sqrt{n}(\operatorname{tr}(\Sigma_1) - \operatorname{tr}(\Sigma_2))^2/p^{3/2} \rightarrow \infty\). In such cases, \(M_{n} \overset{P}{\rightarrow} \infty\). We present the proof of this result in the Supplementary Materials.
}
\end{remark}

\subsection{Location estimator and its consistency}
In this section, we establish the consistency of the change-point location estimator. While the practical algorithm uses the studentized statistic $T_n(k)$ to define $\widehat{\nu}$ (see \eqref{single cp candidate location}), the theoretical analysis of the studentized version under the alternative is technically involved due to the complex behavior of the pooled variance estimator $\widehat{S}(\mathbf{X}_{1:k},\mathbf{X}_{k+1:n})$. To facilitate a clear theoretical exposition without unnecessary technical complications, we focus our analysis on the un-studentized estimator $\widehat{\nu}^*$, which corresponds to the numerator of our test statistic:
\begin{equation}\label{def-nu}
\widehat{\nu}^* = \argmax_{k}\frac{k(n-k)}{n^2}\widehat{E}_{\gamma}(\mathbf{X}_{1:k},\mathbf{X}_{k+1:n}).
\end{equation}
We denote $\widehat{\zeta}^* = \widehat{\nu}^*/n$ as the estimator of the change-point proportion. The following theorem establishes the consistency rate of this estimator.

\begin{theorem}\label{alt:consistency} 
    Under the single change-point alternative with structural break located at $\nu$, suppose $\zeta=\lim_{n\rightarrow+\infty} \nu/n \in (0,1)$. Define the signal strength sequence
    $$a_n := \frac{n(2\tau_3-\tau_1-\tau_2)}{\sqrt{V'}} \rightarrow \infty,$$
    where $V' = \max\{\log(n)V_1, \log(n)V_2, \log(n)V_3, n\Gamma_1, n\Gamma_2\}$. 
    If Assumptions \ref{ass0_new}, \ref{ass1_new}, and \ref{ass2_new} hold, then as $n, p \rightarrow \infty$,
    $$|\widehat{\zeta}^* - \zeta| = O_p(a_n^{-1}).$$ 
\end{theorem}

\rev{\begin{remark}
{\rm 
The maximizer in \eqref{def-nu} is a mathematical location estimator defined for any observed sequence, but in the testing procedure a change-point location is reported only when the corresponding global test rejects the null hypothesis. Theorems~\ref{alt:fix_alternative} and \ref{alt:consistency} describe two complementary pieces of this behavior. Under the alternative, the signal-to-noise ratio makes the scan statistic diverge, so the test using a fixed asymptotic critical value rejects with probability tending to one. If the reported location is the unstudentized locator $\widehat\nu^*$, its conditional-on-reporting rate is the rate in Theorem~\ref{alt:consistency}. The practical locator $\widehat\nu$ in \eqref{single cp candidate location} maximizes the studentized statistic, and the present theorem does not establish the same rate for that different estimator. For permutation calibration, let $c_{n,\mathrm{perm}}$ denote the data-dependent conditional critical value. A sufficient additional condition for rejection with probability tending to one is $c_{n,\mathrm{perm}}=o_p(M_n)$. Theorem~\ref{alt:fix_alternative} does not by itself establish this condition under the alternative, and we do not claim a general permutation-threshold consistency result here.
}
\end{remark}}

Multiple change-point estimation is developed in Section~\ref{sec:methods:multiple cp} using a Seeded NOT recursion driven by our single-change-point locator, with a consistency guarantee in Theorem~\ref{mult_consistency_SBS}.

\subsection{Computational aspects of the procedure}\label{sec:computational}
\subsubsection{Recursive updates for scanning over split points}\label{sec:computational:recursive}

To implement the proposed test efficiently, recall that our statistic takes the form:
$$ M_n := \max_{4 \leq k \leq n-4} \frac{k(n-k)}{n^2} T_n(k), $$
where 
$$ T_n(k) = T(\mathbf{X}_{1:k}, \mathbf{X}_{k+1:n}) = \frac{\widehat{E}_{\gamma}(\mathbf{X}_{1:k}, \mathbf{X}_{k+1:n})}{a_{k,n-k} \widehat{S}(\mathbf{X}_{1:k}, \mathbf{X}_{k+1:n})}. $$
A naive implementation computing $\widehat{E}_{\gamma}$ and $\widehat{S}$ from scratch for each $k$ would result in an overall complexity of $O(n^3)$. However, the core terms can be computed sequentially for $k = 4, 5, \ldots, n-4$ using quantities derived from the previous step $k-1$.

We first compute the pairwise distance matrix $\mathbf{D} = (\gamma(X_i,X_j))_{i,j=1}^n \in \mathbb{R}^{n \times n}$ with $\gamma(X_i,X_i) = 0$. Let $\mathbf{D}_{S_1,S_2}$ denote the submatrix of $\mathbf{D}$ consisting of rows indexed by $S_1$ and columns indexed by $S_2$. The U-statistic numerator $\widehat{E}_{\gamma}(\mathbf{X}_{1:k}, \mathbf{X}_{k+1:n})$ can be expressed as:
\begin{align*}
\widehat{E}_{\gamma}(\mathbf{X}_{1:k}, \mathbf{X}_{k+1:n}) = \frac{2 \mathbf{1}^\top_k \mathbf{D}_{1:k,k+1:n} \mathbf{1}_{n-k}}{k(n-k)} - \frac{\mathbf{1}^\top_k \mathbf{D}_{1:k,1:k} \mathbf{1}_k}{k(k-1)} - \frac{\mathbf{1}^\top_{n-k} \mathbf{D}_{k+1:n,k+1:n} \mathbf{1}_{n-k}}{(n-k)(n-k-1)}.
\end{align*}
Each term allows for recursive updates. Define $A_{1,k} = \mathbf{1}^\top_k \mathbf{D}_{1:k,k+1:n} \mathbf{1}_{n-k}$ and $A_{2,k}=\mathbf{1}^\top_k \mathbf{D}_{1:k,1:k} \mathbf{1}_k$. The third term is symmetric to the second and is handled similarly. Given $A_{1,k-1}$ and $A_{2,k-1}$, the updates are:
\begin{align*}
& A_{1,k} = A_{1,k-1} - \mathbf{1}^\top_k \mathbf{D}_{1:k,k} + \mathbf{D}_{k,k+1:n}\mathbf{1}_{n-k},\\
& A_{2,k} = A_{2,k-1} + 2\mathbf{D}_{k,1:k} \mathbf{1}^\top_k.
\end{align*}
Consequently, computing the sequence of numerators for all $k$ requires $O(n^2)$ operations, given $\mathbf{D}$.

The pooled variance estimator $\widehat{S}(\mathbf{X}_{1:k}, \mathbf{X}_{k+1:n})$ involves the terms $\widehat{\mathcal{D}}^2(\mathbf{X}_{1:k})$, $\widehat{\mathcal{D}}^2(\mathbf{X}_{k+1:n})$, and $\widehat{\mathcal{C}}(\mathbf{X}_{1:k},\mathbf{Y}_{k+1:n})$. For the generalized distance variance $\widehat{\mathcal{D}}^2(\mathbf{X}_{1:k})$, utilizing results from \cite{zys2018}, we have:
\[
\widehat{\mathcal{D}}^2(\mathbf{X}_{1:k}) = \frac{1}{k(k-3)} \left\{\text{tr}(\mathbf{D}_{1:k,1:k}^2) + \frac{(\mathbf{1}_k^\top \mathbf{D}_{1:k,1:k} \mathbf{1}_k)^2}{(k-1)(k-2)} - \frac{2 \mathbf{1}_k^\top \mathbf{D}^2_{1:k,1:k} \mathbf{1}_k}{k-2}\right\}.
\]
Let $B_{1,k} = \text{tr}(\mathbf{D}_{1:k,1:k}^2)$, $B_{2,k} = \mathbf{1}_k^\top \mathbf{D}_{1:k,1:k} \mathbf{1}_k$, $B_{3,k} = \mathbf{1}_k^\top \mathbf{D}^2_{1:k,1:k} \mathbf{1}_k$, and let $\mathbf{r}_k \in \mathbb{R}^{k}$ be the row sums of $\mathbf{D}_{1:k}$. The updates are:
\begin{align*}
B_{1,k} &= B_{1,k-1} + 2 \mathbf{D}_{1:k-1,k}^\top  \mathbf{D}_{1:k-1,k}, \\
B_{2,k} &= B_{2,k-1} + 2 \mathbf{1}^\top_{k-1} \mathbf{D}_{1:k-1,k}, \\
B_{3,k} &= B_{3,k-1} + (\mathbf{1}^\top_{k-1} \mathbf{D}_{1:k-1,k})^2 + \mathbf{D}_{1:k-1,k}^\top  \mathbf{D}_{1:k-1,k} + 2 \mathbf{r}_{k-1}^\top \mathbf{D}_{1:k-1,k}, \\
\mathbf{r}_k &= \begin{pmatrix} \mathbf{r}_{k-1} + \mathbf{D}_{1:k-1,k} \\ \mathbf{1}^\top_{k-1} \mathbf{D}_{1:k-1,k} \end{pmatrix}.
\end{align*}
Thus, $\widehat{\mathcal{D}}^2(\mathbf{X}_{1:k})$ is computed in $O(n^2)$ total. A similar logic applies to $\widehat{\mathcal{D}}^2(\mathbf{X}_{k+1:n})$. 

Finally, the cross-distance covariance $\widehat{\mathcal{C}}(\mathbf{X}_{1:k},\mathbf{Y}_{k+1:n})$ expands as:
\begin{align*}
\widehat{\mathcal{C}} &= \frac{1}{(k-1)(n-k-1)}\bigg\{ \text{tr}(\mathbf{D}_{1:k,k+1:n}\mathbf{D}_{1:k,k+1:n}^\top) + \frac{(\mathbf{1}_k^\top \mathbf{D}_{1:k,k+1:n}\mathbf{1}_{n-k})^2}{k(n-k)} \\
& \quad - \frac{1}{k}\mathbf{1}_k^\top \mathbf{D}_{1:k,k+1:n}\mathbf{D}_{1:k,k+1:n}^\top\mathbf{1}_k - \frac{1}{n-k}\mathbf{1}_{n-k}^\top \mathbf{D}_{1:k,k+1:n}^\top\mathbf{D}_{1:k,k+1:n}\mathbf{1}_{n-k} \bigg\}.    
\end{align*}
Define the four trace/quadratic terms as $C_{1,k}, \dots, C_{4,k}$. Let $\widetilde{\mathbf{r}}_{k}$ and $\widetilde{\mathbf{c}}_k$ be the row and column sums of $\mathbf{D}_{1:k,k+2:n}$. The updates are:
\begin{align*}
C_{1,k}& =C_{1,k-1} - \mathbf{D}_{1:k,k}^\top\mathbf{D}_{1:k,k}+ \mathbf{D}_{k,k+1:n}\mathbf{D}_{k,k+1:n}^\top,\\
C_{2,k}& =C_{2,k-1}  - \mathbf{1}_{k}^\top\mathbf{D}_{1:k,k} + \mathbf{D}_{k,k+1:n}\mathbf{1}_{n-k},\\
C_{3,k}& =C_{3,k-1}- (\mathbf{1}_{k}^\top\mathbf{D}_{1:k,k})^2+ \mathbf{D}_{k,k+1:n}\mathbf{D}_{k,k+1:n}^\top + 2\widetilde{\mathbf{c}}_{k-1}^\top \mathbf{D}_{k,k+1:n},\\
C_{4,k}& =C_{4,k-1}- \mathbf{D}_{1:k,k}^\top\mathbf{D}_{1:k,k} + (\mathbf{D}_{k,k+1:n}\mathbf{1}_{n-k})^2 -2 \widetilde{\mathbf{r}}_{k-1}^\top \mathbf{D}_{1:k-1,k},\\
\widetilde{\mathbf{r}}_{k} & =\begin{pmatrix} \widetilde{\mathbf{r}}_{k-1} - \mathbf{D}_{1:k-1,k+1} \\ \mathbf{D}_{k,k+2:n}\mathbf{1}_{n-k-1} \end{pmatrix}, \quad 
\widetilde{\mathbf{c}}_{k}  = \widetilde{\mathbf{c}}_{k-1,-1} + \mathbf{D}_{k,k+2:n},
\end{align*}
where $\widetilde{\mathbf{c}}_{k-1,-1}$ is $\widetilde{\mathbf{c}}_{k-1}$ with the first element removed. This maintains the $O(n^2)$ complexity for the denominator.

The calculation of the distance matrix $\mathbf{D}$ dominates the cost, requiring $O(n^2p)$ time and $O(n^2)$ space. The recursive updates for the statistic $M_n$ add only $O(n^2)$. Thus, the total complexity is $O(n^2p)$. This matches the theoretical complexity of \cite{MJ}, though we note that for ultra-high dimensional data (large $p$), the $O(n^2p)$ cost can still be prohibitive. To address this, we introduce computationally efficient surrogates in the next subsection. Our efficient C++ implementation is available in the R package \texttt{KDist} at \url{https://github.com/zhangxiany-tamu/KDist}.

\begin{remark}
{\rm 
When approximating the null distribution via permutation, we avoid recomputing the distance matrix. We simply permute the indices of the rows and columns of the pre-computed matrix $\mathbf{D}$ and run the $O(n^2)$ recursive updates. This makes the permutation test highly efficient.
}    
\end{remark}

\rev{\begin{remark}
{\rm 
For the monotone-invariant statistic, the pooled marginal ranks can be computed once before permutation. This preprocessing costs $O(pn\log n)$ using coordinatewise sorting. Since a permutation only reorders observations, it simply permutes the rows of the rank matrix and does not require recomputing ranks for each replicate. Therefore, after the rank matrix and the corresponding distance matrix have been computed, $B_{\mathrm{perm}}$ permutation replicates cost $O(B_{\mathrm{perm}}n^2)$ additional time using the recursive updates. Including rank construction and the $O(n^2p)$ distance-matrix calculation, the overall practical cost is $O(pn\log n+n^2p+B_{\mathrm{perm}}n^2)$. A naive implementation that recomputes ranks separately for each permutation would add an unnecessary $O(B_{\mathrm{perm}}pn\log n)$ cost.
}
\end{remark}}

\subsubsection{Computational surrogates for high-dimensional settings}\label{sec:surrogates}
The recursive algorithm described in Section \ref{sec:computational:recursive} achieves a complexity of $O(n^2p)$, significantly improving upon the naive $O(n^3p)$ implementation. However, for ultra-high dimensional data where $p$ is very large, the linear dependence on $p$ combined with the quadratic dependence on $n$ can still be computationally demanding. To address this, we propose two practical surrogates that reduce computational cost while preserving the non-parametric nature of our approach.

\paragraph{Surrogate A: Coordinate/Group Sketching.}
When $p$ is large, we can compute the test statistic on random low-dimensional \emph{sketches} of the data. Specifically, for $r=1,\dots,R$, we draw a random subset of coordinates (or feature groups) $\mathcal{S}_r\subseteq\{1,\dots,p\}$ of size $s\ll p$. We then apply our procedure to the projected data $\mathbf{X}_{t,\mathcal{S}_r}=(X_{t,j}:j\in\mathcal{S}_r)$ to obtain the sketch-specific statistics $M_n^{(r)}$ and estimators $\widehat{\nu}^{(r)}$. The results are aggregated across sketches, for instance, by taking the maximum:
\[
M_n^{\mathrm{A}} := \max_{1\le r\le R} M_n^{(r)}, \qquad \widehat{\nu}^{\mathrm{A}} := \widehat{\nu}^{(r^\star)}, \quad \text{where } r^\star=\argmax_{1\le r\le R} M_n^{(r)}.
\]
This strategy reduces the computational cost from $O(n^2p)$ to $O(Rn^2s)$ when $Rs=o(p)$. Max-aggregation is particularly effective when distributional changes are sparse (concentrated on a subset of coordinates), whereas mean or median aggregation can be employed when changes are diffuse.

\paragraph{Surrogate B: Incomplete U-statistic Approximation.}
Our generalized energy distance estimator is a U-statistic involving sums over all within- and between-segment pairs. To reduce the cost, we can approximate these full sums using \emph{incomplete} U-statistics based on a random subset of pairs. \rev{For a given split $k$, we sample up to $N_{\mathrm{pair}}$ pairs uniformly without replacement from each of the two within-segment pair sets $\{(i,j):1\le i<j\le k\}$ and $\{(i,j):k<i<j\le n\}$ and from the between-segment set $\{(i,j):1\le i\le k<j\le n\}$. If a set contains fewer than $N_{\mathrm{pair}}$ pairs, all of its pairs are used. Thus $N_{\mathrm{pair}}$ is a per-pair-set cap; the constant factor of at most three is suppressed in the complexity notation.} We then compute the corresponding subsampled estimator $\widetilde{E}_\gamma$ and self-normalized statistic $\widetilde{T}_n(k)$. The resulting scan statistic is defined as:
\[
\widetilde{M}_n^{\mathrm{B}} := \max_{4\le k\le n-4}\frac{k(n-k)}{n^2}\,\widetilde{T}_n(k).
\]
\rev{This approach reduces the pair-evaluation cost for a given split from $O(n^2p)$ to $O(N_{\mathrm{pair}}p)$. If pairs are sampled separately at each of the $O(n)$ candidate splits, the corresponding full-scan cost is $O(nN_{\mathrm{pair}}p)$, so a computational improvement over the exact $O(n^2p)$ scan requires $N_{\mathrm{pair}}=o(n)$ or an implementation that reuses sampled pairs across split points. As with any incomplete $U$-statistic approximation, overly aggressive subsampling may reduce power because too few sampled pairs may fail to capture the distributional discrepancy. Example~\ref{eg:surrogate_tradeoff} illustrates an analogous statistical--computational trade-off for Surrogate~A and thus provides qualitative guidance only; it is not a Surrogate~B-specific validation or a calibration of $N_{\mathrm{pair}}$.}

\begin{remark}
{\rm 
Surrogates A and B can be combined: for each sketch $\mathcal{S}_r$, one may use an incomplete U-statistic on the projected data. From a theoretical perspective, both surrogates introduce an additional approximation error (arising from sketching or subsampling). \rev{Retaining the consistency results of the exact method requires rates for $s$ and $N_{\mathrm{pair}}$ that control these approximation errors. We do not establish such rates here; their derivation and Surrogate~B-specific empirical calibration are left for future work.}
}
\end{remark}

\subsection{Recursive estimation of multiple change-point locations}\label{sec:methods:multiple cp}
In practice, the number and locations of change points are unknown. To consistently estimate all change-point locations while leveraging our single-change-point locator, we adopt a deterministic Seeded NOT that isolates an interval containing a single change point, applies the localized maximizer, and then recurses on the two sub-intervals.
\rev{Throughout this subsection and its theoretical analysis, the number of true change-points $N$ is fixed and does not depend on the sample size $n$.}

\subsubsection{Seeded Narrowest-Over-Threshold (Seeded NOT)}\label{sec:seeded_not}
To consistently \emph{estimate all} change-point locations, we recommend a seeded Narrowest-Over-Threshold (NOT; \cite{baranowski2019narrowest}) strategy.
It combines (i) a deterministic ``seeded'' family of candidate intervals (as in seeded binary segmentation \citep{k2023}) with (ii) the NOT selection rule: among all intervals whose single-change-point evidence exceeds a threshold, pick the \emph{narrowest} one.

Let $\mathcal I_n$ denote the seeded interval family
\[
\mathcal I_n := \bigcup_{j=0}^{\lfloor \log_2 n\rfloor}\Big\{[s,e]:\ e-s+1=2^j,\ s=1,\,1+2^{j-1},\,1+2\cdot 2^{j-1},\dots,\ e\le n\Big\}.
\]

For an interval $I = [s,e]$ in $\mathcal{I}_n$ and a split point $b \in \{s+3,\dots,e-4\}$, we define
\begin{align*}
&M(I) := \max_{b\in\{s+3,\dots,e-4\}}\frac{(b-s+1)(e-b)}{\rev{n^2}}\,T(\mathbf{X}_{s:b},\mathbf{X}_{(b+1):e}),\\
&\widehat\nu(I):=\argmax_b \frac{(b-s+1)(e-b)}{\rev{n^2}}\,\widehat E_{\gamma}(\mathbf{X}_{s:b},\mathbf{X}_{(b+1):e}).
\end{align*}

Given a threshold $\lambda_{I}$ (which may depend on the interval), we form the set of ``significant'' intervals
\[
\mathcal I_n(s,e;\lambda):=\Big\{I'=[s',e']\in\mathcal I_n:\ [s',e']\subseteq[s,e],\ M(I')>\lambda_{I'}\Big\},
\]
and select the narrowest significant interval $I^\star=\argmin_{I'\in\mathcal I_n(s,e;\lambda)} |I'|$ (breaking ties by the largest $M(I')$).
We then output $\widehat\nu(I^\star)$ and recurse on $[s,\widehat\nu(I^\star)]$ and $[\widehat\nu(I^\star)+1,e]$.
A precise pseudocode description is given in Algorithm~\ref{alg:seededNOT}.

\begin{algorithm}[t]
\caption{Seeded NOT with the single-change-point locator $\widehat\nu(\cdot)$}\label{alg:seededNOT}
\begin{algorithmic}[1]
\Require Data $\{\mathbf X_t\}_{t=1}^n$, seeded interval family $\mathcal I_n$, thresholds $\{\lambda_m\}_{m\ge 1}$.
\State Initialize an empty set of estimated change points $\widehat{\mathcal T}\gets\emptyset$.
\Procedure{NOT-Recursion}{$s,e$}
\If{$e-s+1<8$} \State \Return \EndIf
\State Form $\mathcal I_n(s,e;\lambda)$ and \textbf{if} it is empty, \Return
\State Choose $I^\star=\argmin_{I'\in\mathcal I_n(s,e;\lambda)}|I'|$ (ties by largest $M(I')$)
\State $\widehat\nu\gets \widehat\nu(I^\star)$, \ \ $\widehat{\mathcal T}\gets \widehat{\mathcal T}\cup\{\widehat\nu\}$
\State \Call{NOT-Recursion}{$s,\widehat\nu$}; \ \ \Call{NOT-Recursion}{$\widehat\nu+1,e$}
\EndProcedure
\State \Call{NOT-Recursion}{$1,n$}; \ \ Output $\widehat{\mathcal T}$.
\end{algorithmic}
\end{algorithm}

\begin{remark}
{\rm 
The seeded family $\mathcal I_n$ has cardinality $|\mathcal I_n|=O(n\log n)$ and contains, for each change-point that is separated from its neighbors,
an ``isolating'' interval that contains this change-point but no others; this property underpins the consistency result in Theorem~\ref{mult_consistency_SBS}
below. 
}
\end{remark}

For the multiple change point model specified in (\ref{alt:mult}), we define $\delta_{i,j} = 2\tau^{(i,j)} - \tau^{(i,i)} - \tau^{(j,j)}$. Also, let $V_n = \max\{\log(n) V_n^{(1)},V_n^{(2)}\}$,where
\[\begin{split}
    V_{n,i}^{(1)} = \max\{\E(H(X,X')^2), \E(H(Y,Y')^2), \E(H(X,Y)^2)\},
\end{split}\]
and
\begin{equation*}
\begin{split}V_{n,i}^{(2)} = \max\{n\var(\tau^{(i,i+1)}\E[L(X,Y)|X] - \tau^{(i,i)}\E[L(X,X')|X]),\\n\var(\tau^{(i,i+1)}\E[L(X,Y)|X] - \tau^{(i,i)}\E[L(Y,Y')|Y])\}.
\end{split}
\end{equation*}
where $X \sim F_i$ and $Y \sim F_{i+1}$ for $i = 1,...,N$. The next theorem states the consistency result for our multiple change point estimators.

\begin{theorem}\label{mult_consistency_SBS}
(Consistency of Seeded NOT for multiple change-points.)
Consider the multiple change-point model \rev{\eqref{alt:mult} with $N\ge 1$ change-points at $0< \zeta_1<\cdots<\zeta_N<1$, and write $\zeta_\ell:=\nu_\ell/n$.
Assume that $N$ is fixed as $n,p\to\infty$.}
Let $\widehat{\mathcal T}^{\mathrm{NOT}}$ be the set of estimated change-point locations returned by Algorithm~\ref{alg:seededNOT},
and denote its ordered elements by $\widehat\nu^{\mathrm{NOT}}_{(1)}<\cdots<\widehat\nu^{\mathrm{NOT}}_{(\widehat N)}$
(with $\widehat N:=|\widehat{\mathcal T}^{\mathrm{NOT}}|$) and $\widehat\zeta^{\mathrm{NOT}}_{(\ell)}:=\widehat\nu^{\mathrm{NOT}}_{(\ell)}/n$.

If for an interval $I = [s,e]$, the threshold $\lambda_{I}$ for each subsample of $X_s,...,X_e$ is chosen such that $\lambda_{{I}}\rightarrow \infty$ and
\begin{enumerate}
\item  $\sqrt{V_{n,l}/V_{n,l}^{(1)}} = o(\lambda_{I})$;
\item $\lambda_{I} = o_p(b_{|I|,\ell})$,
\end{enumerate}
then under Assumptions \ref{ass0_new}-\ref{ass2_new}, 
\[
P\!\left(\widehat N=N\ \text{ and }\ \max_{1\le \ell\le N} b_{n,\ell}\big|\widehat\zeta^{\mathrm{NOT}}_{(\ell)}-\zeta_\ell\big|\le C\right)\to 1,
\]
where $b_{n,\ell} := n\delta_{\ell,\ell+1}/\sqrt{V_{n,\ell}}\rightarrow \infty$ and $C$ is a positive constant. 
\end{theorem}

\begin{remark}
    In the above theorem, the two conditions together determine an admissible growth rate for $\lambda_I$. The first condition ensures that, with high probability, no interval that does not contain a change point is selected. The second condition ensures that at least one interval containing a change point is selected. In practice, $\lambda_I$ is chosen as a high quantile of the limiting null distribution.
\end{remark}

\section{Numerical studies} \label{sec:num}
\subsection{Simulation studies}\label{simulations}

In this subsection, we examine the finite sample performance of our proposed methodology for single and multiple change-point detection via simulation studies. We evaluate two variations of our procedure implemented in the R package \texttt{KDist}:
\begin{enumerate}
    \item \textbf{KDist:} The standard procedure using the distance metric $\gamma(z,z')=\|z-z'\|_1^{1/2}$ on the original data.
    \item \textbf{KDist-MI:} The componentwise monotone-invariant estimator ($M_n^{\mathrm{MI}}$) described in Section \ref{sec:mono_invariant}, which applies the same metric to the rank-transformed data.
\end{enumerate}
For multiple change-point detection, we employ the Seeded NOT algorithm (Section \ref{sec:seeded_not}). We compare our approach against the following state-of-the-art methods:
\begin{itemize}
	\item \textbf{MJ:} The E-Divisive procedure \citep{MJ} (R package `ecp').
	\item \textbf{CZ:} The graph-based original scan statistic \citep{CZ} (R package `gSeg').
	\item \textbf{CC:} The max-type edge-count test \citep{CC} (R package `gSeg').
	\item \textbf{WS:} The INSPECT procedure \citep{WS} (R package `InspectChangepoint').
	\item \textbf{AB:} The covariance change-point test \citep{avanesov} (R package `covcp').
    \item \textbf{KCPD:} The kernel multiple change-point algorithm \citep{arlot} using the `ruptures' Python library (RBF kernel). \textbf{KCPD*} denotes a modified version that permits returning zero change-points.
\end{itemize}
It is important to note that the methodology proposed by \cite{WS} focuses on detecting mean shifts in high-dimensional data. Similarly, the procedure developed by \cite{arlot} is aimed at identifying changes in high-dimensional covariance structures. In KCPD, as described by \cite{arlot}, we utilize the radial basis function kernel, with the bandwidth determined by the median heuristics. Additionally, we set the constants \(c_1\) and \(c_2\) in KCPD using the ``slope heuristics'' method outlined in Section 6.2 of \cite{arlot}. We note that the original KCPD will always return at least one change-point. To address this issue, we also implement a version of KCPD (denoted by KCPD*) with the penalty term being zero when there is no change point, which allows the algorithm to return a zero number of change points. We compare our method to these competitors to demonstrate that, when changes occur in higher-order moments, our approach outperforms theirs in both detecting and localizing the unknown change-points. We first consider examples under the null hypothesis and single change-point alternatives.

\begin{example}[No structural break]\label{eg0}
	{\rm 
		~
		\begin{enumerate}
			\item $X_t \overset{i.i.d.}{\sim} N(0, \mathbf{I}_p)$ \, for\, $1\leq t \leq  n $.
			\item $X_t \overset{i.i.d.}{\sim} N(0, \Sigma)$ \, for\, $1\leq t \leq n$,\, where $\Sigma = (\sigma_{ij})_{i,j=1}^p$ with $\sigma_{ij} = 0.7^{|i-j|}$.
			\item For each\, $i=1, \dots, p$, $\{X_{1,i}, \dots , X_{n,i}\}$ is generated independently from the ARCH(2) model $X_{t,i} = \sigma_{t,i}\, \epsilon_{t,i}$, with\, $\sigma_{t,i}^2 = \alpha_0 + \alpha_1\,X_{t-1,i}^2 + \alpha_2\,X_{t-2,i}^2\,,$ where $\epsilon_{t,i} \overset{i.i.d.}{\sim} N(0,1)$ for $1\leq t \leq n$. We consider\, $\alpha_0=10^{-6},\, \alpha_1=0.008$\, and\, $\alpha_2=0.001$.
			\item For each\, $i=1, \dots, p$, $\{X_{1,i}, \dots , X_{n,i}\}$ is generated independently from the GARCH(1,1) model $X_{t,i} = \sigma_{t,i}\, \epsilon_{t,i}$, with\, $\sigma_{t,i}^2 = \alpha_0 + \alpha_1\,X_{t-1,i}^2 + \beta_1\,\sigma_{t-1,i}^2\,,$ where $\epsilon_{t,i} \overset{i.i.d.}{\sim} N(0,1)$ for $1\leq t \leq n$. We consider\, $\alpha_0=10^{-6},\, \alpha_1=0.001$\, and\, $\beta_1=0.001$.
		\end{enumerate}
	}
\end{example}

\begin{example}[Single change-point in mean]\label{eg1} 
	{\rm 
		~
		\begin{enumerate}
			\item $X_t \overset{i.i.d.}{\sim} N(0, \mathbf{I}_p)$ \, for\, $1\leq t \leq \lfloor n/2 \rfloor$\, and \,$X_t \overset{i.i.d.}{\sim} N(\mu, \mathbf{I}_p)$ \, for $ \lfloor n/2 \rfloor +1 \leq t \leq n$, where $\mu = (0.6, \dots, 0.6) \in \mathbb{R}^p$.
			\item $X_t \overset{i.i.d.}{\sim} N(0, \Sigma)$ \, for\, $1\leq t \leq \lfloor n/2 \rfloor$\, and \,$X_t \overset{i.i.d.}{\sim} N(\mu, \Sigma)$ \, for $ \lfloor n/2 \rfloor +1 \leq t \leq n$, where $\Sigma = (\sigma_{ij})_{i,j=1}^p$ with $\sigma_{ij} = 0.7^{|i-j|}$, and $\mu = (0.6, \dots, 0.6) \in \mathbb{R}^p$.
		\end{enumerate}
	}
\end{example}

\begin{example}[Single change-point in higher-order moments]\label{eg2}
	{\rm 
		~
		\begin{enumerate}
			\item $X_t \overset{i.i.d.}{\sim} N(\mu, \mathbf{I}_p)$ with $\mu = (1, \dots, 1)\in\mathbb{R}^p$\, for\, $1\leq t \leq \lfloor n/2 \rfloor$\, and \,$X_{t,i} \overset{i.i.d.}{\sim}$ Exponential\,$(1)$ for $i=1,\dots,p$\, and\, $ \lfloor n/2 \rfloor +1 \leq t \leq n$.
			
			\item $X_t = \underbrace{(X_{t,1}, \dots, X_{t,p})}_{ \overset{i.i.d.}{\sim} \text{Poisson}(1)} -\, 1$\, for\, $1\leq t \leq \lfloor n/2 \rfloor$\, and \,$X_t = (X_{t,1},\, \dots,\, X_{t,\lfloor  p/2\rfloor},\, X_{t,(\lfloor  p/2\rfloor +1)},\,\\\dots ,\,  X_{t,p})$ \, where $X_{t,1}, \dots, X_{t,\lfloor  p/2\rfloor} \overset{i.i.d.}{\sim}$ \text{Poisson}\,$(1) -1$,\, and \,$X_{t,(\lfloor p/2\rfloor +1)},\, \dots ,\,  X_{t,p} \overset{i.i.d.}{\sim} $ Rademacher\,$(0.5)$\, for \, $\lfloor n/2 \rfloor +1 \leq t \leq n$.
			
			\item $X_t = \underbrace{(X_{t,1}, \dots, X_{t,p})}_{ \overset{i.i.d.}{\sim} \text{Poisson}(1)} -\, 1$\, for\, $1\leq t \leq \lfloor n/2 \rfloor$\, and \,$X_t = (X_{t,1},\, \dots,\, X_{t,\lfloor 4 p/5\rfloor},\, X_{t,(\lfloor 4 p/5\rfloor +1)},\,\\ \dots ,\,  X_{t,p})$ \, where $X_{t,1},\, \dots,\, X_{t,\lfloor 4 p/5\rfloor} \overset{i.i.d.}{\sim}$ \text{Poisson}\,$(1) -1$,\, and \,$X_{t,(\lfloor 4 p/5\rfloor +1)},\, \dots ,\,  X_{t,p} \overset{i.i.d.}{\sim} $ Rademacher\,$(0.5)$\, for \, $\lfloor n/2 \rfloor +1 \leq t \leq n$.
			
			\item $X_t = R^{1/2} Z_{1t}$ \, for\, $1\leq t \leq \lfloor n/2 \rfloor$\, and \,$X_t = R^{1/2} Z_{2t}$ \, for $ \lfloor n/2 \rfloor +1 \leq t \leq n$,\, where $R = (r_{ij})_{i,j=1}^p$ with $r_{ii}=1$ for\, $i=1, \dots, p$, $r_{ij} = 0.25$ if $1 \leq |i-j| \leq 2$ and $r_{ij} = 0$ otherwise, $Z_{1t} \overset{i.i.d.}{\sim} N(0, \mathbf{I}_p)$ and  $Z_{2t} = \underbrace{(Z_{2t,1}, \dots, Z_{2t,p})}_{ \overset{i.i.d.}{\sim} \text{Exponential}(1)} -\, 1.$
		\end{enumerate}
	}
\end{example}

In Example \ref{eg2}, the change occurs in the higher-order moments or the distributional form, while the mean and covariance structure may remain constant (or similar). We consider $n=100$ and $p=100,200$. We implement Algorithm 1 with \rev{$B_{\mathrm{perm}}=199$} permutation replicates and a significance level of $\alpha=0.05$. We cluster the observations based on the estimated significant change-point locations and compute the Adjusted Rand Index (ARI) \citep{MA}. The ARI is a positive value between 0 and 1. The ARI value is 0 when there is no change-point, but the method estimates one (or more) change-point location. The ARI value is 1 when the estimation is perfect. The higher the value of ARI, the more accurate the estimation of the change-point locations. We conduct 100 simulations for each example mentioned above, calculating the ARI value and reporting it in the table below.

\begin{table}[!h]\footnotesize 
	\centering
	\caption{Comparison of average ARI values for different methods over 100 simulations, where $n=100$.}
	\label{tab1}
	\begin{tabular}{c c c c c c c c c c c c}
		\toprule
		& & $p$ & \textbf{KDist} & \textbf{KDist-MI} & MJ & CC & CZ & WS & AB & KCPD  & KCPD* \\
		\hline
		\multirow{8}{*}{Ex \ref{eg0}} & (1) & 100 & 0.980 & 0.940 & 0.970 & 0.970 & 0.970 & 0.000 & 1.000  & 0.000 & 1.000 \\
		& (1) & 200 & 0.970 & 0.920 & 0.980 & 0.960 & 0.960 & 0.000 & 1.000 & 0.000 & 1.000\\
		& (2) & 100 & 0.930 & 0.890 & 0.970 & 0.910 & 0.920 & 0.000 & 1.000 & 0.000 & 1.000\\
		& (2) & 200 & 0.970 & 0.920 & 0.970 & 0.950 & 0.980 & 0.000 & 1.000 & 0.000 & 1.000\\
		& (3) & 100 & 0.960 & 0.920 & 0.940 & 0.980 & 0.910 & 0.000 & 1.000  & 0.000 & 1.000\\
		& (3) & 200 & 0.970 & 0.920 & 0.950 & 0.960 & 0.960 & 0.000 & 1.000 & 0.000& 1.000\\
		& (4) & 100 & 0.950 & 0.960 & 0.950 & 0.990 & 0.930 & 0.000 & 1.000  & 0.000 & 1.000\\
		& (4) & 200 & 0.970 & 0.970 & 0.960 & 0.920 & 0.920 & 0.000 & 1.000  & 0.000 & 1.000\\
		\hline
		\multirow{4}{*}{Ex \ref{eg1}} & (1) & 100 & 1.000 & 1.000 & 1.000 & 0.997 & 0.999 & 1.000 & 0.121  & 1.000  & 1.000\\
		& (1) & 200 &  1.000 & 1.000 & 1.000 & 0.999 & 0.999 & 1.000 & 0.111 & 1.000 & 0.000\\
		& (2) & 100 & 0.984 & 0.987 & 0.986 & 0.867 & 0.946 & 0.981  & 0.256 & 0.986 & 0.992 \\
		& (2) & 200 & 0.996 & 0.996 & 0.996 & 0.978 & 0.983 & 0.993 & 0.138 & 0.993 & 0.000\\
		\hline
		\multirow{6}{*}{Ex \ref{eg2}}  & (1) & 100 & 0.993 & 0.991 & 0.014 & 0.004 & 0.027 & 0.390 & 0.000 & 0.197 & 0.000\\
		& (1) & 200 & 1.000 & 0.998 & 0.030 & 0.007 & 0.037 & 0.414  & 0.000  &  0.238 & 0.000\\
		& (2) & 100 & 0.999 & 0.998 & 0.034 & 0.001 & 0.059 & 0.468 & 0.425   & 0.214  & 0.000 \\
		& (2) & 200 & 1.000 & 1.000 & 0.032 & 0.001 & 0.055 & 0.502 & 0.529  & 0.243 & 0.000\\
		& (3) & 100 &  0.976 & 0.962 & 0.018 & 0.002 & 0.040 & 0.450 & 0.188  & 0.213 & 0.000\\
		& (3) & 200 & 0.992 & 0.987 & 0.042 & 0.000 & 0.050 &  0.494 & 0.242 & 0.214 & 0.000 \\
		& (4) & 100 &  0.978 & 0.964 & 0.024 & 0.021 & 0.065 & 0.402 & 0.000 & 0.154 & 0.000 \\
		& (4) & 200 & 0.992 & 0.978 & 0.029 & 0.006 & 0.040 & 0.363 & 0.000 & 0.249 & 0.000 \\
		\bottomrule
	\end{tabular}
\end{table}

The results presented in Table \ref{tab1} show that most methods perform similarly well when there is no structural break or when a simple mean shift occurs. However, several competitors exhibit specific limitations. In the absence of a structural break (Ex \ref{eg0}), the procedure developed by \cite{WS} incorrectly detects a break, resulting in a zero ARI value. By design, KCPD also fails under the null as it always reports at least one change point; in contrast, KCPD* operates effectively in situations where there are no change points.

Although our methodology is designed for an i.i.d. sequence of observations, the results from Examples \ref{eg0}.2 and \ref{eg2}.2 indicate that both KDist and KDist-MI perform reasonably well even in the presence of relatively weak conditional heteroskedasticity (ARCH/GARCH) and temporal dependence. In contrast, the method proposed by \cite{avanesov} (AB) does not perform well in Examples \ref{eg1}.1-\ref{eg1}.2, where there is a change in the mean while the covariance structure remains unchanged. Additionally, we note that KCPD* fails when \(n=100\) and \(p=200\) in Example \ref{eg1}.

Most interestingly, when changes occur in the distribution beyond the first two moments (Example \ref{eg2}), our method significantly outperforms the competitors. Both KDist and KDist-MI maintain ARI scores above 0.96 across these challenging scenarios (including tail and copula changes). Intuitively, the E-Divisive procedure (MJ) has low detection power here, as the Euclidean energy distance fails to capture inhomogeneity between two high-dimensional distributions beyond the first two moments. Similarly, while the graph-based methods of \cite{CZ} (CZ) and \cite{CC} (CC) are effective for location and scale alternatives, they are ineffective in detecting changes in higher-order moments. The performances of KCPD and KCPD* are also lacking in this scenario, likely due to the ineffectiveness of the standard Gaussian kernel in detecting complex distributional changes in high dimensions. Notably, KDist-MI performs nearly identically to KDist across all alternatives, demonstrating that the gain in theoretical invariance comes with minimal loss of statistical power.

In Table \ref{tab1}, we illustrated the average ARI values obtained over 100 simulated datasets on Examples \ref{eg0}-\ref{eg2}, implementing Algorithm 1 by approximating the quantiles of $M_n$ via a permutation procedure. Alternatively, in Remark \ref{remark quantiles} (Section \ref{sec:methods:test stat}), we presented the approximated quantiles of the asymptotic null distribution. Table \ref{tab1_asymp_quantiles1} below compares the Type I error rates (the proportion of false detections in Example \ref{eg0}) using these permutation-based versus asymptotic critical values. The results indicate that the proportions of false detection are quite close in both cases. Using asymptotic quantiles provides approximately valid Type~I error control for both KDist and KDist-MI, with empirical sizes generally close to the nominal 5\% level. While we observe slight size inflation for the asymptotic calibration in some settings (most notably under dependence), the permutation method tends to be slightly more conservative. \rev{These empirical results suggest asymptotic calibration as a faster heuristic alternative to Algorithm~\ref{alg1}, especially for KDist-MI, rather than establishing its general validity. In this example it is about seven times faster because it avoids repeated resampling, while remaining reasonably accurate under the simulated weak conditional heteroskedasticity.}
\begin{table}[!ht]\footnotesize
	\centering
	\caption{Percentage of false positives (Type I error) in Example \ref{eg0} over 100 simulations ($n=100, \alpha=0.05$).}
	\label{tab1_asymp_quantiles1}
	\begin{tabular}{c c c c c c c}
		\toprule
		& & $p$ 
		& \multicolumn{2}{c}{\textbf{KDist}} 
		& \multicolumn{2}{c}{\textbf{KDist-MI}}\\
		\cmidrule(lr){4-5}\cmidrule(lr){6-7}
		& & 
		& Permutation & Asymptotic
		& Permutation & Asymptotic \\
		\midrule
		\multirow{8}{*}{Ex \ref{eg0}} 
		& (1) & 100 & 0.02 & 0.04 & 0.02 & 0.02 \\
		& (1) & 200 & 0.03 & 0.03 & 0.04 & 0.04 \\
		& (2) & 100 & 0.07 & 0.10 & 0.08 & 0.10 \\
		& (2) & 200 & 0.03 & 0.07 & 0.07 & 0.11 \\
		& (3) & 100 & 0.04 & 0.04 & 0.05 & 0.05 \\
		& (3) & 200 & 0.03 & 0.03 & 0.05 & 0.04 \\
		& (4) & 100 & 0.05 & 0.05 & 0.04 & 0.05 \\
		& (4) & 200 & 0.03 & 0.03 & 0.06 & 0.06 \\
		\bottomrule
	\end{tabular}
\end{table}

The following examples illustrate the performance of Algorithm 2 in the cases of two change-points.

\begin{example}[Two change-points in mean]\label{eg3}
	{\rm 
		~
		\begin{enumerate}
			\item $X_t \overset{i.i.d.}{\sim} N(0, \mathbf{I}_p)$ \, for\, $1\leq t \leq \lfloor n/3 \rfloor$\, and\, $2\lfloor n/3 \rfloor +1\leq t \leq n$,\, and \,$X_t \overset{i.i.d.}{\sim} N(\mu, \mathbf{I}_p)$ \, for $ \lfloor n/3 \rfloor +1 \leq t \leq 2\lfloor n/3 \rfloor$, where\, $\mu = (0.6, \dots, 0.6) \in \mathbb{R}^p$.
			\item $X_t \overset{i.i.d.}{\sim} N(0, \Sigma)$ \, for\, $1\leq t \leq \lfloor n/3 \rfloor$\, and\, $2\lfloor n/3 \rfloor +1\leq t \leq n$,\, and \,$X_t \overset{i.i.d.}{\sim} N(\mu, \Sigma)$ \, for $ \lfloor n/3 \rfloor +1 \leq t \leq 2\lfloor n/3 \rfloor$, where\, $\Sigma = (\sigma_{ij})_{i,j=1}^p$ with $\sigma_{ij} = 0.7^{|i-j|}$ and $\mu = (0.6, \dots, 0.6) \in \mathbb{R}^p$.
		\end{enumerate}
	}
\end{example}

\begin{example}[Two change-points in higher-order moments]\label{eg4}
	{\rm 
		~
		\begin{enumerate}
			\item $X_t \overset{i.i.d.}{\sim} N(\mu, \mathbf{I}_p)$ with $\mu = (1, \dots, 1)\in\mathbb{R}^p$\, for\, $1\leq t \leq \lfloor n/3 \rfloor$\, and\, $2\lfloor n/3 \rfloor +1\leq t \leq n$,\\ and \,$X_{t,i} \overset{i.i.d.}{\sim}$ Exponential\,$(1)$\, for \,$i=1,\dots,p$\, and \,$ \lfloor n/3 \rfloor +1 \leq t \leq 2\lfloor n/3 \rfloor$.
			
			\item $X_t = \underbrace{(X_{t,1},\, \dots,\, X_{t,p})}_{ \overset{i.i.d.}{\sim} \text{Poisson}(1)} -\, 1$ \, for\, $1\leq t \leq \lfloor n/3 \rfloor$\, and \,\,$2\lfloor n/3 \rfloor +1\leq t \leq n$,\, and $X_t = (X_{t,1},\, \dots,\, X_{t,\lfloor  p/2 \rfloor},\, X_{t,(\lfloor p/2\rfloor +1)},\, \dots,\,  X_{t,p})$, where $X_{t,1},\, \dots,\, X_{t,\lfloor p/2\rfloor} \overset{i.i.d.}{\sim}$ \text{Poisson}\,$(1) -1$,\, and $X_{t,(\lfloor p/2\rfloor +1)},\, \dots,\,  X_{t,p} \overset{i.i.d.}{\sim} $ Rademacher\,$(0.5)$\, for \,$ \lfloor n/3 \rfloor +1 \leq t\leq 2\lfloor n/3 \rfloor$.
			
			\item $X_t = \underbrace{(X_{t,1},\, \dots,\, X_{t,p})}_{ \overset{i.i.d.}{\sim} \text{Poisson}(1)} -\, 1$ \, for\, $1\leq t \leq \lfloor n/3 \rfloor$\, and \,\,$2\lfloor n/3 \rfloor +1\leq t \leq n$,\, and $X_t = (X_{t,1},\, \dots,\, X_{t,\lfloor 4 p/5\rfloor},\, X_{t,(\lfloor 4 p/5\rfloor +1)},\, \dots ,  X_{t,p})$, where $X_{t,1},\, \dots,\, X_{t,\lfloor 4 p/5\rfloor} \overset{i.i.d.}{\sim}$ \text{Poisson}\,$(1) -1$,\, and $X_{t,(\lfloor 4 p/5\rfloor +1)},\, \dots,\,  X_{t,p} \overset{i.i.d.}{\sim} $ Rademacher\,$(0.5)$\, for \,$ \lfloor n/3 \rfloor +1 \leq t\leq 2\lfloor n/3 \rfloor$.
		\end{enumerate}
	}
\end{example}

\begin{table}[!ht]\footnotesize
	\centering
	\caption{Comparison of average ARI values for multiple change-points (100 simulations, $n=100$).}
	\label{tab2}
	\begin{tabular}{c c c c c c c c c c c}
		\toprule
		& & $p$ & \textbf{KDist} & \textbf{KDist-MI} & MJ & CC & CZ &  WS & KCPD & KCPD*\\
		\hline
		\multirow{4}{*}{Ex \ref{eg3}}  
		& (1) & 100 & 0.994 & 0.997 & 1.000 & 0.975 & 0.960 & 0.916 & 0.620 & 0.020 \\
		& (1) & 200 & 0.998 & 0.998 & 1.000 & 0.996 & 0.959 & 0.886 & 0.568 & 0.000 \\
		& (2) & 100 & 0.958 & 0.960 & 0.978 & 0.747 & 0.885 & 0.643 & 0.526 & 0.162 \\
		& (2) & 200 & 0.984 & 0.982 & 0.994 & 0.912 & 0.935 & 0.619 & 0.543 & 0.000 \\
		\hline
		\multirow{6}{*}{Ex \ref{eg4}}   
		& (1) & 100 & 0.967 & 0.942 & 0.024 & 0.487 & 0.412 & 0.383 & 0.210 & 0.000 \\
		& (1) & 200 & 0.989 & 0.985 & 0.054 & 0.496 & 0.413 & 0.279 & 0.190 & 0.000 \\
		& (2) & 100 & 0.987 & 0.982 & 0.028 & 0.214 & 0.528 &  0.179 & 0.143 & 0.000 \\
		& (2) & 200 & 0.989 & 0.991 & 0.028 & 0.244 & 0.534 & 0.172 & 0.170 & 0.000 \\
		& (3) & 100 & 0.599 & 0.510 & 0.023 & 0.255 & 0.473 &  0.204 & 0.149 & 0.000 \\
		& (3) & 200 & 0.863 & 0.801 & 0.016 & 0.277 & 0.504 & 0.193 & 0.184 & 0.000 \\
		\bottomrule
	\end{tabular}
\end{table}

The method proposed by \cite{avanesov} (AB) cannot be compared in this setting because the R package `covcp' is restricted to testing for a single change-point. Examples \ref{eg4}.2 and \ref{eg4}.3 are similar in structure; the primary difference is that \ref{eg4}.3 considers distributional changes in a sparser subset of the components.

The results from Table \ref{tab2} indicate that almost all methods (except for KCPD) perform nearly equally well in Example \ref{eg3}.1, where there are two change-points in the mean. However, in Example \ref{eg3}.2 (Rademacher Mean Shift), our methodology—both KDist and KDist-MI—along with the E-Divisive procedure and graph-based methods, perform considerably better than the rest.
Most interestingly, when there are two change-points in the distribution beyond the first two moments (Example \ref{eg4}), our method significantly outperforms the other competitors in accurately estimating the locations. As expected, the E-Divisive procedure (MJ) suffers from low detection power here, as the Euclidean energy distance fails to capture inhomogeneity between high-dimensional distributions beyond the first two moments. Similarly, KCPD and KCPD* are ineffective at detecting these distributional changes in high dimensions, likely due to their reliance on the standard Gaussian kernel. Our results also indicate that our method performs significantly better than the graph-based methods proposed by \cite{CZ} and \cite{CC} in detecting and localizing general types of changes in the underlying distribution beyond the first two moments.
In summary, our numerical investigations illustrate that our methodology is more sensitive to general distributional changes in high dimensions compared to existing state-of-the-art methods.

The following examples illustrate the performance of Algorithm 2 in scenarios with changes at varying sparsity levels.

\begin{example}[Changes with varying sparsity levels]\label{eg6}
{\rm 
	\begin{enumerate}
	\item $X_t \overset{i.i.d.}{\sim} N(0, \mathbf{I}_p)$ \, for\, $1\leq t \leq \lf 
         n/3 \rf$\, and\, $2\lf n/3 \rf +1\leq t \leq n$,\, and \,$X_t \overset{i.i.d.}{\sim} N(\mu_s, \mathbf{I}_p)$ \, for $ \lf n/3 \rf +1 \leq t \leq 2\lf n/3 \rf$, where\, $\mu_s = (\underbrace{0.4, \dots, 0.4}_{s \cdot p}, \underbrace{0, \dots, 0}_{(1-s) \cdot p}) \in {\bb R}^p$ and $s \in (0,1]$ represents the sparsity level. We consider $n = 150$ observations in $p = 100$ dimensions, with true change-points at positions $\lfloor n/3 \rfloor = 50$ and $\lfloor 2n/3 \rfloor = 100$. The parameter $s$ controls the proportion of components exhibiting the mean shift, ranging from very sparse ($s = 0.05$) to dense ($s = 1$), while maintaining a constant signal strength of $\delta = 0.4$ across all affected components.

        \item $X_t \overset{i.i.d.}{\sim} N(0, \mathbf{I}_p)$ \, for\, $1\leq t \leq \lf n/3 \rf$\, and\, $2\lf n/3 \rf +1\leq t \leq n$. For $ \lf n/3 \rf +1 \leq t \leq 2\lf n/3 \rf$, $X_{t,i} \overset{i.i.d.}{\sim} \chi^2_2$\, (centered and scaled to have mean 0 and variance 1) for $1 \leq i \leq \lfloor s \cdot p\rfloor$, and $X_{t,i} \overset{i.i.d.}{\sim} N(0,1)$ for $\lfloor s \cdot p \rfloor +1 \leq  i \leq p$, where $s \in (0,1]$ represents the sparsity level.
		\end{enumerate}
}
\end{example}

Figure \ref{fig_sparse} presents a comparative analysis of our method against the approach of \cite{WS} across various sparsity levels in the detection of change-points. In Example~\ref{eg6}.1, where the changes manifest as mean shifts in only a subset of components, we observe an interesting pattern in performance. When the sparsity level is low ($s < 0.2$), Wang and Samworth's method demonstrates superior performance, which aligns with its design optimization for sparse mean-shift detection. However, as the signal becomes more dense ($s \geq 0.2$), our method achieves higher ARI values, indicating better change-point detection accuracy.

We observe a similar trend in Example~\ref{eg6}.2, as illustrated in Figure \ref{fig_sparse}. Our method outperforms Wang and Samworth's approach when \( s = 0.3 \) and maintains a relatively high power for denser signals. In contrast, the power of Wang and Samworth's method does not increase with the level of sparsity. This underscores an advantage of our distribution-based approach: it can detect changes beyond the first two moments and exhibits a monotonically increasing power as the sparsity level rises.

\begin{figure}[!ht]
	\caption{Comparison of change-point detection performance between our method and \cite{WS} across varying sparsity levels. Panel A shows performance for sparse mean shifts with constant signal strength $\delta=0.4$, while Panel B shows performance for sparse distributional changes in higher-order moments (using $\chi^2_2$ distribution with 2 degrees of freedom, centered and scaled).}\label{fig_sparse}
	\centering
	\includegraphics[width=1\textwidth]{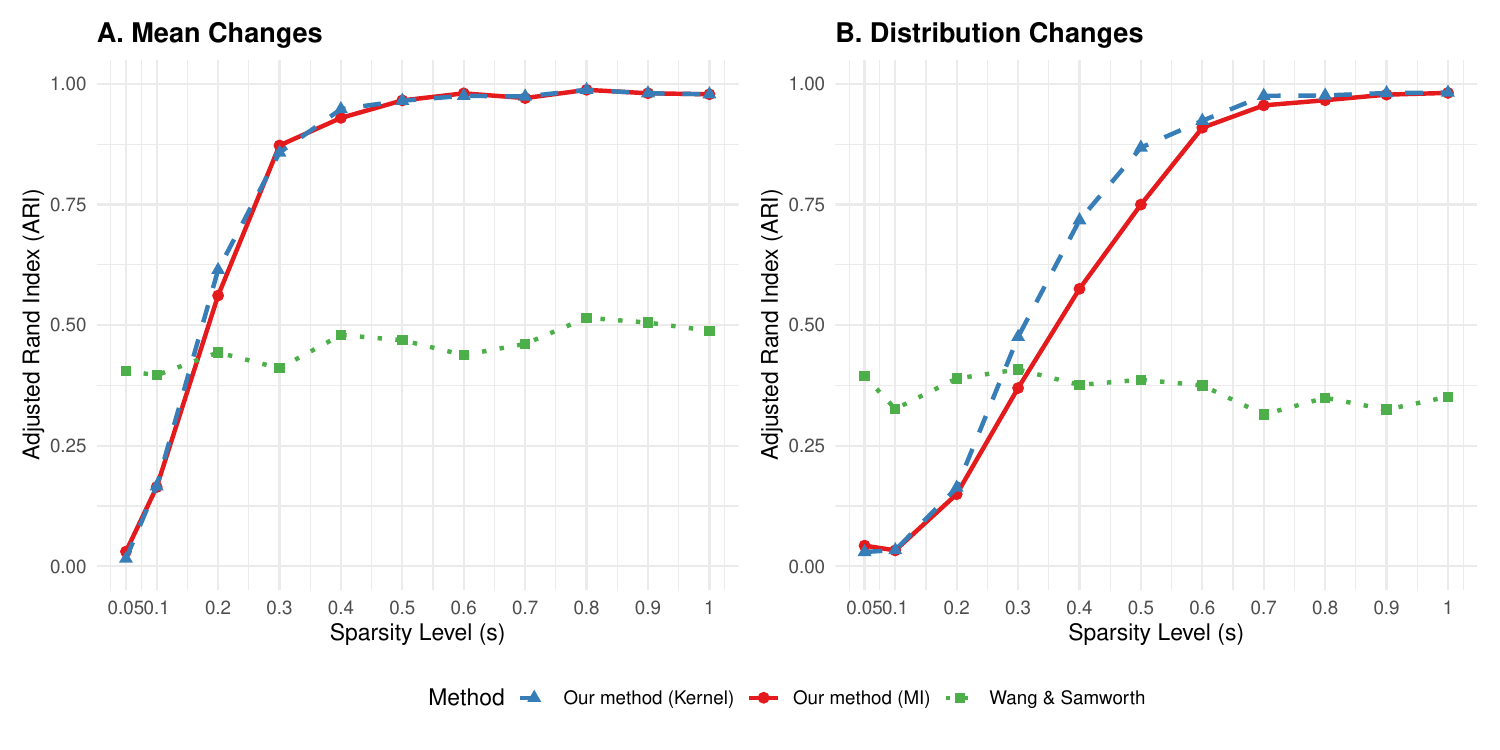}
\end{figure}

Our next example illustrates the practical trade-off introduced by the sketching surrogate in Section~\ref{sec:surrogates}. 

\begin{example}
\label{eg:surrogate_tradeoff}
{\rm
We illustrate the trade-off induced by Surrogate~A (coordinate sketching) in an ultra-high dimensional regime.
We generate an independent sequence $\{X_t\}_{t=1}^n$ with $X_t\sim N(0,\mathbf{I}_p)$, $n=200$ and $p=10{,}000$, with a single change-point at $\nu=100$.
We apply KDist-MI using asymptotic calibration, and
evaluate the localization accuracy $\mathbb{P}(|\widehat{\nu}-\nu|\le \delta)$ with tolerance $\delta=5$ over 100 replications.
We consider a sketching grid with subset sizes $s\in\{20,50,100,200,500,1000\}$ and repetitions $R\in\{5,10,20\}$. We consider two signal structures:

\smallskip
\noindent\emph{Sparse change}: Only $20$ coordinates change: for $j=1,\dots,20$ we add a mean shift of size $1.5$ after time $\nu$, while the remaining coordinates are unchanged.
We use max-aggregation across sketches, i.e., we select the sketch with the largest scan statistic and report its estimated location.
The exact (non-sketched) method attains accuracy $1.00$ with average runtime $0.78$s.
Sketching yields substantial speed-ups but may lose accuracy when most subsets miss the signal:
for example, $(s,R)=(50,20)$ achieves accuracy $0.81$ with a $3.9\times$ speed-up (0.20s),
while $(s,R)=(100,20)$ improves accuracy to $0.86$ with a $2.9\times$ speed-up (0.27s).
Very aggressive sketching such as $(s,R)=(20,5)$ is fast ($17.6\times$) but much less accurate (0.31).

\smallskip
\noindent \emph{Diffuse change}: All $p$ coordinates shift by a small amount $0.12$ after time $\nu$.
We use mean-aggregation of the $R$ location estimates (rounded to an integer), which stabilizes the estimator when every subset carries signal.
In this regime, moderate sketches match the exact accuracy with clear computational gains:
for example, $(s,R)=(200,5)$ attains accuracy $0.99$ with an $8.0\times$ speed-up (0.10s),
and $(s,R)=(500,5)$ achieves accuracy $1.00$ with a $3.9\times$ speed-up (0.20s).
Overall, sketching is markedly more effective for diffuse changes than for sparse changes, and the aggregation rule
(max vs.\ mean) is critical for approaching the best speed--accuracy frontier.
}
\end{example}

\begin{figure}[!ht]
	\centering
	\includegraphics[width=0.95\textwidth]{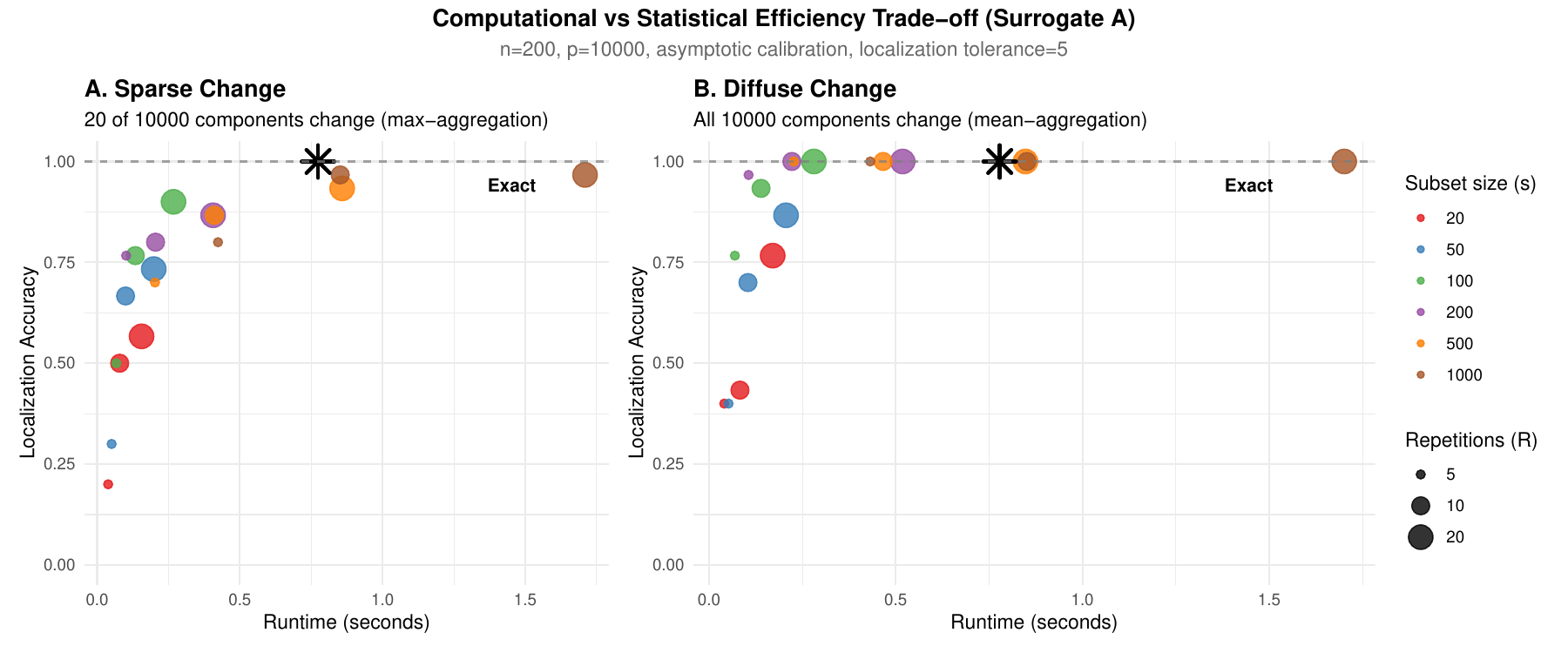}
	\caption{Trade-off between computational efficiency (runtime) and statistical efficiency (localization accuracy $\mathbb{P}(|\widehat{\nu}-\nu|\le 5)$) for Surrogate~A (coordinate sketching) with $n=200$ and $p=10{,}000$ over 100 replications.
	Left: sparse change with max-aggregation across sketches.
	Right: diffuse change with mean-aggregation across sketches.
	Each point corresponds to a sketch configuration $(s,R)$, while the exact method (no sketching, denoted by the black asterisk) serves as the baseline.}
	\label{fig:surrogate_tradeoff}
\end{figure}

Finally, we consider the following example to demonstrate the advantage of the 
componentwise monotone-invariant procedure.

\begin{example}\label{eg:heavy_tails_outliers}
{\rm 
We consider a three-segment mean-change model with two change-points at
$\nu_1 = n/3$ and $\nu_2 = 2n/3$ (so $\nu_1=50$ and $\nu_2=100$ when $n=150$).
The dimension is $p=100$, and only the first $p/2$ components are affected by the change.
Specifically, for $j=1,\dots,p/2$ we add a mean shift of magnitude $\delta=0.4$ in the middle
segment $(\nu_1,\nu_2]$, while the remaining components $j>p/2$ remain unchanged.
We run $50$ Monte Carlo replications per setting and evaluate segmentation accuracy using the ARI. We compare:
(i) the componentwise monotone-invariant (rank-based) version (KDist-MI), obtained by transforming
each coordinate to mid-rank pseudo-observations $U_{t,j}=(\mathrm{rank}(X_{t,j})-0.5)/n$, \rev{where tied values receive their average rank,} and then
applying \texttt{kcpd\_sbs} with $\gamma(z,z')=\|z-z'\|_1^{1/2}$, and
(ii) the corresponding non-rank version (KDist) that applies the same procedure directly to the raw data.

\smallskip
\noindent\emph{Setting A (heavy tails).}
We generate independent coordinates from a Student-$t$ distribution with degrees of freedom
$\mathrm{df}\in\{1,1.5,2,3,5,10,30,\infty\}$ and impose the above mean shift on the first $p/2$
components in the middle segment.

\smallskip
\noindent\emph{Setting B ($\varepsilon$-contamination).}
We generate clean data with i.i.d.\ $N(0,1)$ entries and impose the same mean shift in the middle
segment for $j\le p/2$. We then contaminate an $\varepsilon$ fraction of the $np$ entries by
randomly selecting $\varepsilon np$ cells and replacing them by outliers of magnitude $\pm 10$
(with random signs), for $\varepsilon\in\{0,0.02,0.05,0.08,0.10,0.15,0.20,0.25,0.30\}$.

\smallskip
As shown in Figure~\ref{fig:heavy_tails_outliers}, the KDist-MI procedure is substantially
more robust under heavy tails and outlier contamination. In Setting A, KDist-MI improves markedly over
the non-rank version for $\mathrm{df}\le 2$ (e.g., at $\mathrm{df}=1$ (Cauchy) the average ARI is
$0.456$ for KDist-MI versus $0.261$ for the non-rank method), while the two procedures become comparable
as the tails become lighter. In Setting B, KDist-MI maintains strong performance up to moderate
contamination levels (e.g., ARI $\approx 0.88$ at $10\%$ contamination and $0.889$ at $15\%$),
whereas the non-rank method degrades rapidly (e.g., ARI $0.490$ at $10\%$ and $0.235$ at $15\%$).
}
\end{example}

\begin{figure}[!ht]
	\centering
	\includegraphics[width=0.95\textwidth]{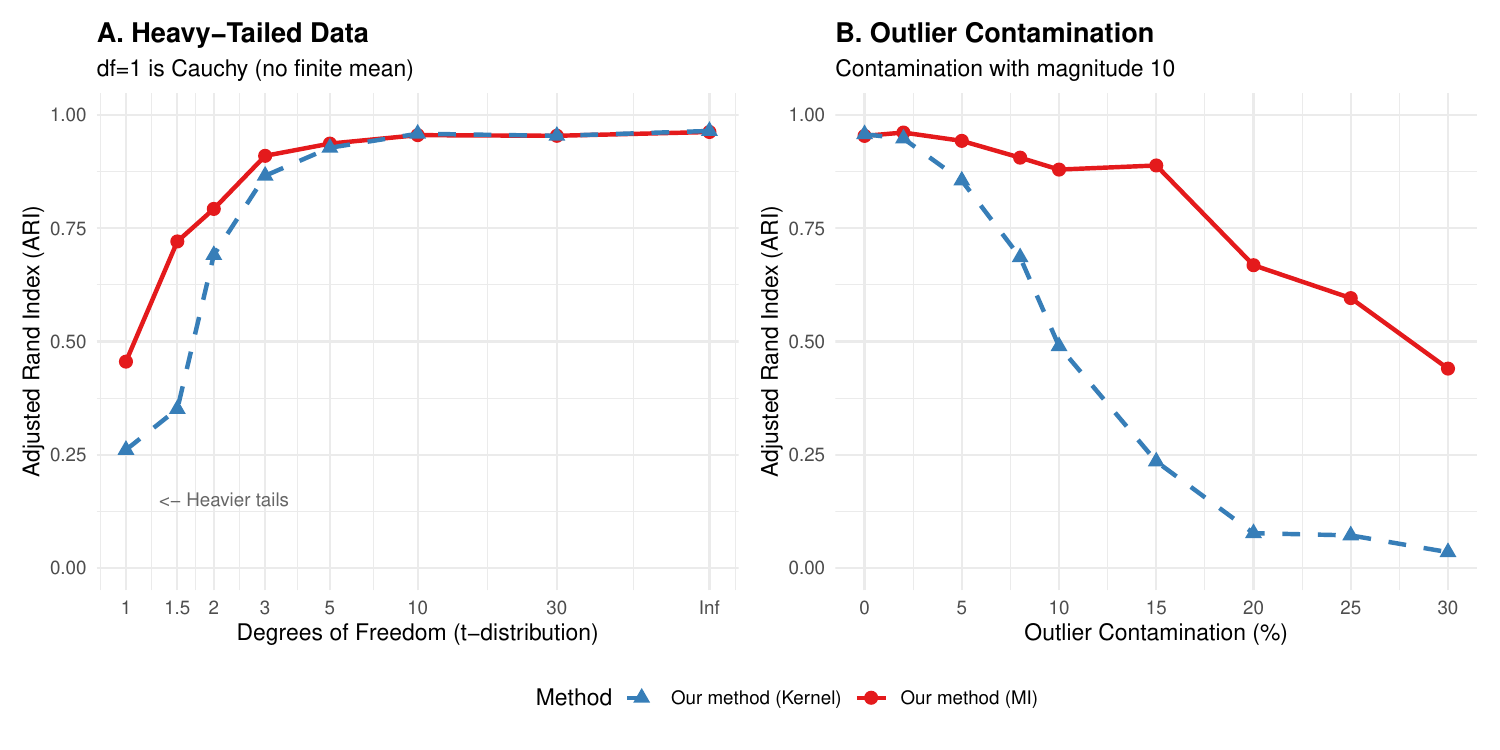}
	\caption{Robustness to heavy tails and outliers. Panel A: average ARI versus degrees of freedom
	for Student-$t$ data ($\mathrm{df}=1$ is Cauchy; $\mathrm{df}=\infty$ is Gaussian). Panel B:
	average ARI versus the fraction of contaminated entries replaced by $\pm 10$.}
	\label{fig:heavy_tails_outliers}
\end{figure}

\subsection{Real data illustration}\label{real data}
We analyze the daily closing stock prices of $p=72$ companies in the Consumer Defensive sector, listed on the NYSE and NASDAQ exchanges. The data consists of observations on the first trading day of each month from January 1, 2005, to December 31, 2010, obtained from Yahoo Finance via the R package \texttt{quantmod}. Let $X_t = (X_{t,1}, \dots, X_{t,p})$ denote the closing prices at time $t$ for $1\leq t \leq 72$. We perform the analysis on the log-returns, defined as $S^X_t = (S^X_{t,1}, \dots, S^X_{t,p})$ where $S^X_{t,i} = \log (X_{t+1,i}/X_{t,i})$, resulting in a sample size of $n=71$ with dimension $p=72$.

This period encompasses the Global Financial Crisis. According to the U.S. National Bureau of Economic Research (NBER), the recession officially began in December 2007 and ended in June 2009, lasting 19 months. Key events included the onset of the liquidity crisis in late 2007, the unprecedented \$700 billion bank bailout in October 2008, and the \$787 billion fiscal stimulus package in February 2009. While the Consumer Defensive sector is traditionally considered resilient to economic downturns compared to sectors like Finance or Real Estate, the magnitude of this crisis suggests that structural breaks should be detectable even in these stable assets.

We apply our proposed methodology using the $L_1$-based metric $\gamma(x,x')=\|x-x'\|^{1/2}_1$ and compare it with several state-of-the-art methods. The results are summarized below and visualized in Figure \ref{fig_real}:

\begin{itemize}
	\item \textbf{Proposed Method:} Both KDist and KDist-MI detect the same two change-points: \textbf{October 1, 2007}, and \textbf{February 1, 2009}. These dates align closely with the onset of the recession (preceding the official December start by a quarter, reflecting market anticipation) and the rollout of major fiscal stimulus measures in early 2009.
	
	\item \textbf{Matteson and James (2014):} Notably, the E-Divisive procedure fails to detect \emph{any} change-points during this period. This null result highlights the limitation of standard Euclidean energy distance in high-dimensional settings ($p \approx n$), as predicted by our theoretical analysis.
	
	\item \textbf{Chen and Zhang (2015):} The graph-based original scan statistic detects a single change-point on \textbf{March 1, 2009}, coinciding with the market bottom and the stimulus package.
	
	\item \textbf{Chu and Chen (2019):} The max-type edge-count test identifies two change-points: \textbf{May 1, 2008}, and \textbf{September 1, 2008} (around the collapse of Lehman Brothers).
	
	\item \textbf{Wang and Samworth (2018):} This methodology detects 18 change-points. This high number likely indicates an over-segmentation or high false positive rate in this medium-sample-size regime.
\end{itemize}

Overall, our method provides a parsimonious and interpretable segmentation that captures the beginning and the turning point of the crisis, whereas competitors either miss the signal entirely or produce fragmented segmentations.

\begin{figure}[!ht]
	\centering
	\includegraphics[width=0.7\textwidth]{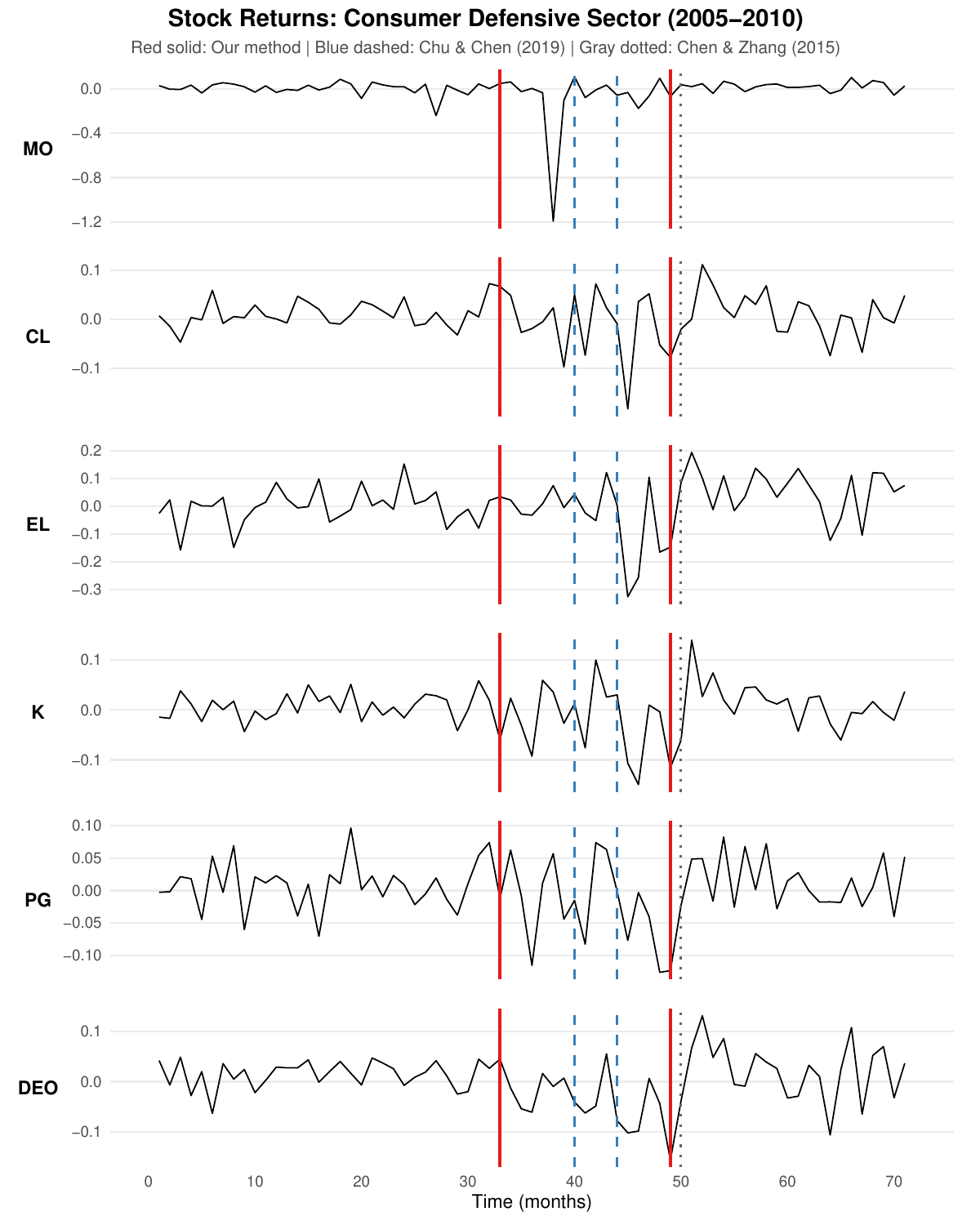}
	\caption{Time series plots of stock returns for six representative companies in the Consumer Defensive sector. The solid red lines indicate change-points detected by our proposed methodology. The dotted blue and gray lines represent change-points detected by \cite{CC} and \cite{CZ}, respectively.}
	\label{fig_real}
\end{figure}

\section{Incorporating graph information}\label{sec:graph_info}
Two critical questions remain regarding the generalized Euclidean distance defined in Definition \ref{def_GED}: first, how to perform the partitioning or grouping optimally in practice, and second, whether it is possible to completely characterize the homogeneity between two high-dimensional random vectors using these partitions. We present two examples below where external undirected or directed graph information is available to guide the partition. In both cases, the corresponding generalized energy distance completely characterizes the homogeneity between two high-dimensional random vectors. By incorporating this structural information, our change-point detection procedure can achieve higher statistical efficiency.

\subsection{Undirected graph parameterized by exponential family}
Suppose the distribution of $X=(x_1,\dots,x_p)$ belongs to an exponential family of the form
\begin{align*}
	\exp\left\{\sum_{C\in\mathcal{I}}\theta_C\phi_C(\widetilde{x}_C)-A(\theta)\right\}, 
\end{align*}
where $\mathcal{I}$ denotes a collection of subsets of $[p]=\{1,2,\dots,p\}$, $\widetilde{x}=(\widetilde{x}_1,\dots,\widetilde{x}_p)$ are sufficient statistics, and $\theta=\{\theta_C:C\in\mathcal{I}\}$ are canonical parameters. A special case is the pairwise graphical model, where
\begin{align*}
	\sum_{C\in\mathcal{I}}\theta_C\phi_C(\widetilde{x}_C)=
	\sum_{i\in [p]}\theta_{i}\phi_{i}(\widetilde{x}_{i})+\sum_{(i,j)\in E}\theta_{ij}\phi_{ij}(\widetilde{x}_{ij}),
\end{align*}
with $\widetilde{x}_{ij}=(\widetilde{x}_i,\widetilde{x}_j)$ and $E\subseteq [p]\times [p]$ denoting the set of edges. Examples include Gaussian graphical models and Ising models. We assume the minimal representation condition holds, meaning there does not exist a non-zero vector $\{\beta_C:C\in\mathcal{I}\}$ such that $\sum_{C\in\mathcal{I}}\beta_C\phi_C(\widetilde{x}_C)$ is constant.

Let $\mu_{C}=\mathbb{E}_{\theta}[\phi_{C}(x_{C})]$ be the mean parameter. By standard exponential family theory \citep[see e.g., Proposition 3.2 in][]{WJ}, the gradient map $\nabla A: \theta \mapsto \{\mu_C\}_{C\in\mathcal{I}}$ is one-to-one under the minimal representation condition. The mean parameters are determined by the set of marginal distributions $\{p_C:C\in\mathcal{I}\}$ via $\mu_{C}=\int \phi_{C}(\widetilde{x}_{C}) p_{C}(\widetilde{x}_{C})d\widetilde{x}_{C}$. Consequently, the collection of marginal distributions $\{p_C:C\in\mathcal{I}\}$ completely determines the full joint distribution of $X$. Thus, for two such random vectors $X$ and $Y$:
$$X \overset{d}{=} Y \quad \text{if and only if} \quad x_C\overset{d}{=} y_C\text{ for all } C\in\mathcal{I}.$$ 
To leverage this, we construct the graph-guided generalized Euclidean distance as:
\[
\gamma(z,z')=\sqrt{\sum_{C\in\mathcal{I}}\rho_C(z_C,z_C')}.    
\]
We illustrate the utility of this strategy with a toy example.

\begin{example}[Fully visible Boltzmann machine]\label{eg:fvbm}
	{\rm 
		Let $X=(X_1,\dots,X_p)$ be a $p$-variate binary random vector with $X_i \in \{-1,1\}$. Suppose the probability mass function is given by 
		$$f(\widetilde{x}\,;\, b,M) = \frac{1}{Z(b,M)} \exp \left(\frac{1}{2} \widetilde{x}^\top M \widetilde{x} + b^\top\widetilde{x} \right),$$ 
        where $M$ is a symmetric $p\times p$ matrix with zero diagonal entries, $b \in \mathbb{R}^p$, and $Z(b,M)$ is the partition function. This is known as a fully visible Boltzmann machine (FVBM; \cite{nguyen}).
        
        We generate a sequence where for $1\leq t \leq \lfloor n/2 \rfloor$, $X_t$ follows an FVBM with $b=0.1\times {\bf 1}_p$, $M(a,b) = 0.1$ for $|a-b|=1$ and $0$ otherwise. For $\lfloor n/2 \rfloor +1 \leq t \leq n$, $X_t$ follows an FVBM with $b=0.5\times {\bf 1}_p$, $M(a,b) = 0.3$ for $|a-b|=1$ and $0$ otherwise. 
		
	    The clique set is $\mathcal{I}=\{ \{1,2\}, \{2,3\}, \dots, \{p-1,p\} \}$. We consider $n=50$ and $p=25$, using the R package \texttt{BoltzMM} for data generation. We implement our test using \rev{$B_{\mathrm{perm}}=199$} permutations at $\alpha=0.05$. The metric $\gamma$ is defined via the partition induced by $\mathcal{I}$. Table \ref{table:ug} reports the average Adjusted Rand Index (ARI) over $100$ simulations. Despite $p<n$, our graph-aware test significantly outperforms the standard energy distance test \citep{MJ} and other competitors.

    	\begin{table}[h]
    	\centering
    	\caption{Comparison of average ARI values over 100 simulations (FVBM).}
    	\label{table:ug}
    	\begin{tabular}{ccccc}
    		\toprule
    		\textbf{KDist} & MJ & CC & CZ & WS \\
    		\hline
    		0.974 & 0.000 & 0.311 & 0.712 & 0.931  \\
    		\bottomrule
    	\end{tabular}
    \end{table}
	}
\end{example}

\subsection{Directed graph/Bayesian networks}
Consider a Bayesian network, where the joint distribution factorizes over a directed acyclic graph (DAG) as $P(X)=\prod_{i=1}^p P(x_i|x_{\pi(i)})$, with $\pi(i)$ denoting the parent set of node $i$. Two distributions $X$ and $Y$ obeying the same DAG structure are identical if and only if their conditional distributions are identical, which implies:
$$X \overset{d}{=} Y \quad \text{if and only if} \quad x_{i\,\cup\, \pi(i)}\overset{d}{=}y_{i\,\cup\, \pi(i)}\text{ for all } i=1,\dots,p.$$ 
Accordingly, we define the DAG-guided distance:
\begin{align*}
	\gamma(z,z')=\sqrt{\sum_{i=1}^p \rho_i(z_{i\,\cup\, \pi(i)},z_{i\,\cup\, \pi(i)}')}.    
\end{align*}

\begin{example}[Directed Chain]
	{\rm 
	We consider an autoregressive model $X_{i} = \phi X_{i-1} + \epsilon_{i}$ for $2 \leq i \leq p$, with $X_{1}=\epsilon_{1}$. This corresponds to a chain graph where $\pi(i) = \{i-1\}$. We set $n=100,\, p=100$ and $\phi=0.5$. The error terms switch distributions at $t = \lfloor n/2 \rfloor$: pre-change $\epsilon_{t} \sim N(\mathbf{1}_p, \mathbf{I}_p)$, and post-change $\epsilon_{t,i} \sim \text{Exponential}(1)$.
	
    The partition used is $(\{X_1\}, \{X_1,X_2\}, \dots, \{X_{p-1}, X_p\})$. Table \ref{table:dg} shows that our method yields superior localization accuracy.

    \begin{table}[h]
    	\centering
    	\caption{Comparison of average ARI values over 100 simulations (Directed Chain).}
    	\label{table:dg}
    	\begin{tabular}{ccccc}
    		\toprule
    		\textbf{KDist} & MJ & CC & CZ & WS \\
    		\hline
    		0.949 & 0.050 & 0.010 & 0.070 & 0.376   \\
    		\bottomrule
    	\end{tabular}
    \end{table}		
	}
\end{example}

\section{Future directions}\label{sec:conclusion}
We have developed a non-parametric framework for change-point detection in high-dimensional data that is sensitive to changes beyond the first two moments. By leveraging a generalized energy distance and the Seeded NOT, we established consistency for estimating multiple change-points. We further introduced a componentwise monotone-invariant \rev{heuristic} extension to address robustness and invariance concerns.

Several avenues for future research remain. First, extending our methodology to weakly dependent high-dimensional time series is of significant practical interest, though it introduces theoretical challenges regarding the convergence of the empirical process. Second, while we discussed incorporating \emph{known} graph structures, in many applications the graph is unknown. Integrating structure learning (e.g., estimating the DAG or partial correlation graph) simultaneously with change-point detection would be a powerful, albeit computationally demanding, extension. Finally, exploring the optimal choice of the sketching dimension $s$ or the subsampling rate \rev{$N_{\mathrm{pair}}$} for our computational surrogates in specific regimes remains an open question for optimizing the trade-off between statistical power and computational efficiency.


\section{Supplementary information} The Supplementary Materials contain rigorous proofs of the theoretical results presented in the paper.



\bibliography{reference}

@article{baranowski2019narrowest,
  title={Narrowest-over-threshold detection of multiple change points and change-point-like features},
  author={Baranowski, Rafal and Chen, Yining and Fryzlewicz, Piotr},
  journal={Journal of the Royal Statistical Society Series B: Statistical Methodology},
  volume={81},
  number={3},
  pages={649--672},
  year={2019},
  publisher={Oxford University Press}
}

@article{arlot,
  author = {Arlot, S. and Celisse, A. and Harchaoui, Z.},
  title = {A Kernel Multiple Change-point Algorithm via Model Selection},
  journal = {Journal of Machine Learning Research},
  volume = {20},
  pages = {1--56},
  year = {2019}
}

@article{aue,
  author = {Aue, A. and Horv{\'a}th, L.},
  title = {Structural breaks in time series},
  journal = {Journal of Time Series Analysis},
  volume = {34},
  number = {1},
  pages = {1--16},
  year = {2012}
}

@article{avanesov,
  author = {Avanesov, V. and Buzun, N.},
  title = {Change-point detection in high-dimensional covariance structure},
  journal = {Electronic Journal of Statistics},
  volume = {12},
  number = {2},
  pages = {3254--3294},
  year = {2018}
}

@article{bf,
  author = {Baringhaus, L. and Franz, C.},
  title = {On a New Multivariate Two-sample Test},
  journal = {Journal of Multivariate Analysis},
  volume = {88},
  number = {1},
  pages = {190--206},
  year = {2004}
}

@article{biau,
  author = {Biau, G. and Bleakley, K. and Mason, D.M.},
  title = {Long Signal Change-point Detection},
  journal = {Electronic Journal of Statistics},
  volume = {10},
  number = {2},
  pages = {2097--2123},
  year = {2016}
}

@article{carlstein,
  author = {Carlstein, E.},
  title = {Nonparametric change-point estimation},
  journal = {The Annals of Statistics},
  volume = {16},
  number = {1},
  pages = {188--197},
  year = {1988}
}

@article{us,
  author = {Chakraborty, S. and Zhang, X.},
  title = {A New Framework for Distance and Kernel-based Metrics in High Dimensions},
  journal = {Electronic Journal of Statistics},
  volume = {15},
  number = {2},
  pages = {5455--5522},
  year = {2021}
}

@article{CZ,
  author = {Chen, H. and Zhang, N.},
  title = {Graph-based Change-point Detection},
  journal = {Annals of Statistics},
  volume = {43},
  number = {1},
  pages = {139--176},
  year = {2015}
}

@article{CF,
  author = {Cho, H. and Fryzlewicz, P.},
  title = {Multiple-change-point detection for high dimensional time series via sparsified binary segmentation},
  journal = {Journal of the Royal Statistical Society, Series B},
  volume = {77},
  number = {2},
  pages = {475--507},
  year = {2015}
}

@article{CC,
  author = {Chu, L. and Chen, H.},
  title = {Asymptotic Distribution-free Change-point Detection for Multivariate and Non-Euclidean Data},
  journal = {Annals of Statistics},
  volume = {47},
  number = {1},
  pages = {382--414},
  year = {2019}
}

@article{curtis,
  author = {Curtis, R. and Xiang, J. and Parikh, A. and Kinnaird, P. and Xing, E.P.},
  title = {Enabling Dynamic Network Analysis Through Visualization in TVNViewer},
  journal = {BMC Bioinformatics},
  volume = {13},
  number = {204},
  year = {2012}
}

@article{dette,
  author = {Dette, H. and Pan, G. and Yang, Q.},
  title = {Estimating a Change Point in a Sequence of Very High-Dimensional Covariance Matrices},
  journal = {Journal of the American Statistical Association},
  volume = {115},
  number = {530},
  pages = {944--955},
  year = {2020}
}

@article{dumbgen,
  author = {D{\"u}mbgen, L.},
  title = {The asymptotic behavior of some nonparametric change-point estimators},
  journal = {The Annals of Statistics},
  volume = {19},
  number = {3},
  pages = {1471--1495},
  year = {1991}
}

@article{enik,
  author = {Enikeeva, F. and Harchaoui, Z.},
  title = {High-dimensional change-point detection under sparse alternatives},
  journal = {Annals of Statistics},
  volume = {47},
  number = {4},
  pages = {2051--2079},
  year = {2019}
}

@inproceedings{hc,
  author = {Harchaoui, Z. and Capp{\'e}, O.},
  title = {Retrospective Change-point Estimation with Kernels},
  booktitle = {IEEE Workshop on Statistical Signal Processing},
  year = {2007}
}

@article{jan,
  author = {Jandhyala, V. and Fotopoulos, S. and MacNeill, I. and Liu, P.},
  title = {Inference for single and multiple change‐points in time series},
  journal = {Journal of Time Series Analysis},
  volume = {34},
  number = {4},
  pages = {423--446},
  year = {2013}
}

@article{jirak,
  author = {Jirak, M.},
  title = {Uniform Change Point Tests in High Dimension},
  journal = {Annals of Statistics},
  volume = {43},
  number = {6},
  pages = {2451--2483},
  year = {2015}
}

@article{kaul2023inference,
  author = {Kaul, A. and Zhang, H. and Tsampourakis, K.},
  title = {Inference on the Change Point under a High Dimensional Covariance Shift},
  journal = {Journal of Machine Learning Research},
  volume = {24},
  number = {168},
  pages = {1--68},
  year = {2023}
}

@article{k2023,
  author = {Kov{\'a}cs, S. and B{\"u}hlmann, P. and Li, H. and Munk, A.},
  title = {Seeded Binary Segmentation: A General Methodology for Fast and Optimal Changepoint Detection},
  journal = {Biometrika},
  volume = {110},
  number = {1},
  pages = {249--256},
  year = {2023}
}

@article{li2023online,
  author = {Li, L. and Li, J.},
  title = {Online Change-Point Detection in High-Dimensional Covariance Structure with Application to Dynamic Networks},
  journal = {Journal of Machine Learning Research},
  volume = {24},
  number = {51},
  pages = {1--44},
  year = {2023}
}

@article{lung,
  author = {Lung-Yut-Fong, A. and L{\'e}vy-Leduc, C. and Capp{\'e}, O.},
  title = {Homogeneity and Change-point Detection Tests for Multivariate Data Using Rank Statistics},
  journal = {Journal de la Soci{\'e}t{\'e} Francaise de Statistique},
  volume = {156},
  number = {4},
  pages = {133--162},
  year = {2015}
}

@article{lyons,
  author = {Lyons, R.},
  title = {Distance Covariance in Metric Spaces},
  journal = {Annals of Probability},
  volume = {41},
  number = {5},
  pages = {3284--3305},
  year = {2013}
}

@article{MJ,
  author = {Matteson, D.S. and James, N.A.},
  title = {A Nonparametric Approach for Multiple Change Point Analysis of Multivariate Data},
  journal = {Journal of the American Statistical Association},
  volume = {109},
  number = {505},
  pages = {334--345},
  year = {2014}
}

@phdthesis{mcculloh,
  author = {McCulloh, I.},
  title = {Detecting Changes in a Dynamic Social Network},
  school = {Institute for Software Research, School of Computer Science, Carnegie Mellon University},
  year = {2009},
  note = {CMU-ISR-09-104}
}

@article{MA,
  author = {Morey, L.C. and Agresti, A.},
  title = {The Measurement of Classification Agreement: An Adjustment to the Rand Statistic for Chance Agreement},
  journal = {Educational and Psychological Measurement},
  volume = {44},
  number = {1},
  pages = {33--37},
  year = {1984}
}

@article{nguyen,
  author = {Nguyen, H.D. and Wood, I.A.},
  title = {Asymptotic Normality of the Maximum Pseudolikelihood Estimator for Fully Visible Boltzmann Machines},
  journal = {IEEE Transactions on Neural Networks and Learning Systems},
  volume = {27},
  number = {4},
  pages = {897--902},
  year = {2016}
}

@inproceedings{park,
  author = {Park, Y. and Wang, H. and Napp{\"o}bauer, T. and Vaziri, A. and Priebe, C.E.},
  title = {Anomaly Detection on Whole-Brain Functional Imaging of Neuronal Activity using Graph Scan Statistics},
  booktitle = {ACM Conference on Knowledge Discovery and Data Mining (KDD), Workshop on Outlier Definition, Detection, and Description (ODDx3)},
  year = {2015}
}

@article{picard,
  author = {Picard, F. and Robin, S. and Lavielle, M. and Vaisse, C. and Daudin, J.-J.},
  title = {A Statistical Approach for Array CGH Data Analysis},
  journal = {BMC Bioinformatics},
  volume = {6},
  number = {27},
  year = {2005}
}

@book{pollard,
  author = {Pollard, D.},
  title = {Empirical Processes: Theory and Applications},
  series = {NSF-CBMS Regional Conference Series in Probability and Statistics},
  volume = {2},
  pages = {i--iii+v+vii--viii+1--86},
  year = {1990}
}

@article{rabiner,
  author = {Rabiner, L.R. and Sch{\"a}fer, R.W.},
  title = {Introduction to Digital Speech Processing},
  journal = {Foundations and Trends in Signal Processing},
  volume = {1},
  number = {1-2},
  pages = {1--194},
  year = {2007}
}

@article{ryan2023detecting,
  author = {Ryan, S. and Killick, R.},
  title = {Detecting changes in covariance via random matrix theory},
  journal = {Technometrics},
  volume = {65},
  number = {4},
  pages = {480--491},
  year = {2023}
}

@article{ssgf,
  author = {Sejdinovic, D. and Sriperumbudur, B. and Gretton, A. and Fukumizu, K.},
  title = {Equivalence of Distance-based and RKHS-based Statistics in Hypothesis Testing},
  journal = {Annals of Statistics},
  volume = {41},
  number = {5},
  pages = {2263--2291},
  year = {2013}
}

@article{sr2004,
  author = {Sz{\'e}kely, G.J. and Rizzo, M.L.},
  title = {Testing for Equal Distributions in High Dimension},
  journal = {InterStat},
  volume = {5},
  year = {2004}
}

@article{sr2005,
  author = {Sz{\'e}kely, G.J. and Rizzo, M.L.},
  title = {Hierarchical Clustering via Joint Between-within Distances: Extending Ward's Minimum Variance Method},
  journal = {Journal of Classification},
  volume = {22},
  number = {2},
  pages = {151--183},
  year = {2005}
}

@article{T2020,
  author = {Truong, C. and Oudre, L. and Vayatis, N.},
  title = {Selective Review of Offline Change Point Detection Methods},
  journal = {Signal Processing},
  volume = {167},
  pages = {107299},
  year = {2020}
}

@article{wangfeng2023,
  author = {Wang, G. and Feng, L.},
  title = {Computationally efficient and data-adaptive changepoint inference in high dimension},
  journal = {Journal of the Royal Statistical Society Series B: Statistical Methodology},
  volume = {85},
  number = {3},
  pages = {936--958},
  year = {2023}
}

@article{wang2021optimal,
  author = {Wang, D. and Yu, Y. and Rinaldo, A.},
  title = {Optimal covariance change point localization in high dimensions},
  journal = {Bernoulli},
  volume = {27},
  number = {1},
  pages = {554--575},
  year = {2021}
}

@article{WS,
  author = {Wang, T. and Samworth, R.J.},
  title = {High Dimensional Change Point Estimation via Sparse Projection},
  journal = {Journal of the Royal Statistical Society, Series B},
  volume = {80},
  number = {1},
  pages = {57--83},
  year = {2018}
}

@article{WVS,
  author = {Wang, R. and Zhu, C. and Volgushev, S. and Shao, X.},
  title = {Inference for Change Points in High Dimensional Data via Self-Normalization},
  journal = {Annals of Statistics},
  volume = {50},
  number = {2},
  pages = {781--806},
  year = {2022}
}

@book{WJ,
  author = {Wainwright, M.J. and Jordan, M.I.},
  title = {Graphical Models, Exponential Families, and Variational Inference},
  publisher = {Now Publishers Inc.},
  year = {2008}
}

@article{yz,
  author = {Yan, J. and Zhang, X.},
  title = {Kernel Two-sample Tests in High Dimensions: Interplay Between Moment Discrepancy and Dimension-and-Sample Orders},
  journal = {Biometrika},
  volume = {110},
  number = {2},
  pages = {411--430},
  year = {2022}
}

@article{yu2021,
  author = {Yu, M. and Chen, X.},
  title = {Finite sample change point inference and identification for high-dimensional mean vectors},
  journal = {Journal of the Royal Statistical Society Series B: Statistical Methodology},
  volume = {83},
  number = {2},
  pages = {247--270},
  year = {2021}
}

@article{zhang2022,
  author = {Zhang, Y. and Wang, R. and Shao, X.},
  title = {Adaptive inference for change points in high-dimensional data},
  journal = {Journal of the American Statistical Association},
  volume = {117},
  number = {540},
  pages = {1751--1762},
  year = {2022}
}

@article{zys2018,
  author = {Zhang, X. and Yao, S. and Shao, X.},
  title = {Conditional Mean and Quantile Dependence Testing in High Dimension},
  journal = {Annals of Statistics},
  volume = {46},
  number = {1},
  pages = {219--246},
  year = {2018}
}

@article{zs2019,
  author = {Zhu, C. and Shao, X.},
  title = {Interpoint Distance Based Two Sample Tests in High Dimension},
  journal = {Bernoulli},
  volume = {27},
  number = {2},
  pages = {1189--1211},
  year = {2021}
}

@article{zou,
  author = {Zou, C. and Yin, G. and Feng, L. and Wang, Z.},
  title = {Nonparametric Maximum Likelihood Approach to Multiple Change-point Problems},
  journal = {Annals of Statistics},
  volume = {42},
  number = {3},
  pages = {970--1002},
  year = {2014}
}

\newpage

\begin{alphasection}

\begin{center}
{\bf \Large Supplementary Materials for ``High-dimensional Change-point Detection Using Generalized Homogeneity Metrics"}
\end{center}

\vspace{0.3cm}

The Supplementary Materials are organized as follows.
Section~\ref{proof_theorem1} contains the proof of Theorem~\ref{theorem1} (null weak convergence) and auxiliary results used in the null analysis. Section~\ref{proof_alternative} \rev{provides the proof of} Theorem~\ref{alt:fix_alternative}, the consistency of the test under alternatives. Sections~\ref{proof_consistency} and \ref{sec:theory_not} establish the single and multiple change-point consistency results for the NOT procedures. \rev{Section~\ref{sec:theory_MI} provides heuristic considerations, rather than a theorem, for the componentwise monotone-invariant procedure.} Additional proofs of secondary results and further technical lemmas are collected in Section~\ref{tech appx}.

\section{Proof of Theorem \ref{theorem1}}\label{proof_theorem1}
For the ease of notation, we write $\widehat{E}_{n,k}=\widehat{E}_{\gamma}(\mathbf{X}_{1:k},\mathbf{X}_{k+1:n})$, $\widehat{\cal{D}}^2_{1:k}=\widehat{\cal{D}}^2(\mathbf{X}_{1:k})$, $\widehat{\cal{D}}^2_{k+1:n}=\widehat{\cal{D}}^2(\mathbf{X}_{k+1:n})$ and $\widehat{\cal{C}}_{1,k,n}=\widehat{\mathcal{C}}(\mathbf{X}_{1:k},\mathbf{Y}_{k+1:n}).$ From the proof of Lemma D.1 in the Supplementary Materials of \cite{us}, we can write under $H_0$, 
\begin{align}\label{eq_sup_1}
\widehat{E}_{n,k}=L_{n,k}+R_{n,k}, 
\end{align}
where
\begin{align}\label{eq_sup_2}
\begin{split}
L_{n,k} \;=&\; \frac{1}{k (n-k)}\sum_{i_1=1}^{k}\sum_{i_2=k+1}^{n} H(X_{i_1}, X_{i_2})\,-\,\frac{1}{k(k-1)}\sum_{1 \leq i_1< i_2 \leq k} H(X_{i_1},X_{i_2}) 
\\&-\,\frac{1}{(n-k)(n-k-1)}\sum_{k+1 \leq i_1< i_2 \leq n} H(X_{i_1},X_{i_1})\,,\\
R_{n,k} \;=&\; \frac{2\tau}{k (n-k)}\sum_{i_1=1}^{k}\sum_{i_2=k+1}^{n} R(X_{i_1}, X_{i_2})\,-\,\frac{\tau}{k(k-1)}\sum_{1 \leq i_1\neq i_2 \leq k} R(X_{i_1},X_{i_2}) 
\\&-\,\frac{\tau}{(n-k)(n-k-1)}\sum_{k+1 \leq i_1\neq i_2 \leq n} R(X_{i_1},X_{i_1})\,.
\end{split}
\end{align}
Following the discussions in Section D in the Supplementary Materials of \cite{us}, the variance of $L_{n,k}$ is given by
\begin{align}\label{eq_sup_0}
V_{n,k} \;:=a_{k,n-k}^2\E\,[H^2(X,X')] \,,
\end{align}
which can be estimated by 
\begin{align}\label{eq_sup_0.5}
\begin{split}
\widehat{V}_{n,k}=\frac{4\widehat{\cal{C}}_{1,k,n}}{k(n-k)} + \frac{2\widehat{\cal{D}}^2_{1:k}}{k(k-1)} + \frac{2\widehat{\cal{D}}^2_{k+1:n}}{(n-k)(n-k-1)}.
\end{split}
\end{align}
Define 
\begin{align}\label{eq_sup_0.6}
\breve{T}_{1,n,k} \;=\;\frac{\widehat{E}_{n,k}}{\sqrt{V_{n,k}}}.
\end{align}
For $1\leq l< k < m-1 \leq n-1$, define \,$\widetilde{S}_n(k,m) := \sum_{i_2=k+1}^{m}\sum_{i_1=k}^{i_2-1} H(X_{i_1}, X_{i_2})$\, and
\begin{align}\label{eq_sup_3}
\begin{split}
\widetilde{L}_n(k;l,m)
=&\frac{1}{(k-l+1) (m-k)}\sum_{i_2=k+1}^{m}\sum_{i_1=l}^{k} H(X_{i_1}, X_{i_2})
\\&-\frac{1}{(k-l+1)(k-l)}\sum_{l \leq i_1< i_2 \leq k} H(X_{i_1},X_{i_2})
\\&-\frac{1}{(m-k)(m-k-1)}\sum_{k+1 \leq i_1< i_2 \leq m} H(X_{i_1},X_{i_1}).
\end{split}
\end{align}
Let $\widetilde{S}_n(k,m)=0$ and $\widetilde{L}_n(k;l,m)=0$ for $k\leq l$ or $k\geq m-1$ or $m>n$. From (\ref{eq_sup_2}) and (\ref{eq_sup_3}), it is easy to see that $L_{n,k} = \widetilde{L}_n(k\,;1,n)$. With the definition of $\widetilde{S}_n(k,m)$ as above, we can write
\begin{align}\label{eq_sup_4}
\begin{split}
\widetilde{L}_n(k\,;l,m) \;=&\; -\,\frac{1}{(k-l+1)(k-l)}\,\widetilde{S}_n(l,k) \,-\, \frac{1}{(m-k)(m-k-1)}\,\widetilde{S}_n(k+1,m) \\
&  +\, \frac{1}{(k-l+1) (m-k)}\, \left\{ \widetilde{S}_n(l,m) - \widetilde{S}_n(l,k) - \widetilde{S}_n(k+1,m) \right\},
\end{split}
\end{align}
for $1\leq l< k < m-1 \leq n-1$.

Next, we define
\begin{align}\label{new_eq_1}
\begin{split}
 \widetilde{\Delta}_n(k\,;l,m) \;:=&\; (k-l+1)(m-k)\; \widetilde{L}_n(k\,;l,m) \\
 =&\; -\,\frac{(m-k)}{(k-l)}\,\widetilde{S}_n(l,k) \;-\; \frac{(k-l+1)}{(m-k-1)}\,\widetilde{S}_n(k+1,m)\; 
 \\&+\; \, \left\{ \widetilde{S}_n(l,m) - \widetilde{S}_n(l,k) - \widetilde{S}_n(k+1,m) \right\},
\end{split}
\end{align}
and
\begin{align}\label{eq_sup_6}
\widetilde{V}_n(k\,;l,m) \;&:=\; a^2_{k-l+1,m-k}\, V_0, 
\end{align}
for $1\leq l< k < m-1 \leq n-1$ and zero otherwise, where \,$V_0 := \E\,H^2(X,X')$\, and 
\begin{align}\label{new_eq_2}
\begin{split}
 a^2_{k-l+1,m-k} \;&=\; \frac{1}{(k-l+1)(m-k)} \,+\, \frac{1}{2(k-l+1)(k-l)}  \,+\, \frac{1}{2(m-k)(m-k-1)}\\
 &=\; \frac{2(k-l)(m-k-1) \,+\, (m-k)(m-k-1) \,+\, (k-l+1)(k-l)}{2(k-l+1)(k-l)(m-k)(m-k-1)}\,.
 \end{split}
\end{align}
From (\ref{eq_sup_0}) and (\ref{eq_sup_6}), it is easy to check that $V_{n,k} = \widetilde{V}_n(k\,;1,n)$. From (\ref{new_eq_1}), we can write 
\begin{align}\label{new_eq_3}
\begin{split}
&\frac{\widetilde{\Delta}_n(k\,;l,m)}{\sqrt{\widetilde{V}_n(k\,;l,m)}} \;=\; \frac{\widetilde{\Delta}_n(k\,;l,m)}{a_{k-l+1,m-k}\,\sqrt{V_0}}
\\ =&\; \left\{\frac{(k-l+1)(k-l)(m-k)(m-k-1)}{2(k-l)(m-k-1) \,+\, (m-k)(m-k-1) \,+\, (k-l+1)(k-l)}\right\}^{1/2}\\
&\times\, \Bigg[\,-\,\frac{(m-k)}{(k-l)}\,\frac{\sqrt{2}\,\widetilde{S}_n(l,k)}{\sqrt{V_0}} \;-\; \frac{(k-l+1)}{(m-k-1)}\,\frac{\sqrt{2}\,\widetilde{S}_n(k+1,m)}{\sqrt{V_0}}\\
 & \quad +\;  \, \left\{ \frac{\sqrt{2}\,\widetilde{S}_n(l,m)}{\sqrt{V_0}} - \frac{\sqrt{2}\,\widetilde{S}_n(l,k)}{\sqrt{V_0}} - \frac{\sqrt{2}\,\widetilde{S}_n(k+1,m)}{\sqrt{V_0}} \right\}\Bigg]\,.
\end{split}
\end{align}
Dividing both sides by $n^2$,\, we get

\begin{align}\label{new_eq_4}
\begin{split}
\frac{\widetilde{\Delta}_n(k\,;l,m)}{n^2\,\sqrt{\widetilde{V}_n(k\,;l,m)}}\;&=\;  \left\{\frac{ \frac{(k-l+1)(k-l)(m-k)(m-k-1)}{n^4} }{ \frac{2(k-l)(m-k-1)}{n^2} \,+\, \frac{(m-k)(m-k-1)}{n^2} \,+\, \frac{(k-l+1)(k-l)}{n^2} }\right\}^{1/2}\\
&\qquad \times\, \Bigg[\,-\,\frac{(m-k)}{(k-l)}\,\frac{\sqrt{2}\,\widetilde{S}_n(l,k)}{n\,\sqrt{V_0}} \;-\; \frac{(k-l+1)}{(m-k-1)}\,\frac{\sqrt{2}\,\widetilde{S}_n(k+1,m)}{n\,\sqrt{V_0}}\\
 & \qquad \qquad +\; \, \left\{ \frac{\sqrt{2}\,\widetilde{S}_n(l,m)}{n\,\sqrt{V_0}} - \frac{\sqrt{2}\,\widetilde{S}_n(l,k)}{n\,\sqrt{V_0}} - \frac{\sqrt{2}\,\widetilde{S}_n(k+1,m)}{n\,\sqrt{V_0}} \right\}\Bigg]\,.
\end{split}    
\end{align}
Denote\, $S_n(a,b) := \widetilde{S}_n(\lfloor na \rfloor +1, \lfloor nb \rfloor)$ \,for any \,$0\leq a<b\leq 1$. Further let\, $l=\lfloor na \rfloor +1, k=\lfloor nr \rfloor$, and $m=\lfloor nb \rfloor$\, for \,$0\leq a<r<b\leq 1$. 

\begin{theorem}\label{sup_theorem1}
Under Assumption \ref{ass1_new}, as $n,p \to \infty$,
\begin{align*}
&\Big\{ \frac{\sqrt{2}}{n \sqrt{V_0}}\, S_n(a,b) \Big\}_{a,b\, \in\, [0,1]} \;\rightsquigarrow \; Q \qquad \textrm{in} \;\; L^{\infty}\left( [0,1]^2 \right)\,,
\end{align*}
where $Q$ is a centered Gaussian process with the covariance function given by 
\begin{align*}
\cov\,\big( Q(a_1,b_1)\,,\, Q(a_2,b_2) \big) \;=\; \big(b_1 \land b_2 \,-\, a_1 \lor a_2\big)^2 \,\, \mathbbm{1}\big(b_1 \land b_2 > a_1 \lor a_2 \big)\,.
\end{align*}
In particular, $var\,\big(Q(a,b)\big) = (b-a)^2\, \mathbbm{1}(b>a)$.
\end{theorem}
The proof of Theorem \ref{sup_theorem1} is given in Section \ref{tech appx}. Combining (\ref{new_eq_4}) with Theorem \ref{sup_theorem1}, it is not hard to see that as $n,p \to \infty$,
\begin{align}\label{eq_sup_7}
&\Bigg\{ \frac{\widetilde{\Delta}_n(\lfloor nr \rfloor\,;\lfloor na \rfloor +1,\lfloor nb \rfloor)}{n^2\,\sqrt{\widetilde{V}_n(\lfloor nr \rfloor\,;\lfloor na \rfloor +1,\lfloor nb \rfloor)}} \Bigg\}_{a,r,b \,\in\, [0,1]} \;\rightsquigarrow \; G' \qquad \textrm{in} \;\; L^{\infty}\left( [0,1]^3 \right)\,,
\end{align}
where 
\begin{align*}
G'(r\,;\, a,b)
:=& \; \sqrt{\frac{(r-a)^2\,(b-r)^2}{2\,(r-a)(b-r) \,+\, (r-a)^2\,+\, (b-r)^2}}  \times \, \Bigg[ -\,\frac{(b-r)}{(r-a)}\,Q(a,r)\, -\,\frac{(r-a)}{(b-r)}\,Q(r,b) \,
\\&+\,  Q(a,b) - Q(a,r) - Q(r,b) \Bigg]\\
=& \; \frac{(r-a)(b-r)}{(b-a)}\,\left[ -\,\frac{(b-r)}{(r-a)}\,Q(a,r)\, -\,\frac{(r-a)}{(b-r)}\,Q(r,b) \,+\,  Q(a,b) - Q(a,r) - Q(r,b)\right],
\end{align*}
for\, $0\leq a<r<b\leq 1$ and zero otherwise. Putting $a=0$ and $b=1$ yields 
\begin{align}\label{eq_sup_8}
&\Bigg\{ \frac{\widetilde{\Delta}_n(\lfloor nr \rfloor\,;1,n))}{n^2\,\sqrt{\widetilde{V}_n(\lfloor nr \rfloor\,;1,n)}} \Bigg\}_{r \,\in\, [0,1]} \;\rightsquigarrow \; G_0 \qquad \textrm{in} \;\; L^{\infty}\left( [0,1] \right)\,,
\end{align}
where 
\begin{align}\label{eq_sup_9}
\begin{split}
G_0(r)\; :=& \;\; r(1-r)\,\Big[ -\,\frac{(1-r)}{r}\,Q(0,r) \,-\, \frac{r}{1-r}\,Q(r,1)\,  \,+\, \big\{ Q(0,1) - Q(0,r) - Q(r,1) \big\} \Big]\\
=&\;\;  r(1-r)\, Q(0,1)\,-\, (1-r)\, Q(0,r) \,-\, r\, Q(r,1)\,
\end{split}
\end{align}
for $0<r< 1$ and zero otherwise. The second equality in (\ref{eq_sup_9}) follows from some straightforward calculations. \\

Exploring the connection between $\widetilde{\Delta}_n(\lfloor nr \rfloor\,;1,n)$ and $\widetilde{L}_n(\lfloor nr \rfloor\,;1,n)$ from (\ref{new_eq_1}), we get from (\ref{eq_sup_8}) and (\ref{eq_sup_9}) 

\begin{align}\label{new_eq_5}
&\Bigg\{ \frac{\lfloor nr \rfloor  (n-\lfloor nr \rfloor) }{n^2} \,\frac{\widetilde{L}_n(\lfloor nr \rfloor\,;1,n))}{\sqrt{\widetilde{V}_n(\lfloor nr \rfloor\,;1,n)}} \Bigg\}_{r \,\in\, [0,1]} \;\rightsquigarrow \; G_0 \qquad \textrm{in} \;\; L^{\infty}\left( [0,1] \right)
\end{align}
for $0<r< 1$ and zero otherwise.\\

Now for $1\leq l< k < m-1 \leq n-1$, define \,$\widetilde{R}_n(k,m) := \sum_{i_2=k+1}^{m}\sum_{i_1=k}^{i_2-1} \,\tau \, R(X_{i_1}, X_{i_2})$ and
\begin{align}\label{eq_sup_10}
\begin{split}
\widetilde{Q}_n(k\,;l,m)
:=&\; \frac{2\tau}{(k-l+1) (m-k)}\sum_{i_2=k+1}^{m}\sum_{i_1=l}^{k} R(X_{i_1}, X_{i_2})\,
\\&-\,\frac{\tau}{(k-l+1)(k-l)}\sum_{l \leq i_1\neq i_2 \leq k} R(X_{i_1},X_{i_2})\\
& -\,\frac{\tau}{(m-k)(m-k-1)}\sum_{k+1 \leq i_1\neq i_2 \leq m} R(X_{i_1},X_{i_1})\,.
\end{split}
\end{align}
Define $\widetilde{R}_n(k,m)$ and $\widetilde{Q}_n(k\,;l,m)$ to be zero otherwise. Comparing (\ref{eq_sup_2}) and (\ref{eq_sup_10}), it is easy to verify that $R_{n,k} = \widetilde{Q}_n(k\,;1,n)$.\\

With the definition of $\widetilde{R}_n(k,m)$ as above, we have 
\begin{align}\label{eq_sup_11}
\begin{split}
\widetilde{\Delta}^1_n(k\,;l,m) \;:=&\; (k-l+1)(m-k)\; \widetilde{Q}_n(k\,;l,m) \\
 =&\; 2 \, \left\{ \widetilde{R}_n(l,m) - \widetilde{R}_n(l,k) - \widetilde{R}_n(k+1,m) \right\}\; -\,2\frac{(m-k)}{(k-l)}\,\widetilde{R}_n(l,k) \;
 \\&-\; 2\frac{(k-l+1)}{(m-k-1)}\,\widetilde{R}_n(k+1,m)\,.
\end{split}
\end{align}
Letting\, $l=\lfloor na \rfloor +1, k=\lfloor nr \rfloor$, and $m=\lfloor nb \rfloor$\, for \,$0\leq a<r<b\leq 1$\, and proceeding along similar lines as before, we have 
\begin{align}\label{new_eq_6}
\begin{split}
\frac{\widetilde{\Delta}^1_n(k\,;l,m)}{n^2\,\sqrt{\widetilde{V}_n(k\,;l,m)}}\;=&\;  \left\{\frac{ \frac{(k-l+1)(k-l)(m-k)(m-k-1)}{n^4} }{ \frac{2(k-l)(m-k-1)}{n^2} \,+\, \frac{(m-k)(m-k-1)}{n^2} \,+\, \frac{(k-l+1)(k-l)}{n^2} }\right\}^{1/2}\\
&\times\, \Bigg[\,-\,\frac{(m-k)}{(k-l)}\,\frac{2\sqrt{2}\,\widetilde{R}_n(l,k)}{n\,\sqrt{V_0}} \;-\; \frac{(k-l+1)}{(m-k-1)}\,\frac{2\sqrt{2}\,\widetilde{R}_n(k+1,m)}{n\,\sqrt{V_0}}\\
 &\quad +\; \, \left\{ \frac{2\sqrt{2}\,\widetilde{R}_n(l,m)}{n\,\sqrt{V_0}} - \frac{2\sqrt{2}\,\widetilde{R}_n(l,k)}{n\,\sqrt{V_0}} - \frac{2\sqrt{2}\,\widetilde{R}_n(k+1,m)}{n\,\sqrt{V_0}} \right\}\Bigg]\,.
\end{split}    
\end{align}
Define\, $R_n(a,b) := \widetilde{R}_n(\lfloor na \rfloor +1, \lfloor nb \rfloor)$ \, and \, $G_n(a,b) := \frac{1}{n \sqrt{V_0}}\, R_n(a,b)$ \,for any \,$0\leq a<b\leq 1$.\\ 

\begin{theorem}\label{sup_theorem2}
Under Assumption \ref{ass2_new},\, as\, $n,p \to \infty$, $\dis\sup_{a,b\, \in\, [0,1]} \, \vert G_n(a,b) \vert \;=\; o_p(1)$.
\end{theorem}
The proof of Theorem \ref{sup_theorem2} is given in Section \ref{tech appx}. As a consequence of Theorem \ref{sup_theorem2}, we have as\, $n,p \to \infty$,
\begin{align}\label{eq_sup_12}
\dis\sup_{a,r,b\, \in\, [0,1]} \, \left\vert \frac{\widetilde{\Delta}^1_n(\lfloor nr \rfloor\,;\lfloor na \rfloor + 1,\lfloor nb \rfloor)}{n^2\,\sqrt{\widetilde{V}_n(\lfloor nr \rfloor\,;\lfloor na \rfloor + 1,\lfloor nb \rfloor)}} \right\vert \;=\;  o_p(1) \qquad  \textrm{as} \;\;\;n,p \to \infty\,.
\end{align}
As a special case, putting $a=0$ and $b=1$, we get from (\ref{eq_sup_12})
\begin{align}\label{eq_sup_13}
\dis\sup_{r\, \in\, [0,1]} \, \left\vert \frac{\widetilde{\Delta}^1_n(\lfloor nr \rfloor\,; 1,n)}{n^2\,\sqrt{\widetilde{V}_n(\lfloor nr \rfloor\,; 1,n)}}\right\vert \;=\;  o_p(1) \qquad  \textrm{as} \;\;\;n,p \to \infty\,.
\end{align}
Exploring the connection between $\widetilde{\Delta}^1_n(\lfloor nr \rfloor\,; 1,n)$ and $\widetilde{Q}^1_n(\lfloor nr \rfloor\,; 1,n)$ from (\ref{eq_sup_11}), and following similar arguments as before, we can write
\begin{align}\label{new_eq_7}
& \dis\sup_{r\, \in\, [0,1]} \, \left\vert \frac{\lfloor nr \rfloor  (n-\lfloor nr \rfloor) }{n^2} \,\frac{\widetilde{Q}_n(\lfloor nr \rfloor\,;1,n))}{\sqrt{\widetilde{V}_n(\lfloor nr \rfloor\,;1,n)}} \right\vert \;=\;  o_p(1) \qquad  \textrm{as} \;\;\;n,p \to \infty\,.
\end{align}
Finally, define
\begin{align*}
\begin{split}
\widehat{V}_{n}(k\,;l,m)\;&:=\;\frac{4\widehat{\cal{C}}_{l,k,m}}{(k-l+1)(m-k)}\, +\, \frac{2\widehat{\cal{D}}^2_{l:k}}{(k-l+1)(k-l)} \,+\, \frac{2\widehat{\cal{D}}^2_{k+1:m}}{(m-k)(m-k-1)},
\end{split}
\end{align*}
and
\begin{align}\label{eq_sup_6.1}
\begin{split}
\widehat{V}^{\Delta}_n(k\,;l,m) \;:=&\; (k-l+1)(k-l)(m-k)(m-k-1)\,\, \widehat{V}_{n}(k\,;l,m) \\
 =&\; 2(m-k-1)(k-l) \,\, 4\widehat{\cal{C}}_{l,k,m} \; 
 -\;2(m-k)(m-k-1)\,\,2\widehat{\cal{D}}^2_{l:k}  
 \\&-\; 2(k-l+1)(k-l)\,\,2\widehat{\cal{D}}^2_{k+1:m}\\
 =&\; \widehat{V}^{\Delta 1}_n(k\,;l,m) \,+\, \widehat{V}^{\Delta 2}_n(k\,;l,m)  \,+\, \widehat{V}^{\Delta 3}_n(k\,;l,m),
\end{split}
\end{align}
for $1\leq l< k < m-1 \leq n-1$  and zero otherwise. From (\ref{eq_sup_0.5}) and (\ref{eq_sup_6.1}), it is not hard to see that $\widehat{V}_{n,k} = \widehat{V}_{n}(k\,;1,n)$. And from (\ref{eq_sup_6}) and (\ref{new_eq_2}), we can write
\begin{align*}
\begin{split}
\widetilde{V}_n(k\,;l,m) \;&=\; \frac{1}{(k-l+1)(m-k)}\,V_0 \,+\, \frac{1}{2(k-l+1)(k-l)}\,V_0  \,+\, \frac{1}{2(m-k)(m-k-1)}\,V_0,
\end{split}
\end{align*}
and
\begin{align}\label{new_eq_8}
\begin{split}
\widetilde{V}^{\Delta}_n(k\,;l,m) \;&:=\; (k-l+1)(k-l)(m-k)(m-k-1)\,\, \widetilde{V}_{n}(k\,;l,m)\\ &=\; (m-k-1)(k-l)\,V_0 \,+\, \frac{1}{2}\,(m-k)(m-k-1)\,V_0  \,+\, \frac{1}{2}\,(k-l)(k-l+1)\,V_0 \\
&=\; \widetilde{V}^{\Delta 1}_n(k\,;l,m) \,+\, \widetilde{V}^{\Delta 2}_n(k\,;l,m) \,+\, \widetilde{V}^{\Delta 3}_n(k\,;l,m)\,.
\end{split}
\end{align}

\begin{theorem}\label{sup_theorem3}
Under Assumptions \ref{ass1_new} and \ref{ass2_new},\, as $n,p \to \infty$, $$ \dis\sup_{a,r,b\, \in\, [0,1]} \, \left\vert \frac{\widehat{V}^{\Delta}_n(\lfloor nr \rfloor\,;\lfloor na \rfloor + 1,\lfloor nb \rfloor)}{\widetilde{V}^{\Delta}_n(\lfloor nr \rfloor\,;\lfloor na \rfloor + 1,\lfloor nb \rfloor)} \,-\, 1  \right\vert \;=\;  o_p(1) \,.$$
\end{theorem}
The proof of Theorem \ref{sup_theorem3} is given in Section \ref{tech appx}. As a special case, putting $a=0$ and $b=1$, we get from Theorem  \ref{sup_theorem3} that as $n,p \to \infty$,
\begin{align}\label{eq_sup_13.1}
\dis\sup_{r\, \in\, [0,1]} \, \left\vert \frac{\widehat{V}^{\Delta}_n(\lfloor nr \rfloor\,; 1,n)}{\widetilde{V}^{\Delta}_n(\lfloor nr \rfloor\,; 1,n)} \,-\, 1 \right\vert \;=\;  o_p(1) \,.
\end{align}
With all the above, the proof of Theorem \ref{theorem1} can be completed as below.

\begin{proof}[Proof of Theorem \ref{theorem1}]
Combining (\ref{eq_sup_1}) and (\ref{eq_sup_0.6}) with (\ref{new_eq_5}) and (\ref{new_eq_7}) yields that
\begin{align}\label{new_eq_9}
&\Big\{ \breve{T}_{n}(\lfloor nr \rfloor)\Big\}_{r \in [0,1]} \;\rightsquigarrow \; G_0 \qquad \textrm{in} \;\; L^{\infty}\left( [0,1] \right)\,,
\end{align}
as $n,p \to \infty$. And (\ref{eq_sup_13.1}) implies
\begin{align}\label{new_eq_10}
\dis\sup_{r\, \in\, [0,1]} \, \left\vert \frac{\widehat{V}_n(\lfloor nr \rfloor\,; 1,n)}{\widetilde{V}_n(\lfloor nr \rfloor\,; 1,n)} \,-\, 1 \right\vert \;=\;  o_p(1) \,,
\end{align}
as $n,p \to \infty$. This equipped with (\ref{new_eq_9}) completes the proof of Theorem \ref{theorem1}.
\end{proof}

All along our derivations, we use the simple facts that for $0<a \leq 1$, $\lfloor na \rfloor \asymp n$ and 
\begin{align}\label{eq_sup_13.2}
\dis\lim_{n \to \infty} \frac{\lfloor na \rfloor}{n} = \dis\lim_{n \to \infty} \frac{na - \{na\}}{n} = a\,
\end{align}
as $n \to \infty$, since $0 \leq \{na\} <1$. 

\section{Proof of Theorem \ref{alt:fix_alternative}} \label{proof_alternative}
The first result can be proved straightforwardly. Note that $\E[L(X_1,X_1')] = \E[L(X_n,X_n')] = \E[L(X_1,X_n)] = 0$. Therefore, \[E_{\gamma}(X_1,X_n) = 2\E[\tau_3 + \tau_3R(X_1,X_n)] - \E[\tau_1 + \tau_1R(X_1,X_1')] - \E[\tau_2 + \tau_2R(X_n,X_n')].\] By Assumption \ref{ass2_new}, we have $|\E[\tau_1R(X,X_1')]| = o(\sqrt{V_1}/n)$, $|\E[\tau_2R(X_n,X_n')]| = o(\sqrt{V_2}/n)$ and $|\E[\tau_3R(X_1,X_n)]| = o(\sqrt{V_3}/n)$. Hence, $E_{\gamma}(X_1,X_n) = 2\tau_3 - \tau_1 - \tau_2 + o(\sqrt{V}/n)$.

For the second result, by the definition of $M_n$,
    \[M_n = \max_{4 \leq k \leq n-4}\frac{k(n-k)}{n^2}T_n(k)  \geq \frac{\nu(n - \nu)}{n^2}T_n(\nu).\]
    Therefore, it suffices to show \[\frac{\nu(n - \nu)}{n^2}T_n(\nu) \overset{P}{\rightarrow} \infty,\] as $n \wedge p \rightarrow \infty$. Based on Proposition \ref{Prop 4.1 in CZ}, we obtain a decomposition of $\widehat E_{n,k} := \widehat E_\gamma(\mathbf{X}_{1:k},\mathbf{X}_{k+1:n})$. To state the result, we define the following terms: 
		\begin{align*}
				\widetilde E_{n,k} \;=&\; \frac{2}{k (n-k)}\sum_{i_1=1}^{k}\sum_{i_2=k+1}^{n} \tau_{i_1,i_2}\,-\,\frac{2}{k(k-1)}\sum_{1 \leq i_1< i_2 \leq k} \tau_{i_1,i_2} 
		  \\&-\,\frac{2}{(n-k)(n-k-1)}\sum_{k+1 \leq i_1< i_2 \leq n} \tau_{i_1,i_2},\\
				L_{n,k} \;=&\; \frac{1}{k (n-k)}\sum_{i_1=1}^{k}\sum_{i_2=k+1}^{n} H(X_{i_1}, X_{i_2})\,
				-\,\frac{1}{k(k-1)}\sum_{1 \leq i_1< i_2 \leq k} H(X_{i_1},X_{i_2})
				\\& -\,\frac{1}{(n-k)(n-k-1)}\sum_{k+1 \leq i_1< i_2 \leq n} H(X_{i_1},X_{i_2})\,,\\
				R_{n,k} \;=&\; \frac{2}{k (n-k)}\sum_{i_1=1}^{k}\sum_{i_2=k+1}^{n}\tau_{i_1,i_2} R(X_{i_1}, X_{i_2})\,-\,\frac{1}{k(k-1)}\sum_{1 \leq i_1\neq i_2 \leq k} \tau_{i_1,i_2}R(X_{i_1},X_{i_2})\\
				&  -\,\frac{1}{(n-k)(n-k-1)}\sum_{k+1 \leq i_1\neq i_2 \leq n}\tau_{i_1,i_2} R(X_{i_1},X_{i_2})\,,
				\end{align*}
		where $\tau_{i_1,i_2}^2 = \E[\gamma^2(X_{i_1},X_{i_2})]$. Further, for the ease of notations, we let $X$ and $Y$ be two independent random vectors that are also independent of any $X_i$ for all $i$ and $X \sim  F_1$ and $Y \sim F_2$.  Let
		\begin{align*}
			  U_{n,k} =& \frac{1}{k}\sum_{i = 1}^k\frac{n-\nu}{n-k}(\tau_3\E[L(X_i,Y)|X_i] - \tau_1\E[L(X_i,X)|X_i])
			  \\&-\frac{1}{n-k}\sum_{i = k+1}^{\nu}\frac{n-\nu}{n-k-1}(\tau_3\E[L(X_i,Y)|X_i] - \tau_1\E[L(X_i,X)|X_i])\\
				&+\frac{1}{n-k}\sum_{i = \nu+1}^{n}\frac{n-\nu-1}{n-k-1}(\tau_3\E[L(X_i,X)|X_i] - \tau_2\E[L(X_i,Y)|X_i]),
		  \end{align*}
		  for $k \leq \nu$ and
		  \begin{align*}
			  U_{n,k} =& \frac{1}{k}\sum_{i = 1}^{\nu}\frac{\nu-1}{k-1}(\tau_3\E[L(X_i,Y)|X_i] - \tau_1\E[L(X_i,X)|X_i])
			  \\&-\frac{1}{k}\sum_{i = \nu+1}^k\frac{\nu}{k-1}(\tau_3\E[L(X_i,X)|X_i] - \tau_2\E[L(X_i,Y)|X_i])\\
				&+\frac{1}{n-k}\sum_{i = k+1}^n\frac{\nu}{k}(\tau_3\E[L(X_i,X)|X_i] - \tau_2\E[L(X_i,Y)|X_i]),
		  \end{align*}
		for $k >\nu$.
    
    \begin{lemma}\label{alt:decomposition}
		We have a decomposition for $\widehat E_{n,k}$ given by $$\widehat E_{n,k} = \widetilde E_{n,k} + L_{n,k} + U_{n,k} + R_{n,k},$$ where 
		$$\widetilde E_{n,k} = 
		\begin{cases}
		(2\tau_3 - \tau_1 - \tau_2)\frac{(n - \nu)(n - \nu - 1)}{(n - k)(n-k-1)},\quad &\text{ for $k \leq \nu$},\\
		(2\tau_3 - \tau_1 - \tau_2)\frac{\nu(\nu - 1)}{k(k-1)},\quad &\text{ for } k>\nu.
		\end{cases}$$
		\end{lemma}

    Therefore 
    \begin{align*}
        T_{n}(\nu) := \frac{\widetilde{E}_{n,\nu}+L_{n,\nu}+U_{n,\nu} + R_{n,\nu}}{a_{\nu,(n-\nu)}\widehat{S}(\mathbf{X}_{1:\nu},\mathbf{X}_{(\nu+1):n})}
    \end{align*}
 by Lemma \ref{alt:decomposition}, which is a direct consequence of Proposition \ref{Prop 4.1 in CZ}. Next, we analyze the orders of $L_{n,\nu}$, $U_{n,\nu}$, $R_{n,\nu}$ and $\widehat{S}(\mathbf{X}_{1:\nu}, \mathbf{X}_{(\nu+1):n})$. The results are presented in the following lemmas.

		

	\begin{lemma} \label{alt:jointnormal}
	Under Assumptions \ref{ass0_new} and \ref{ass1_new}, 
		$\left( L_{n,\nu}^{(1)},
			L_{n,\nu}^{(2)},
			L_{n,\nu}^{(3)}
		\right)^\top \overset{d}{\rightarrow} N\left(\boldsymbol{0}, \mathbf{I}_p\right)$, where
\begin{align*}
& L_{n,\nu}^{(1)} := \frac{\sqrt{2}\zeta n}{\sqrt{V_1}}\frac{1}{\nu(\nu-1)}\sum_{1 \leq i < j \leq \nu}H(X_i,X_j), \\   
&L_{n,\nu}^{(2)} := \frac{\sqrt{2}(1-\zeta)n}{\sqrt{V_2}}\frac{1}{(n - \nu)(n-\nu-1)}\sum_{\nu+1 \leq i < j \leq n}H(X_i,X_j),
\end{align*}
		and
		\[L_{n,\nu}^{(3)} := \frac{n\sqrt{\zeta(1-\zeta)}}{\sqrt{V_3}}\frac{1}{\nu(n - \nu)}\sum_{i = 1}^{\nu}\sum_{j = \nu+1}^n H(X_i,X_j).\]
			
		\end{lemma}

\begin{lemma}\label{alt:jointnormal2}
    Under Assumption \ref{ass0_new}, we have $U_{n,\nu} = O_p(\sqrt{\max(\Gamma_1,\Gamma_2)/n})$.
\end{lemma}
		
		\begin{lemma}\label{alt:denominator}
			Under Assumptions \ref{ass0_new}, \ref{ass1_new} and \ref{ass2_new}, 
   \begin{align*}
   & 4\widehat{\mathcal{D}}^2(\mathbf{X}_{1:\nu})/V_1\overset{P}{\rightarrow}1,\\
   & 4\widehat{\mathcal{D}}^2(\mathbf{X}_{\nu+1:n})/V_2\overset{P}{\rightarrow}1,
   \end{align*}
			and
			\[4\widehat{\mathcal{C}}(\mathbf{X}_{1:\nu},\mathbf{X}_{\nu+1:n})/V_3\overset{P}{\rightarrow}1.\]
			As a result, $a_{\nu,n-\nu}^2\widehat{S}^2(\mathbf{X}_{1:\nu},\mathbf{X}_{(\nu+1),n}) = O_p(\max\{V_1,V_2,V_3\}/n^2)$.
		\end{lemma}

		\begin{lemma}\label{alt:negligible}
			Under Assumptions \ref{ass0_new}, \ref{ass1_new} and \ref{ass2_new}, $R_{n,\nu} = o_p(n^{-1}\max\{\sqrt{V_1},\sqrt{V_2},\sqrt{V_3}\}).$
		\end{lemma}
	
	According to Lemmas \ref{alt:jointnormal}, \ref{alt:denominator} and \ref{alt:negligible}, 
     \begin{align*}
     & L_{n,\nu} = O_p(n^{-1}\max\{\sqrt{V_1},\sqrt{V_2},\sqrt{V_3}\}),\\& R_{n,\nu} = o_p(n^{-1}\max\{\sqrt{V_1},\sqrt{V_2},\sqrt{V_3}\}),\\& U_{n,\nu} = O_p\left(\sqrt{\frac{\max\{\Gamma_1,\Gamma_2\}}{n})}\right),   
     \end{align*}
 and 
    \[a_{\nu,n-\nu}\widehat{S}(\mathbf{X}_{1:\nu},\mathbf{X}_{(\nu+1),n}) = O_p(n^{-1}\max\{\sqrt{V_1},\sqrt{V_2},\sqrt{V_3}\}).\]
    Hence,
    \begin{align*}
        T_{n}(\nu) &:= \frac{\widetilde{E}_{n,\nu}+L_{n,\nu} + U_{n,\nu} + R_{n,\nu}}{a_{\nu,(n-\nu)}\widehat{S}(\mathbf{X}_{1:\nu},\mathbf{X}_{(\nu+1):n})}
        \\&= \left[\frac{2\tau_3-\tau_1-\tau_2}{n^{-1}\sqrt{V}} + O_p(1)\right]\left[\frac{n^{-1}\sqrt{V}}{a_{\nu,(n-\nu)}\widehat{S}(\mathbf{X}_{1:\nu},\mathbf{X}_{(\nu+1):n})}\right] \rightarrow \infty,
    \end{align*}
    if \[\frac{n(2\tau_3 - \tau_1 - \tau_2)}{\sqrt{V}}\rightarrow \infty,\]
    which completes the proof. 

\section{Proof of Theorem \ref{alt:consistency}}\label{proof_consistency}
We first state the following technical lemmas, which are essential for the proof.

\begin{lemma}\label{lem:alt_process}
    Under Assumptions \ref{ass0_new} and \ref{ass1_new}, define $S_n'(a,b)$ as \[S_n'(a,b):= \sum_{i = \floor{na}+2}^{\floor{nb}}\sum_{j = \floor{na}+1}^{i-1}\frac{H(X_{i},X_{j})}{\sqrt{\E[H(X_{i},X_{j})^2]}},\]
    for any $0\leq a < b \leq 1$. Then we have
    \[\frac{\sqrt{2}}{n}S_{n}'(a,b) \rightsquigarrow Q(a,b),\]
    in $\mathcal{L}^{\infty}([0,1]^2)$, where $Q$ was defined in Theorem \ref{theorem1}. 
\end{lemma}

\begin{lemma} \label{lem:alt_U}
    Under Assumption \ref{ass0_new}, \[\sup_{k}\frac{k(n-k)}{n^2}|U_{n,k}| = O_p\left(\sqrt{\frac{\max\{\Gamma_1, \Gamma_2\}}{n}}\right).\]
\end{lemma}

\begin{lemma}\label{lem:alt_rem}
    Under Assumptions \ref{ass0_new}, \ref{ass1_new} and \ref{ass2_new}, $\sup_{k}|R_{n,k}| = o_p(n^{-1}\max\{\sqrt{V_1},\sqrt{V_2},\sqrt{V_3}\})$.
\end{lemma}

\begin{proof}[Proof of Theorem \ref{alt:consistency}]
    Consider the case when $\widehat{\nu}^* <\nu$. The other case can be handled in a similar fashion.  
    By the definition of $\widehat{\nu}^*$ and Theorem \ref{alt:decomposition},
    \begin{align*}
     0\leq & \frac{\widehat{\nu}^*(n-\widehat{\nu}^*)}{n^2}\widehat{E}_{\gamma}(\mathbf{X}_{1:\widehat{\nu}^*},\mathbf{X}_{(\widehat{\nu}^*+1):n}) - \frac{\nu(n-\nu)}{n^2}\widehat{E}_{\gamma}(\mathbf{X}_{1:\nu},\mathbf{X}_{(\nu+1):n})\\
        =&\frac{\widehat{\nu}^*(n-\widehat{\nu}^*)}{n^2}\left\{(2\tau_3-\tau_1-\tau_2)\frac{(n -\nu)(n-\nu-1)}{(n - \widehat{\nu}^*)(n-\widehat{\nu}^*-1)} + L_{n,\widehat{\nu}^*} + U_{n,\widehat{\nu}^*} + R_{n,\widehat{\nu}^*}\right\}\\
        &-\frac{\nu(n-\nu)}{n^2}\left\{2\tau_3-\tau_1-\tau_2 + L_{n,\nu} + U_{n,\nu} + R_{n,\nu}\right\}\\
        \leq &\frac{2\tau_3 - \tau_1 - \tau_2}{n^2}\left\{\frac{\widehat \nu^*}{n - \widehat \nu^* -1 }(n - \nu)(n-\nu-1) - \nu(n - \nu)\right\} 
        \\&+ 2\sup_{k}\frac{k(n-k)}{n^2}\left|L_{n,k}\right| + 2\sup_k\frac{k(n-k)}{n^2}\left|U_{n,k}\right|+2\sup_{k}\frac{k(n-k)}{n^2}\left|R_{n,k}\right|\\
        \leq &(2\tau_3 - \tau_1 - \tau_2)(\widehat \nu^* - \nu)\frac{(n-1)(n-\nu)}{2n^2(n-\widehat \nu^*-1)}+ 2\sup_{k}\frac{k(n-k)}{n^2}\left|L_{n,k}\right|\\
        &+ 2\sup_k\frac{k(n-k)}{n^2}\left|U_{n,k}\right| +2\sup_{k}\frac{k(n-k)}{n^2}\left|R_{n,k}\right|.
    \end{align*} 
Therefore,
    \begin{align*}
        &(2\tau_3 - \tau_1 - \tau_2)(\nu - \widehat \nu^*)\frac{(n-1)(n-\nu)}{2n^2(n - \widehat \nu^* -1 )}
        \\ \leq& 2\sup_{k}\frac{k(n-k)}{n^2}\left|L_{n,k}\right|+ 2\sup_k\frac{k(n-k)}{n^2}\left|U_{n,k}\right|+2\sup_{k}\frac{k(n-k)}{n^2}\left|R_{n,k}\right|.
        \end{align*}
    By rearranging the terms, we have
        \begin{align*}
        \nu-\widehat \nu^*&\leq \frac{4n^2(n-\widehat \nu^*-1)}{(2\tau_3-\tau_1-\tau_2)(n-1)(n-\nu)}\left\{\sup_{k}\frac{k(n-k)}{n^2}|L_{n,k}|+ 2\sup_k\frac{k(n-k)}{n^2}\left|U_{n,k}\right| + \sup_{k}|R_{n,k}|\right\},
        \end{align*}
    which implies that 
        \begin{align*}
        \zeta - \widehat\zeta^* &\leq \frac{4}{(2\tau_3-\tau_1-\tau_2)(1-n^{-1})(1-\zeta)}\left\{\sup_{k}\frac{k(n-k)}{n^2}|L_{n,k}|+ 2\sup_k\frac{k(n-k)}{n^2}\left|U_{n,k}\right| + \sup_{k}|R_{n,k}|\right\}.
        \end{align*}

    It remains to show that $\sup_{k}k(n-k)n^{-2}|L_{n,k}|$, $\sup_k\frac{k(n-k)}{n^2}\left|U_{n,k}\right|$ and $\sup_{k}|R_{n,k}|$ are of the order $O_p((2\tau_3 - \tau_1 - \tau_2)a_n^{-1})$. To this end,
    assume $2 \leq k \leq \nu$ first. 
    Then $L_{n,k}$ can be expressed as 
        \begin{align*}
            L_{n,k} =& \frac{1}{k(n-k)}\sum_{i_1 = 1}^k\sum_{i_2 = k+1}^{\nu}H(X_{i_1},X_{i_2}) + \frac{1}{k(n-k)}\sum_{i_1 = 1}^k\sum_{i_2 = \nu+1}^{n}H(X_{i_1},X_{i_2}) \\
            &-\frac{1}{k(k-1)}\sum_{i_2 = 2}^k\sum_{i_1 = 1}^{i_2-1}H(X_{i_1},X_{i_2})-\frac{1}{(n-k)(n-k-1)}\sum_{\nu+1 \leq i_1 < i_2 \leq n}H(X_{i_1},X_{i_2})\\
            &-\frac{1}{(n-k)(n-k-1)}\sum_{k+1 \leq i_1 < i_2 \leq n}H(X_{i_1},X_{i_2})-\frac{1}{(n-k)(n-k-1)}\sum_{i_1 = k+1}^{\nu}\sum_{i_2 = \nu+1}^n H(X_{i_1},X_{i_2})\\
            =& \frac{\sqrt{V_1}}{k(n-k)}\left(S_n'(0,\zeta) - S_n'(0,r) - S_n'(r,\zeta)\right) + \frac{\sqrt{V_3}}{k(n-k)}\left(S_n'(0,1) + S_n'(r,\zeta) - S_n'(0,\zeta) - S_n'(\zeta,1)\right) \\
            &-\frac{\sqrt{V_1}}{k(k-1)}S_n'(0,r)-\frac{\sqrt{V_2}}{(n-k)(n-k-1)}S_n'(\zeta,1)\\
            &-\frac{\sqrt{V_1}}{(n-k)(n-k-1)}S_n'(r,\zeta)-\frac{\sqrt{V_3}}{(n-k)(n-k-1)}\left(S_n'(r,1) - S_n'(r,\zeta) - S_n'(\zeta,1)\right),
        \end{align*}
        where $r = k/n$. For the first term,
        \begin{align*}
            &\sup_{2\leq k \leq \nu}\frac{k(n-k)}{n^2}\frac{\sqrt{V_1}}{k(n-k)}|S_n'(0,\zeta) - S_n'(0,r) - S_n'(r,\zeta)| 
            \\ \lesssim&\frac{\sqrt{V_1}}{n}\sup_{2 \leq k \leq \nu}\sup_{0<a<b<1}\frac{1}{n}|S_n'(a,b)| = O_p(\sqrt{V_1}/n),
        \end{align*}
where we have used the fact that $\sup_{0<a<b<1}n^{-1}|S_n'(a,b)| = O_p(1)$ as implied by Lemma \ref{lem:alt_process} and the functional continuous mapping theorem. By similar arguments, we can obtain the order of the second term, which is given by 
\[\sup_{2\leq k \leq \nu}\frac{k(n-k)}{n^2}\frac{\sqrt{V_3}}{k(n-k)}|S_n'(0,1) + S_n'(r,\zeta) - S_n'(0,\zeta) - S_n'(\zeta,1)| \leq O_p(\sqrt{V_3}/n).\]
Similarly because $\sup_{k}k(n-k) \leq n^2$ and $\sup_{2 \leq k \leq \nu}\{(n-k)(n-k-1)\}^{-1} = \{(n - \nu)(n-\nu-1)\}^{-1} = O(n^{-2})$, the orders of the fourth, fifth and sixth terms above are all equal to $O_p((2\tau_3 - \tau_1-\tau_2)a_n^{-1})$. 
Hence, it remains to obtain the order of the third term: 
\begin{align*}
    \frac{k(n-k)}{n^2}\frac{\sqrt{V_1}}{k(k-1)}S_{n}'(r) \leq \frac{\sqrt{V_1}}{n}\frac{1}{k-1}\sum_{i = 2}^{k}\sum_{j = 1}^{i-1}\frac{H(X_{i},X_{j})}{\sqrt{V_1}} = \frac{\sqrt{V_1}}{n}\frac{1}{k-1}\sum_{i = 2}^k\xi_i,
\end{align*}
where $\xi_i = \sum_{j = 1}^{i-1}H(X_{i},X_{j})/\sqrt{V_1}$. To this end, we note that $\sum_{i = 2}^k\xi_{i}$ for $k = 2,3,\dots$ is a martingale sequence. By the H\'ajek-R\'enyi's inequality for martingale difference sequences, we have \[P\left(\sup_{2 \leq k \leq \nu}\left|\frac{1}{k-1}\sum_{i = 2}^k\xi_i\right| \geq \epsilon\right)\leq \epsilon^{-2}\sum_{j = 1}^{\nu-1}\frac{\E(\xi_j^2)}{j^2} = \epsilon^{-2}\sum_{j = 1}^{\nu - 1}j^{-1} \lesssim \log(\nu),\]  
which indicates that $$\sup_{2 \leq k \leq \nu}\left|\frac{k(n-k)}{n^2}\frac{\sqrt{V_1}}{k(k-1)}S_n'(k/n)\right| = O_p(\frac{\sqrt{V_1}\sqrt{\log(n)}}{n}).$$ Combining the above results, we prove that $\sup_{2 \leq k \leq \nu}k(n-k)n^{-2}|L_{n,k}| = O_p((2\tau_3-\tau_1-\tau_2)a_n^{-1})$.
For $U_{n,k}$, Lemma \ref{lem:alt_U} indicates that $$\sup_{k}\frac{k(n-k)}{n^2}|U_{n,k}| = O_p\left(\sqrt{\frac{\max\{\Gamma_1, \Gamma_2\}}{n}}\right).$$ For $R_{n,k}$, Lemma \ref{lem:alt_rem} implies that $\sup_{k \leq \nu}\left|R_{n,k}\right| = o_p((2\tau_3 - \tau_1 - \tau_2)a_n^{-1})$, which proves the desired results. For the case of $k > \nu$, the results can be proved in a similar fashion, which completes the proof.
\end{proof}

\section{Proof of Theorem \ref{mult_consistency_SBS}}\label{sec:theory_not}
We first state two lemmas that are essential for the proof.

\begin{lemma}\label{lem:isolating_seeded}
Let $1 \le \nu_1 < \dots < \nu_N < n$ be the true change-point locations, and define $\nu_0 := 0$ as well as $\nu_{N+1} := n$. Suppose the minimum spacing condition $\min_{0\le \ell\le N}(\nu_{\ell+1}-\nu_\ell) \ge \Delta_n$ holds. Then, for each $\ell \in \{1, \dots, N\}$, there exists a seeded interval $I_\ell \in \mathcal{I}_n$ that isolates $\nu_\ell$---meaning $I_\ell$ contains $\nu_\ell$ and no other change-points---with its length satisfying $|I_\ell| \in [\Delta_n/2, \Delta_n]$.
\end{lemma}

\begin{lemma}\label{lem:SBS}
For each $\ell \in \{1, \dots, N\}$, let $I_\ell = [s_\ell, e_\ell] \in \mathcal{I}_n$ be an isolating interval that contains $\nu_\ell$ and no other change-points, satisfying $(\nu_{\ell} - s_{\ell} + 1)(e_{\ell} - \nu_{\ell}) \asymp |I_{\ell}|^2$. Define the signal strength
\[
b_{n,\ell} := \frac{n\delta_{\ell,\ell+1}}{\sqrt{V_{n,\ell}}},
\]
and let $\lambda_{|I_\ell|}$ be a threshold depending on the interval length $|I_\ell|$ such that $\lambda_{|I_\ell|} \to \infty$ and $\lambda_{|I_\ell|} = o(b_{|I_\ell|,\ell})$. Under Assumptions \ref{ass0_new}--\ref{ass2_new}, if $b_{n,\ell} \to \infty$, then there exists a constant $C>0$ such that as $n, p \to \infty$,
\[
P\!\left(b_{n,\ell}\big|\widehat\zeta(I_\ell)-\zeta_\ell\big|\le C\right)\to 1
\quad\text{and}\quad
P\!\left(M(I_\ell)>\lambda_{|I_\ell|}\right)\to 1,
\]
where $\widehat\zeta(I_\ell) := \widehat\nu(I_\ell)/n$.
\end{lemma}

{\color{red} Suppose there are no change-points in the sequence. By equation (25) in the proof of Theorem \ref{theorem1}, $\{M([\floor{na},\floor{nb}])\}_{0 \leq a < b \leq 1}$ converges to a two-parameter stochastic process in $D([0,1]^2)$. Therefore by the continuous mapping theorem, $$\sup_{I \in \mathcal{I}_n} M(I) \leq \sup_{0 \leq a < b \leq 1}M([\floor{na},\floor{nb}]) = O_p(1).$$ Because the threshold $\lambda_I \to \infty$, $P(\bigcup_{I \in \mathcal{I}_n}\{M(I) > \lambda_I\}) \leq P(\sup_{I \in \mathcal{I}_n}M(I) > \inf_{I \in \mathcal{I}_n} \lambda_I) = o(1)$. Thus, with probability tending to 1, the algorithm terminates immediately and returns an empty set.}

Conversely, suppose there is at least one change-point. By Lemma \ref{lem:isolating_seeded}, for each $\ell \in \{1, \dots, N\}$, there exists an interval $I_\ell = [s_\ell, e_\ell] \in \mathcal{I}_n$ that isolates the change-point $\nu_\ell$ (i.e., contains no other change-points) and satisfies $|I_\ell| \asymp \Delta_n$ and $(\nu_{\ell} - s_{\ell}+1)(e_{\ell} - \nu_{\ell}) \asymp |I_{\ell}|^2$. Therefore, if the threshold satisfies $\lambda_{I_\ell} = o(b_{|I_\ell|,\ell})$ and the minimum signal strength diverges, i.e.,
\[
n \min_{1 \le \ell \le N} \frac{\delta_{\ell,\ell+1}}{\sqrt{V_{n,\ell}}} \to \infty,
\]
then Lemma \ref{lem:SBS} ensures that $P(M(I_\ell) > \lambda_{I_\ell}) \to 1$ for each $\ell = 1, \dots, N$. Furthermore, there exists a constant $C > 0$ such that $P(b_{n,\ell}|\widehat{\zeta}(I_\ell) - \zeta_{\ell}| \le C) \to 1$. \rev{Because $N$ is fixed, a union bound shows that the intersection of these $N$ detection and localization events also has probability tending to one.}

During the procedure, the algorithm selects the shortest significant interval $I^* = [s^*, e^*]$ among all seeded intervals satisfying $M(I) > \lambda_{I}$. With probability tending to 1, this minimal interval $I^*$ isolates exactly one true change-point, say $\nu_{\ell}$, and maintains the proportionality $(e^* - \nu_{\ell})(\nu_{\ell} - s^* + 1) \asymp (e^* - s^* + 1)^2$. Applying Lemma \ref{lem:SBS} to $I^*$, the estimated location $\widehat{\nu}_{\ell} := \widehat{\nu}(I^*)$ satisfies $b_{n,\ell}|\widehat{\zeta}(I^*) - \zeta_{\ell}| \le C$ with probability tending to 1, establishing that $\zeta_{\ell}$ is consistently estimated.

After detecting the first change point, the whole sample will be split into two subsamples, $\mathcal{X}_1 = \{X_1,...,X_{\widehat{\nu}_{\ell}}\}$ and $\mathcal{X}_2 = \{X_{\widehat{\nu}_{\ell+1}},...,X_n\}$. It is easy to see that only one of the two subsamples may contain $\nu_{\ell}$ as a change point (if $\nu_{\ell} \neq \widehat \nu_{\ell}$). Without loss of generality, we assume that $\nu_{\ell}$ is in $\mathcal{X}_1$. The proof of Lemma \ref{lem:SBS} indicates that with probability going to 1, $|\widehat{\nu}_{\ell} - \nu_{\ell}| \lesssim n/b_{n,\ell} \leq \sqrt{V_{|I^*|,\ell}}/\delta_{\ell,\ell+1}$. In the new segment $\mathcal{X}_1$ for an interval $I = [s,e]$ that only contains the change point $\nu_{\ell}$, if $I$ is shorter than $I^*$ from the previous step, then $M(I) < \lambda_I$ otherwise $I^*$ is not the shortest interval containing only one change point with $M(I^*) > \lambda_{I^*}$. If $I \in \mathcal{S}$ where $\mathcal{S}$ is the set of such intervals that is longer than $I^*$, by the results in the proof of Theorem \ref{alt:fix_alternative} (Lemma \ref{alt:decomposition}-\ref{alt:negligible}), 
\[M(I) \leq  {\color{red} \delta_{\ell,\ell+1}\frac{(e-\nu_{\ell})(\nu_{\ell}-s+1)}{n\sqrt{V_{|I_\ell|,\ell}}}} + M_0(I),\]
where $\sup_{I \in \mathcal{S}} M_0(I) = O(\sqrt{V_{n,l}/V_{n,l}^{(1)}})$. We also have \begin{align*}
    {\color{red}\delta_{\ell,\ell+1}\frac{(e-\nu_{\ell})(\nu_{\ell}-s+1)}{n\sqrt{V_{|I_\ell|,\ell}}}}
    \leq \delta_{\ell,\ell+1}\frac{|\widehat \nu_{\ell} - \nu_{\ell}|}{\sqrt{V_{|I_\ell|,\ell}}} 
    \leq C\frac{\sqrt{V_{|I^*|,\ell}}}{\sqrt{V_{|I|,\ell}}}\leq C\\
\end{align*}
for some positive constant $C$, as $V_{|I^*|,\ell} \leq V_{|I|,\ell}$ if $|I^*| \leq |I|$. Therefore $P(\bigcup_{I \in \mathcal{S}}\{M(I) \geq \lambda_{I}\}) = o(1)$ by the first condition of $\lambda_I$, and $\nu_{\ell}$ will not be detected again in future steps with probability tending to 1. By repeating the above arguments \rev{and using that $N$ is fixed}, we can show that all change points will be detected with probability tending to 1 and the rate of convergence is at least $b_{n,\ell}$ for $\nu_{\ell}$.

\section{Theoretical considerations for the componentwise monotone-invariant procedure}\label{sec:theory_MI}

In Section \ref{sec:mono_invariant}, we introduced a componentwise monotone-invariant version of our change-point detection procedure. In this approach, the test is applied to the pseudo-observations $\mathbf{U}_t \in [0,1]^p$, where \rev{$U_{t,j}$ is computed from the pooled empirical mid-distribution function, equivalently the average mid-rank, for the $j$-th coordinate.}

A fully rigorous theoretical treatment of this rank-based test under the HDMSS framework would require highly involved empirical process and higher-order U-statistic techniques that are beyond the scope of the current paper. Instead, we provide a high-level heuristic discussion in this section. \rev{Our goal is only to describe how the pooled mid-ranks change the underlying U-statistic order and why double-centering might attenuate part of the resulting estimation effect; no limiting theorem is asserted.}

\subsection{The oracle transform and copula space}
To build intuition, it is instructive to first consider an ``oracle'' version of the test. Suppose that under the null hypothesis, the true, continuous marginal CDFs $F_1, \dots, F_p$ were known. The oracle pseudo-observations would be given by $\mathbf{U}_t^{\circ} = (F_1(X_{t,1}), \dots, F_p(X_{t,p}))^\top$. By Sklar's theorem, the joint distribution of $\mathbf{U}_t^{\circ}$ corresponds exactly to the copula $C$ of the original random vector $X_t$.

If our test statistic were evaluated on these oracle pseudo-observations, it would remain a standard degree-2 U-statistic, as the distance between two points depends only on those two points. A key theoretical advantage of this oracle transform is that the pseudo-observations are deterministically bounded within the unit hypercube $[0,1]^p$. Consequently, the pairwise distances are uniformly bounded. The asymptotic null distribution would follow directly from Theorem \ref{theorem2}, depending solely on the underlying copula $C$ and the metric $\gamma$, rendering the oracle test strictly margin-free.

\subsection{The empirical transform and the change in U-statistic order}
In practice, the true marginal CDFs are unknown, and \rev{we use the pooled empirical mid-distribution functions $\widehat F_j^{\mathrm{mid}}$ defined in Section~\ref{sec:mono_invariant}.} Using \rev{$\widehat F_j^{\mathrm{mid}}$} instead of $F_j$ introduces a sample-dependent estimation error. From a structural standpoint, this substitution fundamentally alters the nature of the test statistic. Because the empirical \rev{rank transform is computed from} the entire sample, evaluating the pairwise distance $\gamma(\mathbf{U}_s, \mathbf{U}_t)$ inherently couples the observations. 

For example, under the $L_1$-based metric, computing the distance between empirical ranks effectively involves counting the number of observations $X_{l,j}$ that fall between $X_{s,j}$ and $X_{t,j}$. Because the core distance kernel now evaluates triplets $(X_s, X_t, X_l)$, this parameter substitution conceptually elevates the underlying estimator from a standard degree-2 U-statistic to a higher-order generalized U-statistic. 

In standard asymptotic theory, substituting a $\sqrt{n}$-consistent estimator like the empirical CDF into a U-statistic kernel introduces a leading-order perturbation that typically alters the limiting distribution. If one were to bound this perturbation uniformly using standard empirical process bounds, the accumulated error could scale with the dimension $p$, potentially diverging in the HDMSS regime where $p \to \infty$.

\subsection{Heuristic attenuation via double-centering}
It is possible, however, that the impact of such marginal perturbations is mitigated by the specific structural properties of the test statistic. The generalized energy distance operates as a contrast between intra-sample and inter-sample distances, utilizing double-centered weights whose global sum and row sums are exactly zero. Because the empirical mid-ranks are computed using the pooled full sample, the marginal transformation applies a symmetric perturbation to all observations. Heuristically, these zero-sum constraints may act to difference out the leading-order main effects of this marginal estimation error under the null hypothesis. We emphasize that formally establishing this attenuation uniformly over the entire sequence of split points $k$ in the HDMSS regime where $p \to \infty$ is highly non-trivial and remains an open theoretical question. Nevertheless, this double-centering mechanism provides a plausible conceptual rationale for why the empirical rank-based test is observed to maintain valid Type I error control in our numerical studies.

\subsection{Behavior under the alternative}
Under the single change-point alternative, the sequence is no longer identically distributed. \rev{For continuous marginals, $\widehat F_j^{\mathrm{mid}}$ has the same limit as} the pooled mixture CDF, $G_j(x) = \zeta F_j^{-}(x) + (1-\zeta) F_j^{+}(x)$, where $\zeta$ is the true change-point proportion and $F_j^{-}, F_j^{+}$ are the pre- and post-change marginal CDFs. \rev{Whether the transformed distributions retain enough separation for detection depends on the alternative and the metric.}

\rev{When the transformed pre- and post-change distributions remain separated under a sufficiently regular metric, one may heuristically expect a diverging signal-to-noise ratio. This observation is only motivation: neither the location result in Theorem~\ref{alt:consistency} nor the Seeded NOT result in Theorem~\ref{mult_consistency_SBS} is proved here for the empirical-rank statistic.} We leave the formal theoretical investigation of this higher-order rank-based scan statistic in the HDMSS regime as an interesting avenue for future research.

\section{Technical Appendix}\label{tech appx}
\begin{proof}[Proof of Proposition \ref{prop:componentwise_invariance}]
For each coordinate $j$, strict monotonicity implies that the order relations and ties among $\{Y_{t,j}\}_{t=1}^n$ coincide with those among
$\{X_{t,j}\}_{t=1}^n$. Hence $U_{t,j}(Y)=U_{t,j}(X)$ for all $t,j$, so the collections $\{\mathbf U_t(Y)\}$ and $\{\mathbf U_t(X)\}$ are identical. The conclusion follows because $M_n^{\mathrm{MI}}$ and $\widehat\nu^{\mathrm{MI}}$ are deterministic functions of $\{\mathbf U_t\}_{t=1}^n$.
\end{proof}

\begin{proof}[Proof of Lemma \ref{lemma_cusum}]
For the first part, simply note that some direct calculations yield
\begin{align}\label{cusum embed 1}
\begin{split}
S_k \;&=\; \frac{1}{\sqrt{n}} \dis \sum_{t=1}^k \left\{ \phi(X_t) - \frac{1}{n} \sum_{t=1}^n \phi(X_t) \right\} \;=\; \frac{1}{\sqrt{n}} \left\{\dis \sum_{t=1}^k \phi(X_t) \,-\, \frac{k}{n}\dis \sum_{t=1}^k \phi(X_t) \,-\, \frac{k}{n}\dis \sum_{t=k+1}^n \phi(X_t)  \right\}\\
&=\; \frac{k(n-k)}{n^{3/2}} \;\left\{\frac{1}{k}\dis \sum_{t=1}^k \phi(X_t) \,-\, \frac{1}{n-k}\dis \sum_{t=k+1}^n \phi(X_t) \right\}\,.
\end{split}
\end{align}
For the second part, note that following equation (\ref{embed eqn 2}) in the main paper, the expression of $S_k$ in equation (\ref{cusum embed 1}) and some elementary calculations, we can write
\begin{align*}
\;\frac{n^3}{k^2\,(n-k)^2}\; \Vert S_k \Vert^2 \;=&\; \frac{1}{k^2} \dis \sum_{t,t'=1}^k \langle \,\phi(X_t) \,,\,\phi(X_{t'})\, \rangle_{\cal{H}} \,+\, \frac{1}{(n-k)^2} \dis \sum_{t,t'=k+1}^n \langle \,\phi(X_t) \,,\,\phi(X_{t'})\, \rangle_{\cal{H}} 
\\&-\; \frac{2}{k(n-k)} \dis \sum_{t=1}^k \sum_{t'=k+1}^n \langle \,\phi(X_t) \,,\,\phi(X_{t'})\, \rangle_{\cal{H}}\\
=&\; \frac{2}{k(n-k)} \dis \sum_{t=1}^k \sum_{t'=k+1}^n \gamma(X_t,X_{t'}) 
- \frac{1}{k^2} \dis \sum_{t,t'=1}^k \gamma(X_t,X_{t'}) 
\\&- \frac{1}{(n-k)^2} \dis \sum_{t,t'=k+1}^n \gamma(X_t,X_{t'})\,.
\end{align*}
\end{proof}

\begin{proof}[Proof of Proposition \ref{Prop 4.1 in CZ}]
    The proof is identical to the proof of Proposition 4.1 in Chakraborty and Zhang (2021). Therefore, we omit it here.
\end{proof}

\begin{proof}[Proof of Theorem \ref{thm:perm_validity}]
\rev{Let $\mathbf{X}^{(0)} = \mathbf{X}$ denote the observed dataset. Let $\mathbf{X}^{(1)}, \dots, \mathbf{X}^{(B_{\mathrm{perm}})}$ denote $B_{\mathrm{perm}}$ independently and uniformly sampled random permutations of $\mathbf{X}$.}
Under the null hypothesis $H_0$, the observations $X_1, \dots, X_n$ are i.i.d., which implies that their joint distribution is exchangeable. That is, for any permutation $\pi$ of the indices $\{1, \dots, n\}$,
\[
(X_1, \dots, X_n) \stackrel{d}{=} (X_{\pi(1)}, \dots, X_{\pi(n)}).
\]
\begin{revblock}
Condition on the unordered set of observed values (the orbit) $\mathcal{O}=\{\{X_1,\dots,X_n\}\}$. Under $H_0$, the observed ordering is uniform on this orbit; applying independent uniform permutations gives conditionally exchangeable datasets $\mathbf X^{(0)},\dots,\mathbf X^{(B_{\mathrm{perm}})}$. Hence the statistics $T^{(j)}:=\mathcal T_n(\mathbf X^{(j)})$, $j=0,\dots,B_{\mathrm{perm}}$, are conditionally exchangeable.

Let $q_\alpha:=\lfloor\alpha(B_{\mathrm{perm}}+1)\rfloor$ and, for each label $j$, define the upper rank
\[
R_j:=\sum_{l=0}^{B_{\mathrm{perm}}}\mathbbm{1}\{T^{(l)}\ge T^{(j)}\}.
\]
Conditional on $\mathcal O$ and on the unordered multiset of statistic values, label $0$ is uniform among the $B_{\mathrm{perm}}+1$ labels. For any fixed multiset, at most $q_\alpha$ labels can satisfy $R_j\le q_\alpha$; ties can only reduce this number. Since the algorithm rejects exactly when $R_0\le q_\alpha$,
\[
P\!\left(R_0\le q_\alpha\mid\mathcal O\right)
\le \frac{q_\alpha}{B_{\mathrm{perm}}+1}
\le \alpha.
\]
Integrating over $\mathcal O$ proves the unconditional level bound without requiring distinct statistic values.
\end{revblock}
\end{proof}

\begin{proof}[Proof of Theorem \ref{sup_theorem1}]
To establish the uniform weak convergence of $\frac{\sqrt{2}}{n \sqrt{V_0}}\, S_n(a,b)$, by Theorem 10.2 in Pollard\,(1990) we need to show
\begin{enumerate}
\item[T1.] the finite-dimensional convergence, viz. $$ \Big(\frac{\sqrt{2}}{n \sqrt{V_0}}\, S_n(a_1,b_1)\,,\, \dots \,,\, \frac{\sqrt{2}}{n \sqrt{V_0}}\, S_n(a_s,b_s)\Big) \;\rightsquigarrow \; \Big( Q(a_1, b_1), \, \dots \,, Q(a_s,b_s) \Big)$$ as\, $n,p \to \infty$\, for fixed\, $0 \leq a_i < b_i \leq 1$, $1\leq i \leq s$, and
\item[T2.] asymptotic stochastic equicontinuity of $\frac{\sqrt{2}}{n \sqrt{V_0}}\, S_n(a,b)$ on $[0,1]^2$, viz. for any $x>0$, $$\lim_{\delta \downarrow 0}\, \limsup_{n,p \to \infty} \, P\left( \sup_{\Vert (a,b) - (c,d) \Vert \leq \delta} \left\vert \frac{\sqrt{2}}{n \sqrt{V_0}}\, S_n(a,b)\,-\,\frac{\sqrt{2}}{n \sqrt{V_0}}\, S_n(c,d) \right\vert \right) \;=\; 0\,.$$
\end{enumerate}
To prove T1, we will consider the case of $s=2$, and the general case can be proved in a similar fashion. By Cram{\'e}r-Wold theorem, it is equivalent to prove 
\begin{align}\label{eq_sup_14}
\alpha_1 \, \frac{\sqrt{2}}{n \sqrt{V_0}}\, S_n(a_1,b_1) \,+\, \alpha_2 \, \frac{\sqrt{2}}{n \sqrt{V_0}}\, S_n(a_2,b_2) \;\overset{d}{\longrightarrow} \; \alpha_1\,Q(a_1,b_1) \,+\, \alpha_2\,Q(a_2,b_2)
\end{align}
for any fixed $\alpha_1, \alpha_2 \in {\bb R}$, as $n,p \to \infty$. As $0 \leq a_i < b_i \leq 1$, $i=1,2$, we consider the following three cases : i) $a_1 \leq a_2 \leq b_2 \leq b_1$, ii) $a_1 \leq a_2 \leq b_1 \leq b_2$, and iii) $a_1 \leq b_1 \leq a_2 \leq b_2$. We will prove T1 and T2 for case (ii), and similar arguments can prove them for the other two cases.  

\begin{proof}[Proof of T1]
We can write 
\begin{align}\label{eq_sup_15}
\begin{split}
&\; \alpha_1 \, \frac{\sqrt{2}}{n \sqrt{V_0}}\, S_n(a_1,b_1) \,+\, \alpha_2 \, \frac{\sqrt{2}}{n \sqrt{V_0}}\, S_n(a_2,b_2)\\
= &\; \frac{\sqrt{2}}{n \sqrt{V_0}}\, \left\{ \alpha_1 \, \sum_{i=\lfloor na_1 \rfloor +2}^{\lfloor nb_1 \rfloor}\sum_{j=\lfloor na_1 \rfloor +1}^{i-1} \, H(X_i, X_j) \,+\,\alpha_2 \, \sum_{i=\lfloor na_2 \rfloor +2}^{\lfloor nb_2 \rfloor}\sum_{j=\lfloor na_2 \rfloor +1}^{i-1} \, H(X_i, X_j)  \right\} \\=& \; \sum_{i=\lfloor na_1 \rfloor +2}^{\lfloor nb_2 \rfloor} \widetilde{\xi}_{n,i} \,,
\end{split}
\end{align}
where 
\vspace{-0.2in}
\begin{align}\label{eq_sup_16}
\begin{split}
\widetilde{\xi}_{n,i} \;&:=\;\frac{\sqrt{2}}{n \sqrt{V_0}}\, \begin{cases}
 \alpha_1 \,\xi_{1,i} & \textrm{if}\;\; \lfloor na_1 \rfloor +2 \leq i \leq \lfloor na_2 \rfloor + 1,\\
\alpha_1 \,\xi_{1,i} + \alpha_2\, \xi_{2,i} &  \textrm{if}\;\; \lfloor na_2 \rfloor +2 \leq i \leq \lfloor nb_1 \rfloor,\\
\alpha_2 \,\xi_{2,i} & \textrm{if}\;\; \lfloor nb_1 \rfloor +1 \leq i \leq \lfloor nb_2 \rfloor,
\end{cases}
\end{split}
\end{align}
with $\xi_{1,i}=\sum_{j=\lfloor na_1 \rfloor +1}^{i-1} H(X_i, X_j)$ and $\xi_{2,i}=\sum_{j=\lfloor na_2 \rfloor +1}^{i-1}H(X_i, X_j).$ Define $\cal{F}_i=\sigma(X_i, X_{i-1},\dots)$. By Theorem 3.2 and Corollary 3.1 in Hall and Heyde\,(1980), it suffices to show
\begin{enumerate}
\item[P1.] For each $n\geq 1$, $\{\sum_{l=2}^i \widetilde{\xi}_{n,\lfloor na_1 \rfloor + l}, \cal{F}_i\}_{i=2}^{\lfloor nb_2 \rfloor - \lfloor na_1 \rfloor}$ is a square integrable mean-zero martingale sequence;
\item[P2.] $V_n := \sum_{i=2}^{\lfloor nb_2 \rfloor - \lfloor na_1 \rfloor} \E\,\big[ \,\widetilde{\xi}^2_{n,\lfloor na_1 \rfloor + i} \,\vert\, \cal{F}_{\lfloor na_1 \rfloor + i -1} \big] \,\overset{P}{\longrightarrow} \, \alpha_1^2\,(b_1-a_1)^2\,+\,\alpha_2^2\,(b_2-a_2)^2\,+\, 2\,\alpha_1\,\alpha_2\,(b_1-a_2)^2$, as\, $n,p \to \infty$;
\item[P3.] $\sum_{i=\lfloor na_1 \rfloor + 2}^{\lfloor nb_2 \rfloor} \, \E\,\big[\, \widetilde{\xi}_{n,i}^4 \big] \longrightarrow \,0$,\, as\, $n,p \to \infty$.
\end{enumerate}

From Theorem 3.2 in Hall and Heyde\,(1980), the variance of $\alpha_1\,Q(a_1,b_1) \,+\, \alpha_2\,Q(a_2,b_2)$ should be $\plim_{n,p \to \infty} V_n$ as in P2. From there, it is intuitive that 
\begin{align*}
\cov\,\big( Q(a_1,b_1)\,,\, Q(a_2,b_2) \big) \;=\; \big(b_1 \land b_2 \,-\, a_1 \lor a_2\big)^2 \,\, \mathbbm{1}\big(b_1 \land b_2 > a_1 \lor a_2 \big)\,.
\end{align*}
To show P1, it is easy to see that $\widetilde{\xi}_{n,\lfloor na_1 \rfloor + l}$ is square integrable, $\E\,\big(\widetilde{\xi}_{n,\lfloor na_1 \rfloor + l}\big) = 0$ and $\cal{F}_2 \subseteq \cal{F}_l$.
Moreover, $\mathbb{E}[\sum_{l=2}^i \widetilde{\xi}_{n,\lfloor na_1 \rfloor + l}|\mathcal{F}_{i'}]=\sum_{l=2}^{i'} \widetilde{\xi}_{n,\lfloor na_1 \rfloor + l}$ for $i\geq i'$ using the double-centering property of $H(X_i,X_j)$.
To prove P3, note that using the power mean inequality
\begin{align}\label{power mean ineq}
 \left|\dis \sum_{i=1}^n a_i \right|^r \; \leq \; n^{r-1} \, \dis \sum_{i=1}^n |a_i|^r 
\end{align}
for \,$a_i \in \mathbb{R},\, 1\leq i \leq n, \, n\geq 2$ and $r>1$, we can write 
\begin{align}\label{eq_sup_17}
\begin{split}
&\sum_{i=\lfloor na_1 \rfloor + 2}^{\lfloor nb_2 \rfloor} \, \E\,\big[\, \widetilde{\xi}_{n,i}^4 \big] 
\\=&\;\; \sum_{i=\lfloor na_1 \rfloor + 2}^{\lfloor na_2 \rfloor +1} \, \E\,\big[\, \widetilde{\xi}_{n,i}^4 \big] \,+\, \sum_{i=\lfloor na_2 \rfloor + 2}^{\lfloor nb_1 \rfloor} \, \E\,\big[\, \widetilde{\xi}_{n,i}^4 \big]\,+\, \sum_{i=\lfloor nb_1 \rfloor + 1}^{\lfloor nb_2 \rfloor} \, \E\,\big[\, \widetilde{\xi}_{n,i}^4 \big] \\
=&\;\; \sum_{i=\lfloor na_1 \rfloor + 2}^{\lfloor na_2 \rfloor +1} \, \E\,\left[\left(\alpha_1 \,\frac{\sqrt{2}}{n \sqrt{V_0}}\,\, \xi_{1,i}\right)^4\right] \,+\,\sum_{i=\lfloor na_2 \rfloor + 2}^{\lfloor nb_1 \rfloor} \, \E\,\left[\left(\, \alpha_1 \,\frac{\sqrt{2}}{n \sqrt{V_0}}\,\, \xi_{1,i} \,+\,\alpha_2 \,\frac{\sqrt{2}}{n \sqrt{V_0}}\,\, \xi_{2,i} \right)^4\right]\\
& +\, \sum_{i=\lfloor nb_1 \rfloor + 1}^{\lfloor nb_2 \rfloor} \, \E\,\left[\left(\alpha_2 \,\frac{\sqrt{2}}{n \sqrt{V_0}}\,\, \xi_{2,i} \right)^4\right] \\ 
\lesssim& \;\; \frac{1}{n^4\,V_0^2} \left(\,\alpha_1^4\, \sum_{i=\lfloor na_1 \rfloor + 2}^{\lfloor nb_1 \rfloor} \, \E\, \big[\xi_{1,i}^4\big] \;+\; \alpha_2^4\, \sum_{i=\lfloor na_2 \rfloor + 2}^{\lfloor nb_2 \rfloor} \, \E\, \big[\xi_{2,i}^4\big] \right)\,.
\end{split}
\end{align}
We have essentially used the definitions in (\ref{eq_sup_16}) in the above calculations. Now for the first summand in the RHS of (\ref{eq_sup_17}), using (\ref{eq_sup_16}), we have
\begin{align}\label{eq_sup_18}
\begin{split}
&\qquad\frac{1}{n^4\,V_0^2}  \sum_{i=\lfloor na_1 \rfloor + 2}^{\lfloor nb_1 \rfloor} \, \E\, \big[\xi_{1,i}^4\big]  \;=\; \frac{1}{n^4\,V_0^2}  \sum_{i=\lfloor na_1 \rfloor + 2}^{\lfloor nb_1 \rfloor} \, \E\,\left[\left( \sum_{j=\lfloor na_1 \rfloor +1}^{i-1} \, H(X_i, X_j) \,\right)^4\right] \\
&=\; \frac{1}{n^4\,V_0^2}  \sum_{i=\lfloor na_1 \rfloor + 2}^{\lfloor nb_1 \rfloor} \, \left\{\, \sum_{j=\lfloor na_1 \rfloor +1}^{i-1} \, \E[H^4(X_i, X_j)]+\; 3 \sum_{\lfloor na_1 \rfloor +1 \leq j_1\neq j_2\leq i-1} \E[H^2(X_i, X_{j_1})\,H^2(X_i, X_{j_2})]\right\}\\
&= \; \frac{1}{n^4}\,\, O\left( \frac{n^2\,\E[H^4(X, X') ]+\, n^3\, \E[H^2(X, X')H^2(X, X'')]}{(\E[H^2(X, X')])^2} \right).
\end{split}
\end{align}
Similar expressions hold for the second summand in the RHS of (\ref{eq_sup_17}). With this, it is easy to see that under Assumption \ref{ass1_new}, $$\sum_{i=\lfloor na_1 \rfloor + 2}^{\lfloor nb_2 \rfloor} \, \E\,\big[\, \widetilde{\xi}_{n,i}^4 \big] \;=\;o(1) \qquad \textrm{as}\;\; n,p \to \infty\,,$$
which completes the proof of P3. To prove P2, write 
\vspace{-0.1in}
\begin{align}\label{eq_sup_19}
\begin{split}
V_n \;&=\; \sum_{i=2}^{\lfloor nb_2 \rfloor - \lfloor na_1 \rfloor} \E\,\big[ \,\widetilde{\xi}^2_{n,\lfloor na_1 \rfloor + i} \,\vert\, \cal{F}_{\lfloor na_1 \rfloor + i -1} \big] \;=\;  \sum_{l=\lfloor na_1 \rfloor + 2}^{\lfloor nb_2 \rfloor} \E\,\big[ \,\widetilde{\xi}^2_{n,l} \,\vert\, \cal{F}_{l -1} \big]
\end{split}
\end{align}
where we have simply substituted $l=\lfloor na_1 \rfloor + i$. From (\ref{eq_sup_19}) we have
\begin{align}\label{eq_sup_20}
\begin{split}
V_n \;=&\; \sum_{l=\lfloor na_1 \rfloor + 2}^{\lfloor na_2 \rfloor +1} \E\,\left[ \,\left(\frac{\sqrt{2}}{n \sqrt{V_0}}\,\,\alpha_1\, \xi_{1,l}\right)^2 \,\Bigg\vert\, \cal{F}_{l -1} \right] +\; \sum_{l=\lfloor nb_1 \rfloor + 1}^{\lfloor nb_2 \rfloor} \E\,\left[ \,\left(\frac{\sqrt{2}}{n \sqrt{V_0}}\,\,\alpha_2\, \xi_{2,l}\right)^2 \,\Bigg\vert\, \cal{F}_{l -1} \right] \\
& \,+\, \sum_{l=\lfloor na_2 \rfloor + 2}^{\lfloor nb_1 \rfloor} \E\,\left[ \,\left\{\frac{\sqrt{2}}{n \sqrt{V_0}}\,\,\big(\alpha_1\, \xi_{1,l} \,+\, \alpha_2\, \xi_{2,l})\right\}^2 \,\Bigg\vert\, \cal{F}_{l -1} \right]\\
=&\; \frac{2}{n^2 V_0}\, \Bigg( \alpha_1^2\, \sum_{l=\lfloor na_1 \rfloor + 2}^{\lfloor nb_1 \rfloor} \E\,\big[\xi_{1,l}^2\,\big\vert\, \cal{F}_{l -1} \big] \,+\, \alpha_2^2\, \sum_{l=\lfloor na_2 \rfloor + 2}^{\lfloor nb_2 \rfloor} \E\,\big[\xi_{2,l}^2\,\big\vert\, \cal{F}_{l -1} \big] 
\\&\qquad \quad +\, 2\,\alpha_1\,\alpha_2\, \sum_{l=\lfloor na_2 \rfloor + 2}^{\lfloor nb_1 \rfloor} \E\,\big[\xi_{1,l}\,\xi_{2,l}\,\big\vert\, \cal{F}_{l -1} \big]\Bigg)\\
=&\; \alpha_1^2\, V_{1n} \,+\, \alpha_2^2\, V_{2n} \,+\, 2\,\alpha_1\alpha_2\, V_{3n}\,,
\end{split}
\end{align}
where 
\begin{align}\label{eq_sup_21}
\begin{split}
V_{1n} \;&=\; \frac{2}{n^2 V_0}\,\sum_{l=\lfloor na_1 \rfloor + 2}^{\lfloor nb_1 \rfloor} \E\,\big[\xi_{1,l}^2\,\big\vert\, \cal{F}_{l -1} \big]\,,\\
V_{2n} \;&=\; \frac{2}{n^2 V_0}\,\sum_{l=\lfloor na_2 \rfloor + 2}^{\lfloor nb_2 \rfloor} \E\,\big[\xi_{2,l}^2\,\big\vert\, \cal{F}_{l -1} \big]\,,\\
V_{3n} \;&=\; \frac{2}{n^2 V_0}\,\sum_{l=\lfloor na_2 \rfloor + 2}^{\lfloor nb_1 \rfloor} \E\,\big[\xi_{1,l}\,\xi_{2,l}\,\big\vert\, \cal{F}_{l -1} \big]\,.
\end{split}
\end{align}
Using the definition of $\xi_{1,l}$ from (\ref{eq_sup_16}), we can write
\begin{align}\label{eq_sup_21.5}
V_{1n} \;&=\; \frac{2}{n^2 V_0}\,\sum_{l=\lfloor na_1 \rfloor + 2}^{\lfloor nb_1 \rfloor} \,\sum_{j_1,j_2=\lfloor na_1 \rfloor + 1}^{l-1} \E\,\big[ H(X_l,X_{j_1})\,H(X_l,X_{j_2}) \,\big\vert\, \cal{F}_{l -1} \big]\,,
\end{align}
and therefore
\begin{align}\label{eq_sup_22}
\E\,[V_{1n}] \;&=\; \frac{2}{n^2 V_0}\,\sum_{l=\lfloor na_1 \rfloor + 2}^{\lfloor nb_1 \rfloor}\, \sum_{j=\lfloor na_1 \rfloor + 1}^{l-1} \E\,\big[ H^2(X_l,X_j)\big]\,,
\end{align}
as \,$\E\,\big[ H(X_l,X_{j_1})\,H(X_l,X_{j_2}) \big] = 0$\, for\, $j_1 \neq j_2$. Using the fact that $V_0 = \E\,\big[ H^2(X,X')\big]$, some straightforward calculations yield
\begin{align}\label{eq_sup_23}
\begin{split}
\E\,[V_{1n}] \;&=\; \frac{2}{n^2 V_0}\,\sum_{\lfloor na_1 \rfloor + 1 \leq j < l \leq \lfloor nb_1 \rfloor} \E\,\big[ H^2(X,X')\big]\;=\;  \frac{2}{n^2}\,\binom{\lfloor nb_1 \rfloor - \lfloor na_1 \rfloor}{2}\\
&=\; \frac{1}{n^2}\, \big(\lfloor nb_1 \rfloor - \lfloor na_1 \rfloor)\,\big(\lfloor nb_1 \rfloor - \lfloor na_1 \rfloor - 1)\\
& \rightarrow \; (b_1 - a_1)^2 \,,
\end{split}
\end{align}
as \,$n \to \infty$. Define $L_l(j_1,j_2) := \E\,\big[ H(X_l,X_{j_1})\,H(X_l,X_{j_2}) \,\big\vert\, \cal{F}_{l -1} \big]$. Then from (\ref{eq_sup_21.5}) we can write
\begin{align*}
V_{1n} \;&=\; \frac{2}{n^2 V_0}\,\sum_{l=\lfloor na_1 \rfloor + 2}^{\lfloor nb_1 \rfloor} \,\sum_{j_1,j_2=\lfloor na_1 \rfloor + 1}^{l-1} L_l(j_1,j_2)\,,
\end{align*}
and therefore
\begin{align*}
\var \,(V_{1n}) \;&=\; \frac{4}{n^4 V_0^2}\,\sum_{l,l'=\lfloor na_1 \rfloor + 2}^{\lfloor nb_1 \rfloor} \,\sum_{j_1,j_2=\lfloor na_1 \rfloor + 1}^{l-1}\,\sum_{j_1',j_2'=\lfloor na_1 \rfloor + 1}^{l'-1} \cov\, \big( L_l(j_1,j_2), L_{l'}(j_1',j_2')\big)\,.
\end{align*}
Following the proof of Lemma D.1 in the Supplementary Materials of Chakraborty and Zhang\,(2021), we have $\E\,L_l(j_1,j_2) = 0$\, for\, $j_1 \neq j_2$, and 
\begin{align*}
& \E\, \big[L_l(j_1,j_2) \, L_{l'}(j_1',j_2')\big] \\=& \begin{cases}
 \E\,\left[H^2(X_l,X_{j_1})\,H^2(X_{l'}',X_{j_1})\right] & \textrm{if} \;\;\; j_1=j_2=j_1'=j_2'\,,\\
\E\,\left[H(X_l,X_{j_1})\,H(X_l,X_{j_2})\,H(X_{l'}',X_{j_1})\,H(X_{l'}',X_{j_2})\right]  & \textrm{if} \;\;\; j_1=j_1'\neq j_2=j_2' \;\; \textrm{or} \;\; j_1=j_2'\neq j_1'=j_2\,,\\
 \E\,\left[H^2(X_l,X_{j_1})\right] \E\,\left[H^2(X_{l'},X_{j_1'})\right]  &  \textrm{if} \;\;\; j_1=j_2 \neq j_1'=j_2'\,,
\end{cases}
\end{align*}
where the above expression holds for $l=l'$ as well. Therefore
\begin{align*}
\var\,(V_{1n}) \;=&\; \frac{4}{n^4 V_0^2}\,\Bigg[ \dis \sum_{l=l'} \Bigg\{\sum_{j_1=\lfloor na_1 \rfloor + 1}^{l-1} \cov\,\big(H^2(X_l,X_{j_1}),H^2(X_l',X_{j_1})\big) \,\\
&+\, 2\sum_{\lfloor na_1 \rfloor + 1 \leq j_1\neq j_2\leq l-1}\E\,\big[ H(X_l,X_{j_1})\,H(X_l,X_{j_2})\,H(X_l',X_{j_1})\,H(X_l',X_{j_2}) \big]\,\Bigg\}\\
&  + \; 2\sum_{\lfloor na_1 \rfloor + 2 \leq l < l' \leq \lfloor nb_1 \rfloor} \Bigg\{\sum_{j_1=\lfloor na_1 \rfloor + 1}^{l-1} \cov\,\big(H^2(X_l,X_{j_1}),H^2(X_{l'}',X_{j_1})\big) \\
&  +\,  2\sum_{\lfloor na_1 \rfloor + 1 \leq j_1\neq j_2\leq l-1} \E\,\big[H(X_l,X_{j_1})\,H(X_l,X_{j_2})\,H(X_{l'}',X_{j_1})\,H(X_{l'}',X_{j_2})\big]\,\Bigg\}\Bigg]\,.
\end{align*}
This implies 
\begin{align}\label{eq_sup_24}
\begin{split}
\var\,(V_{1n})=&\frac{1}{V_0^2}\, O\Big(n^{-1}\E\,\big[H^2(X, X')\, H^2(X, X'') \big] \,
\\&+\, \, \E\,\big[ H(X, X'')\, H(X', X'')\, H(X, X''')\, H(X', X''') \big]\, \Big)=\; o(1)\,,
\end{split}
\end{align}
as \,$n,p \to \infty$, under Assumption \ref{ass1_new}. Combining (\ref{eq_sup_23}) and (\ref{eq_sup_24}), we get
\begin{align*}
\E\, \left[\Big(V_{1n} - (b_1-a_1)^2 \Big)^2\right] \;&=\; \var\,(V_{1n}) \,+\, \Big\{\E\,[V_{1n}] - (b_1-a_1)^2 \Big\}^2=\; o(1)\,,
\end{align*}
which, combined with Chebyshev's inequality, implies
\begin{align}\label{eq_sup_25}
V_{1n} \;& \overset{P}{\longrightarrow}\; (b_1-a_1)^2 \quad \textrm{as} \;\;\; n,p \to \infty\,.
\end{align}
Likewise it can be shown that as\, $n,p \to \infty$,
\begin{align}\label{eq_sup_26}
V_{2n} \;& \overset{P}{\longrightarrow}\; (b_2-a_2)^2 \quad \textrm{and} \quad V_{3n} \; \overset{P}{\longrightarrow}\; (b_1-a_2)^2\,.
\end{align}
Combining (\ref{eq_sup_25}) and (\ref{eq_sup_26}), we get from (\ref{eq_sup_20})
\begin{align}\label{eq_sup_27}
V_n \;& \overset{P}{\longrightarrow}\; \alpha_1^2\,(b_1-a_1)^2\,+\,\alpha_2^2\,(b_2-a_2)^2\,+\, 2\,\alpha_1\,\alpha_2\,(b_1-a_2)^2\,.
\end{align}
This completes the proof of P2 and, thereby, the proof of T1, i.e., the finite-dimensional convergence.
\end{proof}

\begin{proof}[Proof of T2]
Denote $u=(a,b)$ and $v=(c,d)$. Also define \,$W_n(u) := \frac{\sqrt{2}}{n \sqrt{V_0}}\, S_n(u)$\, for\, $u \in [0,1]^2$. To prove the stochastic equicontinuity of $W_n(u)$ for $u \in [0,1]^2$, we need to show for any $\epsilon >0$ 
\begin{align*}
\lim_{\delta \downarrow 0}\, \limsup_{n,p \to \infty} \, P\left(\, \sup_{\substack{u,v \,\in\,[0,1]^2 \\ \kappa(u,v) < \delta}} \left\vert W_n(u)\,-\,W_n(v) \right\vert\, \right) \;=\; 0\,,
\end{align*}
where $\big([0,1]^2,\kappa \big)$ is compact. By Theorem A.8 in Li and Racine\,(2007) (also see Theorem 3 of Wichura (1969), which is applicable to a martingale sequence), it suffices to show that $\forall\, u,v \in [0,1]^2$,
\begin{align}\label{eq_sup_28}
\E\,\big\vert W_n(u)\,-\,W_n(v) \big\vert^{\alpha} \;&\lesssim \; \kappa^{\,\gamma}(u,v)
\end{align}
for some $\alpha>0$ and $\gamma>1$. For our purpose, we choose $\kappa(u,v) = \Vert u-v \Vert_1^{1/2}$\, for \, $u,v \in [0,1]^2$. Note that $[0,1]^2 \subseteq {\bb R}^2$ is compact (closed and bounded) with respect to the metric $\rho(u,v) = \Vert u-v \Vert_1$. It is easy to verify that $[0,1]^2$ is closed and bounded (and hence compact) with respect to the metric $\kappa(u,v) = \rho^{1/2}(u,v)$ as well.

Choosing $\alpha=2$ and $\gamma=2$, we will prove that $\forall\, u,v \in [0,1]^2$,
\begin{align}\label{eq_sup_28.1}
\E\,\big\vert W_n(u)\,-\,W_n(v) \big\vert^2 \;&\lesssim \; \kappa^{\,2}(u,v)\,,
\end{align}
which will complete the proof. Towards that end, consider the case\, $a<c<d<b$. We will show that (\ref{eq_sup_28.1}) holds in this case, and similar arguments will do the job for the other cases. Observe that
\begin{align}\label{eq_sup_29}
\begin{split}
W_n(u)\,-\,W_n(v)  \;&= \; \frac{\sqrt{2}}{n \sqrt{V_0}}\, S_n(a,b)\,-\, \frac{\sqrt{2}}{n \sqrt{V_0}}\, S_n(c,d)\\
&=\; \frac{\sqrt{2}}{n \sqrt{V_0}}\, \Bigg[ \sum_{i=\lfloor na \rfloor +2}^{\lfloor nb \rfloor} \sum_{j=\lfloor na \rfloor + 1}^{i-1} H(X_i, X_j)\,-\, \sum_{i=\lfloor nc \rfloor +2}^{\lfloor nd \rfloor} \sum_{j=\lfloor nc \rfloor + 1}^{i-1} H(X_i, X_j) \Bigg]\\
&=\; \frac{\sqrt{2}}{n \sqrt{V_0}}\,\Bigg[ \sum_{i=\lfloor na \rfloor +2}^{\lfloor nc \rfloor} \sum_{j=\lfloor na \rfloor + 1}^{i-1} H(X_i, X_j)\,+\, \sum_{i=\lfloor nc \rfloor +1}^{\lfloor nd \rfloor} \sum_{j=\lfloor na \rfloor + 1}^{\lfloor nc \rfloor} H(X_i, X_j)\\
& \hspace{0.8in} \sum_{i=\lfloor nd \rfloor +1}^{\lfloor nb \rfloor} \sum_{j=\lfloor na \rfloor + 1}^{\lfloor nc \rfloor} H(X_i, X_j) \,+\, \sum_{i=\lfloor nd \rfloor +1}^{\lfloor nb \rfloor} \sum_{j=\lfloor nc \rfloor + 1}^{\lfloor nd \rfloor} H(X_i, X_j)\\
& \hspace{0.8in} \sum_{i=\lfloor nd \rfloor +2}^{\lfloor nb \rfloor} \sum_{j=\lfloor nd \rfloor + 1}^{i-1} H(X_i, X_j) \Bigg]\\
&=:\; I \,+\, II\,+\, III\,+\, IV\,+\, V\,.
\end{split}
\end{align}
By power mean inequality,
\begin{align}\label{eq_sup_29.5}
\begin{split}
(I \,+\, II\,+\, III\,+\, IV\,+\, V)^2 \; \lesssim \; I^2 \,+\, II^2\,+\, III^2\,+\, IV^2\,+\, V^2\,.
\end{split}
\end{align}
Now 
\begin{align*}
\E\,[I^2] \;=\; \frac{2}{n^2 V_0}\, \sum_{i_1,i_2=\lfloor na \rfloor +2}^{\lfloor nc \rfloor} \sum_{j_1=\lfloor na \rfloor + 1}^{i_1-1} \sum_{j_2=\lfloor na \rfloor + 1}^{i_2-1} \E\,\big[ H(X_{i_1}, X_{j_1})\,H(X_{i_2}, X_{j_2})\big]\,.
\end{align*}
Clearly, $\E\,\big[ H(X_{i_1}, X_{j_1})\,H(X_{i_1}, X_{j_1})\big] = 0$ if the cardinality of the set $\{i_1,j_1\} \cap \{i_2,j_2\}$ is 0 or 1. Therefore we have 
\begin{align}\label{eq_sup_30}
\begin{split}
\E\,[I^2] \;&=\; \frac{2}{n^2 V_0}\, \sum_{i=\lfloor na \rfloor +2}^{\lfloor nc \rfloor} \sum_{j=\lfloor na \rfloor + 1}^{i-1}  \E\,\big[ H^2(X_i, X_j)\big] \;=\; \frac{2}{n^2 V_0}\, \sum_{\lfloor na \rfloor +1\, \leq j \,<i \,\leq \lfloor nc \rfloor}  V_0\\
&=\; \frac{1}{n^2}\, \big( \lfloor nc \rfloor - \lfloor na \rfloor \big)\,\big( \lfloor nc \rfloor - \lfloor na \rfloor -1 \big)\,.
\end{split}
\end{align}
Note that 
\begin{align}\label{eq_sup_31}
\begin{split}
\lfloor nc \rfloor - \lfloor na \rfloor -1 \;&\leq \; nc - na + na - \lfloor na \rfloor -1 = \; n(c-a) \,+\, (\{na\} - 1)\leq \; n(c-a)\,,
\end{split}
\end{align}
as \,$\{na\} \leq 1$. Therefore we have from (\ref{eq_sup_30}) and (\ref{eq_sup_31})
\begin{align}\label{eq_sup_32}
\begin{split}
\E\,[I^2] \;&\lesssim\; c-a\,.
\end{split}
\end{align}
Likewise, it can be shown that
\begin{align}\label{eq_sup_33}
\begin{split}
\E\,[V^2] \;&\lesssim\; b-d\,.
\end{split}
\end{align}
Now
\begin{align}\label{eq_sup_34}
\begin{split}
\E\,[II^2] \;&=\; \frac{2}{n^2 V_0}\, \sum_{i_1,i_2=\lfloor nc \rfloor +1}^{\lfloor nd \rfloor} \sum_{j_1, j_2=\lfloor na \rfloor + 1}^{\lfloor nc \rfloor}  \E\,\big[ H(X_{i_1}, X_{j_1})\,H(X_{i_2}, X_{j_2})\big]\\
&=\; \frac{2}{n^2 V_0}\, \sum_{i=\lfloor nc \rfloor +1}^{\lfloor nd \rfloor} \sum_{j=\lfloor na \rfloor + 1}^{\lfloor nc \rfloor}  \E\,\big[ H^2(X_i, X_j)\big]\\
&=\; \frac{2}{n^2} \, \big( \lfloor nd \rfloor - \lfloor nc \rfloor \big)\,\big( \lfloor nc \rfloor - \lfloor na \rfloor \big)\\
&\lesssim \; \frac{1}{n}\, \big[n(c-a) + 1\big]\\
&\lesssim \; c-a\,.
\end{split}
\end{align}
Similarly, it can be shown that
\begin{align}\label{eq_sup_35}
\begin{split}
\E\,[III^2] \;&\lesssim \; c-a \quad \textrm{and} \quad \E\,[IV^2] \;\lesssim \; b-d\,.
\end{split}
\end{align}
Combining (\ref{eq_sup_32})-(\ref{eq_sup_35}) with (\ref{eq_sup_29}) and (\ref{eq_sup_29.5}), we get
\begin{align*}
\E\,\left[\big\vert W_n(u)\,-\,W_n(v) \big\vert^2\right] \;&\lesssim \; (c-a)\,+\,(b-d) \;=\; \Vert u-v \Vert_1 \;=\; \kappa^{\,2}(u,v)\,.
\end{align*}
This proves (\ref{eq_sup_28.1}) and thereby completes the proof of T2.
\end{proof}
Combining the above results, we complete the proof of Theorem \ref{sup_theorem1}.
\end{proof}

\begin{proof}[Proof of Theorem \ref{sup_theorem2}]
Again consider the subset $[0,1]^2 \subseteq {\bb R}^2$ equipped with the metric $\kappa(u,v) = \Vert u-v \Vert_1^{1/2}$\, for\, $u,v \in [0,1]^2$. By Theorem 1 in Andrews\,(1992), we essentially need to show
\begin{enumerate}
\item[A1.] $[0,1]^2$ is totally bounded with respect to the metric $\kappa$;
\item[A2.] Pointwise convergence: $G_n(u)\, \overset{P}{\rightarrow} \,0 \;\;\;\forall\, u\in[0,1]^2$\, as\, $n,p \to \infty$;
\item[A3.] Asymptotic stochastic equicontinuity: for any $\epsilon>0$, $$\lim_{\delta \downarrow 0}\, \limsup_{n,p \to \infty} \, P\left(\, \sup_{\substack{u,v \,\in\,[0,1]^2 \\ \kappa(u,v) \leq \delta}} \left\vert G_n(u)\,-\,G_n(v) \right\vert\, \right) \;=\; 0\,.$$
\end{enumerate}
As $[0,1]^2$ is compact with respect to the metric $\kappa$, it is, therefore, totally bounded. To see A2, note that for fixed $u \in [0,1]^2$, using Chebyshev's inequality we have for any $\epsilon>0$ 
\begin{align}\label{eq_sup_36}
P\big( \vert G_n(u)\vert > \epsilon \big) \;&\leq \; \frac{1}{\epsilon^2}\,\, \E\,G_n^2(u) \;=\; \frac{1}{n^2 \epsilon^2 V_0}\,\,\E\,R_n^2(a,b)\,.
\end{align}
Recalling that $R_n(a,b) = \widetilde{R}_n(\lfloor an \rfloor +1, \lfloor bn \rfloor)$ and the definition of $\widetilde{R}_n(k,m)$, it is not hard to verify that
\begin{align*}
R_n^2(a,b)\;=\; \sum_{\lfloor na \rfloor +1 \leq i_1 < i_2 \leq \lfloor nb \rfloor} \sum_{\lfloor na \rfloor +1 \leq i_1' < i_2' \leq \lfloor nb \rfloor} \tau^2\,R(X_{i_1},X_{i_2})\,R(X_{i_1'},X_{i_2'})\,.
\end{align*}
Therefore, by H{\"o}lder's inequality, we have
\begin{align}\label{eq_sup_37}
\begin{split}
\E\,R_n^2(a,b) \;&\leq \; \sum_{\lfloor na \rfloor +1 \leq i_1 < i_2 \leq \lfloor nb \rfloor} \sum_{\lfloor na \rfloor +1 \leq i_1' < i_2' \leq \lfloor nb \rfloor} \tau^2\, \big\{ \E\,R^2(X_{i_1},X_{i_2}) \big\}^{1/2}\,\big\{ \E\,R^2(X_{i_1'},X_{i_2'}) \big\}^{1/2}\\
&= \; \tau^2\, \Big[\frac{1}{2}\,\big(\lfloor nb \rfloor - \lfloor na \rfloor \big) \,\big(\lfloor nb \rfloor - \lfloor na \rfloor -1 \big)\,\big\{\E\,R^2(X,X') \big\}^{1/2} \Big]^2\\
&=\; O\Big(n^4\,\tau^2\,\,\E\,R^2(X,X')  \Big)\\
&=\; O\Big(n^4\,\Big\{\tau^4\,\,\E\,R^4(X,X') \Big\}^{1/2} \Big)\,.
\end{split}
\end{align}
Combining (\ref{eq_sup_36}) and (\ref{eq_sup_37}), we get
\begin{align}\label{eq_sup_38}
\begin{split}
P\big( \vert G_n(u)\vert > \epsilon \big) \;&= \; O\left(\frac{n^2}{\epsilon^2\,\, \E\,H^2(X,X')}\,\left\{\tau^4\,\,\E\,R^4(X,X') \right\}^{1/2} \right)
\\&= \; O\left(\frac{1}{\epsilon^2}\,\left[\,\frac{n^4\,\tau^4\,\,\E\,R^4(X,X')}{\big(\E\,H^2(X,X')\big)^2} \,\right]^{1/2}\, \right)\,.
\end{split}
\end{align}
Under Assumption \ref{ass2_new}, it is easy to see from (\ref{eq_sup_38}) that 
\begin{align*}
P\big( \vert G_n(u)\vert > \epsilon \big) \;&= \; o(1)\,,
\end{align*}
which implies $G_n(u) \overset{P}{\rightarrow} 0$ for any fixed  $u \in [0,1]^2$\, as\, $n,p \to \infty$. This proves A2.

Finally, to prove A3, again by Theorem A.8 in Li and Racine\,(2007) (also see Theorem 3 of Wichura (1969)), it will suffice to show that $\forall \,u,v \in [0,1]^2$
\begin{align}\label{eq_sup_39}
\E\left[\,\big\vert G_n(u)\,-\,G_n(v) \big\vert^2\right] \;&\lesssim \; \kappa^{\,2}(u,v)\,.
\end{align}
We will show that (\ref{eq_sup_39}) holds in the case $a<c<d<b$. Similar arguments can prove (\ref{eq_sup_39}) for other cases. Similar to the proof of T2, we have 
\begin{align}\label{eq_sup_40}
\begin{split}
G_n(u)\,-\,G_n(v)  \;&= \; \frac{1}{n \sqrt{V_0}}\, R_n(a,b)\,-\, \frac{1}{n \sqrt{V_0}}\, R_n(c,d)\\
&=\; \frac{\tau}{n \sqrt{V_0}}\, \left\{ \sum_{i=\lfloor na \rfloor +2}^{\lfloor nb \rfloor} \sum_{j=\lfloor na \rfloor + 1}^{i-1} R(X_i, X_j)\,-\, \sum_{i=\lfloor nc \rfloor +2}^{\lfloor nd \rfloor} \sum_{j=\lfloor nc \rfloor + 1}^{i-1} R(X_i, X_j) \right\}\\
&=\; \frac{\tau}{n \sqrt{V_0}}\,\Bigg\{ \sum_{i=\lfloor na \rfloor +2}^{\lfloor nc \rfloor} \sum_{j=\lfloor na \rfloor + 1}^{i-1} R(X_i, X_j)\,+\, \sum_{i=\lfloor nc \rfloor +1}^{\lfloor nd \rfloor} \sum_{j=\lfloor na \rfloor + 1}^{\lfloor nc \rfloor} R(X_i, X_j)\\
& \hspace{0.8in} \sum_{i=\lfloor nd \rfloor +1}^{\lfloor nb \rfloor} \sum_{j=\lfloor na \rfloor + 1}^{\lfloor nc \rfloor} R(X_i, X_j) \,+\, \sum_{i=\lfloor nd \rfloor +1}^{\lfloor nb \rfloor} \sum_{j=\lfloor nc \rfloor + 1}^{\lfloor nd \rfloor} R(X_i, X_j)\\
& \hspace{0.8in} \sum_{i=\lfloor nd \rfloor +2}^{\lfloor nb \rfloor} \sum_{j=\lfloor nd \rfloor + 1}^{i-1} R(X_i, X_j) \Bigg\}\\
&=:\; I_G \,+\, II_G\,+\, III_G\,+\, IV_G\,+\, V_G\,.
\end{split}
\end{align}
By the power mean inequality,
\begin{align}\label{eq_sup_41}
\begin{split}
(I_G \,+\, II_G\,+\, III_G\,+\, IV_G\,+\, V_G)^2 \; \lesssim \; I_G^2 \,+\, II_G^2\,+\, III_G^2\,+\, IV_G^2\,+\, V_G^2\,.
\end{split}
\end{align}
Now 
\begin{align}\label{eq_sup_41.5}
\E\,[I_G^2] \;=\; \frac{\tau^2}{n^2 V_0}\, \sum_{i_1,i_2=\lfloor na \rfloor +2}^{\lfloor nc \rfloor} \sum_{j_1=\lfloor na \rfloor + 1}^{i_1-1} \sum_{j_2=\lfloor na \rfloor + 1}^{i_2-1} \E\,\big[ R(X_{i_1}, X_{j_1})\,R(X_{i_2}, X_{j_2})\big]\,.
\end{align}
Again, using H{\"o}lder's inequality and similar arguments as used in deriving (\ref{eq_sup_37}), we get from (\ref{eq_sup_41.5})
\begin{align}\label{eq_sup_42}
\begin{split}
\E\,[I_G^2] \;&=\; \frac{\tau^2}{n^2 V_0}\, \left[\sum_{i=\lfloor na \rfloor +2}^{\lfloor nc \rfloor} \sum_{j=\lfloor na \rfloor + 1}^{i-1}  \big\{\E\, R^2(X_i, X_j)\big\}^{1/2}\right]^2 
\\&=\; \frac{\tau^2}{n^2 V_0}\, \,\frac{\big( \lfloor nc \rfloor - \lfloor na \rfloor \big)^2\,\big( \lfloor nc \rfloor - \lfloor na \rfloor -1 \big)^2}{4\, n^2}\,\, n^2\,\E\, R^2(X,X').
\end{split}
\end{align}
Using the fact that $\lfloor nc \rfloor - \lfloor na \rfloor -1 \leq n(c-a)$, 
$(c-a)^2 \leq (c-a)$ and H{\"o}lder's inequality, we get from (\ref{eq_sup_42})
\begin{align}\label{eq_sup_43}
\begin{split}
\E\,[I_G^2] \;&\lesssim\; (c-a)\,\, \left(\frac{n^2\,\tau^2\,\,\E\,R^2(X,X')}{\E\,H^2(X,X')}\right) \;\leq\; (c-a)\,\, \left(\frac{n^2\,\tau^2\,\,\Big(\E\,R^4(X,X')\Big)^{1/2}}{\E\,H^2(X,X')}\right)
\\&\leq \; (c-a)\,\, \left(\frac{n^4\,\tau^4\,\,\E\,R^4(X,X')}{\big[\E\,H^2(X,X')\big]^2}  \right)^{1/2}\,.
\end{split}
\end{align}
Under Assumption \ref{ass2_new}, $\frac{n^4\,\tau^4\,\,\E\,R^4(X,X')}{\big[\E\,H^2(X,X')\big]^2} = o(1)$ as $n,p \to \infty$, and hence $\frac{n^4\,\tau^4\,\,\E\,R^4(X,X')}{\big[\E\,H^2(X,X')\big]^2}$ must be a bounded sequence in $n$ and $p$. Therefore, we have from (\ref{eq_sup_43})
\begin{align}\label{eq_sup_44}
\begin{split}
\E\,[I_G^2] \;&\lesssim\; c-a\,.
\end{split}
\end{align}
Likewise, it can be shown that 
\begin{align}\label{eq_sup_45}
\begin{split}
\E\,[II_G^2] \,&\lesssim\, c-a,\quad \E\,[III_G^2] \,\lesssim\,c-a,\quad \E\,[IV_G^2] \,\lesssim\,b-d,\quad \E\,[V_G^2] \,\lesssim\,b-d\,.
\end{split}
\end{align}
Combining (\ref{eq_sup_44})-(\ref{eq_sup_45}) with (\ref{eq_sup_40}) and (\ref{eq_sup_41}), we get
\begin{align*}
\E\,\big\vert G_n(u)\,-\,G_n(v) \big\vert^2 \;&\lesssim \; (c-a)\,+\,(b-d) \;=\; \Vert u-v \Vert_1 \;=\; \kappa^{\,2}(u,v)\,.
\end{align*}
This proves (\ref{eq_sup_39}) and thereby completes the proof of A3 and hence the theorem.
\end{proof}

\begin{proof}[Proof of Theorem \ref{sup_theorem3}]
It suffices to prove 
\begin{align}\label{eq_sup_45.5}
\dis\sup_{a,r,b\, \in\, [0,1]} \, \left\vert \frac{\widehat{V}^{\Delta \eta}_{n}(\lfloor nr \rfloor\,;\lfloor na \rfloor + 1,\lfloor nb \rfloor)}{\widetilde{V}^{\Delta \eta}_{n,}(\lfloor nr \rfloor\,;\lfloor na \rfloor + 1,\lfloor nb \rfloor))} - 1 \right\vert \;=\; o_p(1)\,
\end{align}
as $n,p \to \infty$\, for $\eta=1,2,3$. We will prove it for $\eta=2$, and other cases can be proved in a similar fashion. Denote\, $\omega(n\,;a,r) := (\lfloor nr \rfloor - \lfloor na \rfloor) (\lfloor nr \rfloor - \lfloor na \rfloor - 3)$. From equations (\ref{eq_sup_6.1}) and (\ref{new_eq_8}), we can write
\begin{align}\label{eq_sup_47}
\frac{\widehat{V}^{\Delta 2}_{n}(\lfloor nr \rfloor\,;\lfloor na \rfloor + 1,\lfloor nb \rfloor)}{\widetilde{V}^{\Delta 2}_{n}(\lfloor nr \rfloor\,;\lfloor na \rfloor + 1,\lfloor nb \rfloor))} \;&=\; \frac{4}{V_0}\, \widehat{\cal{D}}^2_{(\lfloor na \rfloor + 1):\lfloor nr \rfloor}\;=\; \frac{8}{V_0\,\,\omega(n\,;a,r)}\, \sum_{i=\lfloor na \rfloor + 2}^{\lfloor nr \rfloor} \sum_{j=\lfloor na \rfloor + 1}^{i-1} \widetilde{A}_{i,j}^2 \,,
\end{align}
where $A_{i,j} = \gamma(X_i,X_j)$ and $\widetilde{A}$ is the U-centered version of $A$. The last equality above in (\ref{eq_sup_47}) follows from the definition of $\widehat{\cal{D}}^2_{1:k}$ in Section \ref{sec:t-test}.\\ 

Define\, $C_n(a,r) := 8\,\sum_{i=\lfloor na \rfloor + 2}^{\lfloor nr \rfloor} \sum_{j=\lfloor na \rfloor + 1}^{i-1} \widetilde{A}_{i,j}^2$. Then we need to prove that 
\begin{align}\label{eq_sup_48}
\begin{split}
\dis\sup_{a,r\, \in\, [0,1]} \,  \frac{n^2}{\omega(n\,;a,r)}\,\left\vert \frac{1}{n^2\,V_0}\, C_n(a,r) - \frac{\omega(n\,;a,r)}{n^2} \right\vert \;=\; o_p(1)
\end{split}
\end{align}
as $n,p \to \infty$. Define $J_n(a,r) := \frac{1}{n^2\,V_0}\, C_n(a,r)$ and  $\widetilde{J}_n(a,r) := J_n(a,r) - \frac{\omega(n\,;\,a,r)}{n^2}$. Note that if we can prove 
\begin{align}\label{eq_sup_49.5}
\dis\sup_{a,r\, \in\, [0,1]} \, \left\vert \widetilde{J}_n(a,r) \right\vert \;=\; o_p(1)
\end{align}
as $n,p \to \infty$, then (\ref{eq_sup_48}) will follow by Slutsky's theorem. 
\vspace{0.1in}

Towards that, denote $u=(a,r)$ and $u'=(a',r')$. Consider the subset $[0,1]^2 \subseteq {\bb R}^2$ equipped with the metric $\widetilde{\kappa}(u,u') := \Vert u-u' \Vert$\, for\, $u,u' \in [0,1]^2$. By Theorem 1 in Andrews\,(1992), it suffices to show
\begin{enumerate}
\item[B1.] $[0,1]^2$ is totally bounded with respect to the metric $\widetilde{\kappa}$;
\item[B2.] Pointwise convergence: $\widetilde{J}_n(u)\, \overset{P}{\rightarrow} \,0 \;\;\;\forall\, u\in[0,1]^2$\, as\, $n,p \to \infty$;
\item[B3.] Asymptotic stochastic equicontinuity: for any $\epsilon>0$, $$\lim_{\delta \downarrow 0}\, \limsup_{n,p \to \infty} \, P\left(\, \sup_{\substack{u,u' \,\in\,[0,1]^2 \\ \widetilde{\kappa}(u,u') \leq \delta}} \big\vert \widetilde{J}_n(u)\,-\,\widetilde{J}_n(u') \big\vert\, \right) \;=\; 0\,.$$
\end{enumerate}

To argue B1, note that $[0,1]^2 \subseteq {\bb R}^2$ is compact (closed and bounded) with respect to the $l_2$ distance. It is easy to check that $[0,1]^2$ is compact (and therefore totally bounded) with respect to the metric $\widetilde{\kappa}$ as well. B2 is equivalent to showing \,$\frac{1}{\omega(n\,;\,a,r)}\, \frac{C_n(a,r)}{V_0} \overset{P}{\rightarrow} 1$ as $n,p \to \infty$ for fixed $a, r \in [0,1]$. The proof of B2 will follow similar lines of\, Lemma D.4\, in the Supplementary Materials of Chakraborty and Zhang\,(2021), which essentially proves the pointwise convergence result under Assumptions \ref{ass1_new} and \ref{ass2_new}. 
\vspace{0.1in}

Finally, to prove B3, again by Theorem A.8 in Li and Racine\,(2007) (also see Theorem 3 of Wichura, 1969), it will suffice to show that $\forall \,u,u' \in [0,1]^2$
\begin{align}\label{eq_sup_50}
\E\,\big\vert \widetilde{J}_n(u)\,-\,\widetilde{J}_n(u') \big\vert^2 \;&\lesssim \; \widetilde{\kappa}^{\,2}(u,v)\,.
\end{align}
Similar to the proof of T2 earlier in the proof of Theorem \ref{sup_theorem1}, we will show that (\ref{eq_sup_50}) holds in the case $a<a'<r'<r$. Similar arguments can prove (\ref{eq_sup_50}) for the other cases. 

Note that using the triangle inequality and the power mean inequality, we can write
\vspace{-0.1in}
\begin{align}\label{eq_sup_50.5}
\big\vert \widetilde{J}_n(u)\,-\,\widetilde{J}_n(u') \big\vert^2 \;\lesssim\; \big\vert J_n(u) - J_n(u')\big\vert^2 \,+\, \Bigg\vert\frac{\omega(n\,;a,r)}{n^2} - \frac{\omega(n\,;a',r')}{n^2}\Bigg\vert^2\,.
\end{align}
For $a_i, b_i \in {\bb R}$ with $\vert a_i \vert, \vert b_i \vert \leq 1$ for $1\leq i \leq n$, the product comparison lemma (Lemma 9.7.1 in Resnick, 1999) yields
\begin{align}\label{prod comp lemma}
\Bigg\vert \prod_{i=1}^n a_i \,-\, \prod_{i=1}^n b_i \Bigg\vert \;\leq\; \sum_{i=1}^n \vert  a_i \,-\, b_i \vert\,.
\end{align}
This yields
\begin{align}\label{eq_sup_50.6}
\begin{split}
\Bigg\vert\frac{\omega(n\,;a,r)}{n^2} - \frac{\omega(n\,;a',r')}{n^2} \Bigg\vert^2 \;&\leq \; \left(\frac{2}{n}\,\big((\lfloor nr \rfloor - \lfloor nr' \rfloor) \,+\, (\lfloor na' \rfloor - \lfloor na \rfloor)\big) \right)^2
\\&\lesssim \; \frac{1}{n^2}\,\big((\lfloor nr \rfloor - \lfloor nr' \rfloor)^2 \,+\, (\lfloor na' \rfloor - \lfloor na \rfloor)^2\big)\,,
\end{split}
\end{align}
where we have used the product comparison lemma and power mean inequality to get the first and the second inequalities, respectively. Following (\ref{eq_sup_31}), we can write
\begin{align*}
\lfloor nr \rfloor - \lfloor nr' \rfloor \;\leq \; 1 \,+\, n(r-r')\qquad
\textrm{and} \qquad \lfloor na' \rfloor - \lfloor na \rfloor \;\leq \; 1 \,+\, n(a'-a)\,.
\end{align*}
With this and using the power mean inequality once again, we have from (\ref{eq_sup_50.6})
\begin{align}\label{eq_sup_50.7}
\Bigg\vert\frac{\omega(n\,;a,r)}{n^2} - \frac{\omega(n\,;a',r')}{n^2} \Bigg\vert \; \lesssim \; (a-a')^2 \,+\, (r-r')^2 \;=\; \Vert u-u' \Vert^2 \;=\; \widetilde{\kappa}^{\,2}(u,u'),
\end{align}
and
\begin{align}\label{eq_sup_50.8}
\begin{split}
 &J_n(u)\,-\,J_n(u')  
\\=& \; \frac{1}{n^2 V_0}\, C_n(a,r)\,-\, \frac{1}{n^2 V_0}\, C_n(a',r')\\
=&\; \frac{8}{n^2 V_0}\, \left\{ \sum_{i=\lfloor na \rfloor +2}^{\lfloor nr \rfloor} \sum_{j=\lfloor na \rfloor + 1}^{i-1} \widetilde{A}_{i,j}^2\;\; -\; \sum_{i=\lfloor na' \rfloor +2}^{\lfloor nr' \rfloor} \sum_{j=\lfloor na' \rfloor + 1}^{i-1} \widetilde{A}_{i,j}^2 \right\}\\
=&\; \frac{8}{n^2 V_0}\,\Bigg\{ \sum_{i=\lfloor na \rfloor +2}^{\lfloor na' \rfloor} \sum_{j=\lfloor na \rfloor + 1}^{i-1} \widetilde{A}_{i,j}^2\,+\, \sum_{i=\lfloor na' \rfloor +1}^{\lfloor nr' \rfloor} \sum_{j=\lfloor na \rfloor + 1}^{\lfloor na' \rfloor} \widetilde{A}_{i,j}^2 \,+\, \sum_{i=\lfloor nr' \rfloor +1}^{\lfloor nr \rfloor} \sum_{j=\lfloor na \rfloor + 1}^{\lfloor na' \rfloor} \widetilde{A}_{i,j}^2\\
& \hspace{0.6in}  \,+\, \sum_{i=\lfloor nr' \rfloor +1}^{\lfloor nr \rfloor} \sum_{j=\lfloor na' \rfloor + 1}^{\lfloor nr' \rfloor} \widetilde{A}_{i,j}^2\,+\, \sum_{i=\lfloor nr' \rfloor +2}^{\lfloor nr \rfloor} \sum_{j=\lfloor nr' \rfloor + 1}^{i-1} \widetilde{A}_{i,j}^2 \Bigg\}\\
&=:\; J_1 \,+\, J_2\,+\, J_3\,+\, J_4\,+\, J_5\,.
\end{split}
\end{align}
By the power mean inequality,
\begin{align}\label{eq_sup_50.9}
\begin{split}
(J_1 \,+\, J_2\,+\, J_3\,+\, J_4\,+\, J_5)^2 \; \lesssim \; J_1^2 \,+\, J_2^2\,+\, J_3^2\,+\, J_4^2\,+\, J_5^2\,,
\end{split}
\end{align}
and therefore
\begin{align}\label{eq_sup_54}
\E\,\big\vert J_n(u)\,-\,J_n(u') \big\vert^2 \;&\lesssim \; \E\,[J_1^2] \,+\, \E\,[J_2^2] \,+\, \E\,[J_3^2] \,+\, \E\,[J_4^2] \,+\, \E\,[J_5^2]\,.
\end{align}
Consider the term $J_1$. Clearly 
\vspace{-0.2in}
\begin{align}\label{eq_sup_53}
\E\,[J_1^2] \;=\; \big(\E\,[J_1]\big)^2 \,+\, \var\,(J_1)\,.
\end{align}
Lemma D.3 in the Supplementary Materials of Chakraborty and Zhang\,(2021) essentially proves that under Assumptions \ref{ass1_new} and \ref{ass2_new}, $\var\,(J_1) = o(1)$ as $n, p \to \infty$. Following (\ref{eq_sup_47}), it is not hard to see that $$J_1 \;=\; \frac{\omega(n\,;a,a')}{n^2} \,\frac{\widehat{V}^{\Delta 2}_{n}(\lfloor na' \rfloor\,;\lfloor na \rfloor + 1,\lfloor nr \rfloor)}{\widetilde{V}^{\Delta 2}_{n}(\lfloor na' \rfloor\,;\lfloor na \rfloor + 1,\lfloor nr \rfloor))}\,.$$ 
Following (\ref{eq_sup_13.2}) and the proof of Lemma D.2 in the Supplementary Materials of Chakraborty and Zhang\,(2021), it can be verified that under Assumption \ref{ass2_new}, as $n,p \to \infty$
\begin{align}\label{eq_sup_54.5}
\frac{1}{(a'-a)^2}\,\E\,[J_1] \;&\rightarrow \;  1\,,
\end{align}
i.e., $\E\,[J_1]$ and hence $\E\,[J_1^2]$ is a bounded sequence in $n$ and $p$. Therefore we can write
\begin{align}\label{eq_sup_55}
\E\,[J_1^2]  \;&\lesssim \; (a-a')^4 \;\leq \; (a-a')^2\,.
\end{align}
In the same way we can obtain
\begin{align}\label{eq_sup_56}
\E\,[J_5^2]  \;&\lesssim \; (r-r')^2\,.
\end{align}
To obtain upper bounds for the terms $\E\,[J_2^2]$, $\E\,[J_3^2]$ and $\E\,[J_4^2]$, we first introduce the double centered distance $\bar{A}_{i,j} := A_{i,j} - \E\,[A_{i,j} | X_i] - \E\,[A_{i,j} | X_j] + \E\,[A_{i,j}]$ for $i \neq j$. We define $\bar{L}(X_i,X_j)$ and $\bar{R}(X_i,X_j)$ in a similar way. Following the proof of Lemma D.3 in the Supplementary Materials of Chakraborty and Zhang\,(2021), we can argue in a similar fashion that $\var\,(J_2), \var\,(J_3)$ and $\var\,(J_4)$ are  $o(1)$ as $n, p \to \infty$. Moreover, we have $\E\,[\widetilde{A}^2_{ij}] \lesssim \E\,[\bar{A}^2_{i,j}]$, $\bar{L}(X_i,X_j) = \frac{1}{\tau} H(X_i,X_j)$ and
\begin{align}\label{eq_CZ_1}
\begin{split}
\bar{A}_{i,j} \;&=\; \frac{\tau}{2}\, \bar{L}(X_i,X_j) \,+\, \tau\,\bar{R}(X_i,X_j) \;=\; \frac{1}{2}\, H(X_i,X_j) \,+\, \tau\,\bar{R}(X_i,X_j)\,.
\end{split}
\end{align}
With all these, we can write
\begin{align}\label{eq_sup_56.1}
\begin{split}
\frac{1}{V_0}\,\E\,[\widetilde{A}^2_{ij}] \;&\lesssim\;  \frac{1}{4} \,+\, \frac{\tau^2}{V_0}\,\E\,[\bar{R}^2(X,X')]\;\leq \; \frac{1}{4} \,+\, \,\left(\E\,\left[\frac{\tau^4}{V_0^2}\,\bar{R}^4(X,X')\right]\right)^{1/2}\,,
\end{split}
\end{align}
where the first and the second inequalities follow from the power mean inequality and H{\"o}lder's inequality, respectively. This implies 
\begin{align}\label{eq_sup_56.2}
\begin{split}
\E\,[J_2] \;&=\; \frac{8}{n^2 V_0}\,\sum_{i=\lfloor na' \rfloor +1}^{\lfloor nr' \rfloor} \sum_{j=\lfloor na \rfloor + 1}^{\lfloor na' \rfloor} \E\,[\widetilde{A}_{i,j}^2] 
\\ &\lesssim\;  \frac{1}{n^2}\,\left[(\lfloor nr' \rfloor - \lfloor na' \rfloor)(\lfloor na' \rfloor - \lfloor na \rfloor) \,+\, \,O\left(\left(\frac{n^4\,\tau^4}{V_0^2}\,\E\,[R^4(X,X')]\right)^{1/2}\right)\right]\,,
\end{split}
\end{align}
where we have used the fact that $\E\,[\bar{R}^4(X,X')] = O\left(\E\,[R^4(X,X')] \right)$. Following (\ref{eq_sup_31}) and under Assumption \ref{ass2_new}, we have from (\ref{eq_sup_56.2})  
\begin{align}\label{eq_sup_56.3}
\begin{split}
\E\,[J_2] \;& \lesssim\;(r'-a')(a'-a) \,+\, o(1)\,,
\end{split}
\end{align}
and therefore $(\E\,[J_2])^2 \lesssim (a-a')^2$, which in turn implies\, $\E\,[J_2^2] \lesssim (a-a')^2$ as $\var(J_2)=o(1)$ (and hence is a bounded sequence in $n$ and $p$). 

In similar lines, we can show that $\E\,[J_3^2] \lesssim (a'-a)^2$ and $\E\,[J_4^2] \lesssim (r-r')^2$. Combining all these, we have from (\ref{eq_sup_54}) and
\begin{align}\label{eq_sup_56.4}
\E\,\big\vert J_n(u)\,-\,J_n(u') \big\vert^2 \;&\lesssim \; (a-a')^2 \,+\,(r-r')^2 \;=\; \Vert u-u' \Vert^2 \;=\; \widetilde{\kappa}^{\,2}(u,u')\,.
\end{align}
Finally combining (\ref{eq_sup_50.5}), (\ref{eq_sup_50.7}) and (\ref{eq_sup_56.4}), we get 
\begin{align*}
\E\,\big\vert \widetilde{J}_n(u)\,-\,\widetilde{J}_n(u') \big\vert^2 \;&\lesssim \; \widetilde{\kappa}^{\,2}(u,v)\,,
\end{align*}
which completes the proof of B3 and hence Theorem \ref{sup_theorem3}.

\end{proof}

\begin{proof}[Proof of Lemma \ref{alt:decomposition}]
As a direct consequence of Proposition \ref{Prop 4.1 in CZ}, we have 
\[\widehat{E}_{n,k} = \widetilde{E}_{n,k} + \widetilde L_{n,k} + R_{n,k},\]
where $\widetilde{E}_{n,k}$ and $R_{n,k}$ have been defined earlier, and 
\begin{align*}
\widetilde L_{n,k} := &\frac{1}{k (n-k)}\sum_{i_1=1}^{k}\sum_{i_2=k+1}^{n}\tau_{i_1,i_2} L(X_{i_1}, X_{i_2})\,-\,\frac{1}{k(k-1)}\sum_{1 \leq i_1< i_2 \leq k} \tau_{i_1,i_2}L(X_{i_1},X_{i_2})\\ 
	&- \frac{1}{(n-k)(n-k-1)}\sum_{k+1 \leq i_1< i_2 \leq n}\tau_{i_1,i_2} L(X_{i_1},X_{i_2}).
\end{align*}
According to the definition of $H$,
\begin{align*}
	\tau_{i,j}L(X_i,X_j) =& \frac{\gamma^2(X_i,X_j) - \tau_{i,j}^2}{\tau_{i,j}}\\
	=& \frac{1}{\tau_{i,j}}\sum_{l = 1}^g\left[\rho_l(X_{i,\mathcal{S}_l}, X_{j,\mathcal{S}_l}) - \E[\rho_l(X_{i,\mathcal{S}_l}, X_{j,\mathcal{S}_l})]\right]\\
	=& H(X_i,X_j) + \frac{1}{\tau_{i,j}}\sum_{l = 1}^g\left(\E[\rho_l(X_{i,\mathcal{S}_l}, X_{j,\mathcal{S}_l})|X_{i,\mathcal{S}_l}] - \E[\rho_l(X_{i,\mathcal{S}_l}, X_{j,\mathcal{S}_l})]\right)\\
	&+\frac{1}{\tau_{i,j}}\sum_{l = 1}^g\left(\E[\rho_l(X_{i,\mathcal{S}_l}, X_{j,\mathcal{S}_l})|X_{j,\mathcal{S}_l}] - \E[\rho_l(X_{i,\mathcal{S}_l},X_{j,\mathcal{S}_l})]\right)\\
 =&H(X_i,X_j) + \tau_{i,j}\E(L(X_i,X_j)|X_i) + \tau_{i,j}\E(L(X_i,X_j)|X_j).
\end{align*}

We first consider the case of $k < \nu$. Let $(X,Y)$ be two independent random variables such that they are independent of $X_1,..., X_n$, and $X \overset{\mathcal{D}}{=} X_1$ and $Y \overset{\mathcal{D}}{=} X_{n}$. Then we have  
    \begin{align*}
        &\frac{1}{k(n-k)}\sum_{i = 1}^{k}\sum_{j = k+1}^n\tau_{i,j}\E[L(X_i,X_j)|X_i] + \tau_{i,j}\E[L(X_i,X_j)|X_j]\\
        =&\frac{1}{k(n-k)}\left(\sum_{i = 1}^k\sum_{j = k+1}^{\nu}\tau_1\E[L(X_i,X)|X_i] + \sum_{i = 1}^k\sum_{j = \nu+1}^{n}\tau_3\E[L(X_i,Y)|X_i]\right)\\
        &+\frac{1}{k(n-k)}\left(\sum_{i = 1}^k\sum_{j = k+1}^{\nu}\tau_1\E[L(X_j,X)|X_j] + \sum_{i = 1}^k\sum_{j = \nu+1}^{n}\tau_3\E[L(X_j,X)|X_j]\right)\\
        =&\frac{1}{k}\left(\frac{\nu-k}{n-k}\sum_{i = 1}^k\tau_1\E[L(X_i,X)|X_i] + \frac{n-\nu}{n-k}\sum_{i = 1}^k\tau_3\E[L(X_i,Y)|X_i]\right)\\
        &+\frac{1}{n-k}\left(\sum_{j = k+1}^{\nu}\tau_1\E[L(X_j,X)|X_j] + \sum_{j = \nu+1}^{n}\tau_3\E[L(X_j,X)|X_j]\right).
    \end{align*}
    And
    \begin{align*}
        &\frac{1}{k(k-1)}\sum_{1 \leq i < j \leq k}\tau_{i,j}\E[L(X_i,X_j)|X_i] + \tau_{i,j}\E[L(X_i,X_j)|X_j]\\
        =&\frac{1}{k(k-1)}\sum_{i = 1}^{k-1}\sum_{j = i+1}^k\tau_1\E[L(X_i,X)|X_i]+\frac{1}{k(k-1)}\sum_{j = 2}^{k}\sum_{i = 1}^{j-1}\tau_1\E[L(X_j,X)|X_j]\\
        =&\frac{1}{k(k-1)}\sum_{i = 1}^k(k-1)\E[L(X_i,X)|X_i] = \frac{1}{k}\sum_{i = 1}^k\E[L(X_i,X)|X_i].
    \end{align*}
    In addition, 
    \begin{align*}
        &\frac{1}{(n-k)(n-k-1)}\sum_{k+1 \leq i < j \leq n}\tau_{i,j}\E[L(X_i,X_j)|X_i] + \tau_{i,j}\E[L(X_i,X_j)|X_j]\\
        =&\frac{1}{(n-k)(n-k-1)}\sum_{k+1 \leq i < j \leq \nu}\tau_{i,j}\E[L(X_i,X_j)|X_i] + \tau_{i,j}\E[L(X_i,X_j)|X_j]\\
        &+\frac{1}{(n-k)(n-k-1)}\sum_{\nu+1 \leq i < j \leq n}\tau_{i,j}\E[L(X_i,X_j)|X_i] + \tau_{i,j}\E[L(X_i,X_j)|X_j]\\
        &+\frac{1}{(n-k)(n-k-1)}\sum_{i=k+1}^{\nu}\sum_{j = \nu+1}^n\tau_{i,j}\E[L(X_i,X_j)|X_i] + \tau_{i,j}\E[L(X_i,X_j)|X_j]\\
        =&\frac{1}{(n-k)(n-k-1)}\left(\sum_{i = k+1}^{\nu}(\nu-k-1)\tau_1\E[L(X_i,X)|X_i] + \sum_{i = \nu+1}^{n}(n-\nu-1)\tau_2\E[L(X_i,Y)|X_i]\right)\\
        &+\frac{1}{(n-k)(n-k-1)}\left(\sum_{i = k+1}^{\nu}(n-\nu)\tau_3\E[L(X_i,Y)|X_i] + \sum_{j = \nu+1}^{n}(\nu-k)\tau_3\E[L(X_j,X)|X_j]\right).
    \end{align*}
Combining the above results, we have
    \begin{align*}
        \widetilde{L}_{n,k} = &L_{n,k} + \frac{1}{k}\sum_{i = 1}^k\frac{n-\nu}{n-k}(\tau_3\E[L(X_i,Y)|X_i] - \tau_1\E[L(X_i,X)|X_i])
        \\&-\frac{1}{n-k}\sum_{i = k+1}^{\nu}\frac{n-\nu}{n-k-1}(\tau_3\E[L(X_i,Y)|X_i] - \tau_1\E[L(X_i,X)|X_i])\\
        &+\frac{1}{n-k}\sum_{i = \nu+1}^{n}\frac{n-\nu-1}{n-k-1}(\tau_3\E[L(X_i,X)|X_i] - \tau_2\E[L(X_i,Y)|X_i])\\
        =&L_{n,k} + U_{n,k}.
    \end{align*}
    Similar calculation yields that 
    \begin{align*}
        U_{n,k} =& \frac{1}{k}\sum_{i = 1}^{\nu}\frac{\nu-1}{k-1}(\tau_3\E[L(X_i,Y)|X_i] - \tau_1\E[L(X_i,X)|X_i])
        \\&-\frac{1}{k}\sum_{i = \nu+1}^k\frac{\nu}{k-1}(\tau_3\E[L(X_i,X)|X_i] - \tau_2\E[L(X_i,Y)|X_i])\\
        &+\frac{1}{n-k}\sum_{i = k+1}^n\frac{\nu}{k}(\tau_3\E[L(X_i,X)|X_i] - \tau_2\E[L(X_i,Y)|X_i]),
    \end{align*}
    for $k \geq \nu$.
Furthermore, simple algebra leads to
    \begin{align*}
        \widetilde E_{n,k} =& \frac{2}{k(n-k)}\left(\sum_{i = 1}^k\sum_{j = k+1}^{\nu}\tau_1 + \sum_{i = 1}^k\sum_{j = \nu+1}^{n}\tau_3\right) - \tau_1
        \\&- \frac{2}{(n-k)(n-k-1)}\left(\sum_{k+1 \leq i < j \leq \nu}\tau_1 + \sum_{\nu+1 \leq i < j \leq n}\tau_2 + \sum_{i = k+1}^{\nu}\sum_{j = \nu}^{n}\tau_3\right)\\
        =&\tau_1\left(\frac{2(\nu-k)}{n-k} - 1 - \frac{(\nu-k)(\nu-k-1)}{(n-k)(n-k-1)}\right) - \frac{(n-\nu)(n - \nu-1)}{(n-k)(n-k-1)}\tau_2 \\&+2\tau_3\left(\frac{n - \nu}{n-k} - \frac{(\nu-k)(n - \nu)}{(n-k)(n-k-1)}\right)\\
        =&(2\tau_3 - \tau_1 - \tau_2)\frac{(n-\nu)(n-\nu-1)}{(n-k)(n-k-1)}.
    \end{align*}
    Similar arguments can be applied to show that  \[\widetilde E_{n,k} = (2\tau_3 - \tau_1 - \tau_2)\frac{\nu(\nu-1)}{k(k-1)},\] and \[\widetilde L_{n,k} = L_{n,k}(1+o_p(1)),\]
    if $k \geq \nu$, which completes the proof.
\end{proof}

\begin{proof}[Proof of Lemma \ref{alt:jointnormal}]
    It suffices to show that for any $\alpha_1,\alpha_2,\alpha_3 \in \mathbb{R}$, $\alpha_1L_{n,k}^{(1)} + \alpha_2L_{n,k}^{(2)} + \alpha_3L_{n,k}^{(3)} \overset{d}{\rightarrow} N(0,\alpha_1^2+\alpha_2^2+\alpha_3^2)$. To see this, we further denote \[\xi_{n,j} = \left\{
    \begin{matrix}
    \alpha_1\frac{\sqrt{2}\zeta n}{\sqrt{V_1}}\frac{1}{\nu(\nu-1)}\sum_{i = 1}^{j-1}H(X_i,X_j),\qquad &\text{ for } j = 2,3,\dots,\nu;\\\\
    \alpha_3\frac{n\sqrt{\zeta(1-\zeta)}}{\sqrt{V_3}}\frac{1}{\nu(n- \nu)}\sum_{i = 1}^{\nu}H(X_i,X_{\nu+1}), \qquad &\text{ for } j = \nu+1;\\\\
    \alpha_2\frac{\sqrt{2}(1-\zeta)n}{\sqrt{V_2}}\frac{1}{(n - \nu)(n - \nu - 1)}\sum_{i = \nu+1}^{j-1}H(X_i,X_j) \\+ \alpha_3\frac{n\sqrt{\zeta(1-\zeta)}}{\sqrt{V_3}}\frac{1}{\nu(n- \nu)}\sum_{i = 1}^{\nu}H(X_i,X_{j}), \qquad &\text{ for } j = \nu+2,\dots,n.
    \end{matrix}
    \right.\]
    Consider the natural filtration $\{\mathcal{F}_j\}_{j \geq 0}$ where $\mathcal{F}_j = \sigma(X_j,X_{j-1}\dots)$. It can be easily seen that $\{\epsilon_{n,j}\}_{j \geq 2}$ is a martingale difference sequence relative to $\{\mathcal{F}_j\}_{j \geq 1}$, since it is adapted and for every $j$, $\E[\xi_{n,j}|\cF_{j-1}] = 0$ as $\E[H(X_i,X_j)|\cF_{j-1}] = 0$ for any $i < j$. Therefore $\{\sum_{j = 2}^{k}\xi_{n,j}\}_{k = 2}^n$ is a square integrable martingale relative to $\{\mathcal{F}_j\}_{j \geq 1}$. From Theorem 3.2 in Hall and Heyde (1980), if we can show
    \begin{enumerate}
        \item $\sum_{j = 2}^n\E[\xi_{n,j}^4] \rightarrow 0$,
        \item $\sum_{j = 2}^n\E[\xi_{n,j}^2|\cF_{j-1}] \overset{P}{\rightarrow}\alpha_1^2+\alpha_2^2+\alpha_3^2,$
    \end{enumerate}
    as both $n,p$ grow to $\infty$, then the proof is complete. To show the first condition, for $j = 2,\dots,\nu$,
    \begin{align*}
        \E[\xi_{n,j}^4] &= \frac{4\alpha_1^4n^4\zeta^4}{V_1^2\nu^4(\nu-1)^4}\sum_{i_1,i_2,i_3,i_4 = 1}^{j-1}\E[H(X_{i_1},X_j)H(X_{i_2},X_j)H(X_{i_3},X_j)H(X_{i_4},X_j)]\\
        &= \frac{4\alpha_1^4n^4\zeta^4}{V_1^2\nu^4(\nu-1)^4}\sum_{i = 1}^{j-1}\E[H^4(X_{i},X_j)] + \frac{12\alpha_1^4n^4\zeta^4}{V_1^2\nu^4(\nu-1)^4}\sum_{i_1,i_2 = 1, i_1 \neq i_2}^{j-1}\E[H^2(X_{i_1},X_j)H^2(X_{i_2},X_j)],\\
    \end{align*}
since \begin{align*}
    \E[H^3(X_{i_1},X_j)H(X_{i_2},X_j)] &= \E[\E[H^3(X_{i_1},X_j)H(X_{i_2},X_j)|X_{i_1},X_{j}]]
    \\&= \E[H^3(X_{i_1},X_j)\E[H(X_{i_2},X_j)|X_{i_1},X_{j}]] = 0,
\end{align*}
if $i_1 \neq i_2 \neq j$,
\begin{align*}
    \E[H^2(X_{i_1},X_j)H(X_{i_2},X_j)H(X_{i_3},X_j)] &= \E[\E[H^2(X_{i_1},X_j)H(X_{i_2},X_j)H(X_{i_3},X_j)|X_{i_1},X_{i_2},X_{j}]]\\
    &= \E[H^2(X_{i_1},X_j)H(X_{i_2},X_j)\E[H(X_{i_3},X_j)|X_{j}]] = 0,
\end{align*}
if $i_1, i_2, i_3, j$ are all distinct, and 
\begin{align*}
    &\E[H(X_{i_1},X_j)H(X_{i_2},X_j)H(X_{i_3},X_j)H(X_{i_4},X_j)] \\
    =& \E[\E[H(X_{i_1},X_j)H(X_{i_2},X_j)H(X_{i_3},X_j)H(X_{i_4},X_j)|X_{i_1},X_{i_2},X_{i_3},X_{j}]]\\
=&\E[H(X_{i_1},X_j)H(X_{i_2},X_j)H(X_{i_3},X_j)\E[H(X_{i_4},X_j)|X_{j}]] = 0,
\end{align*}
if $i_1,i_2,i_3,i_4,j$ are all distinct, as $X_1,\dots,X_n$ are independent. Therefore, by Assumption \ref{ass1_new}, and the fact that $\E[H^2(X_{i_1},X_j)H^2(X_{i_2},X_j)] \leq \E[H^4(X_1,X_1')]$, \[\E[\xi_{n,j}^4] = O\left(\frac{1}{n^3}\frac{\E[H^4(X,X')]}{V_1^2} + \frac{1}{n^2}\frac{\E[H^4(X,X')]}{V_1^2}\right) = o\left(n^{-1}\right).\]
Similarly, we have $\E[\xi_{n,j}^4] = o(n^{-2})$ for $j = \nu+1,\dots,n$ by noting that $\E[H(X_i,X_j)^4] = o(nV_2^2)$ if $\nu+1 \leq i < j$, and $\E[H(X_i,X_j)^4] = o(nV_3^2)$ if $i < \nu+1  \leq j$. In summary, $\sum_{j = 2}^n\E[\xi_{n,j}^4] \rightarrow 0$.

To show the second condition, assume $j = 2,\dots,\nu$ first. Some simple algebra leads to
\begin{align*}
    \E[\xi_{n,j}^2|\cF_{j-1}] =& \frac{2\alpha_1^2\zeta^2n^2}{V_1\nu^2(\nu-1)^2}\sum_{i_1,i_2 = 1}^{j-1}\E[H(X_{i_1},X_j)H(X_{i_2},X_j)|\cF_{j-1}]\\
    =& \frac{2\alpha_1^2\zeta^2n^2}{V_1\nu^2(\nu-1)^2}\sum_{i = 1}^{j-1}\E[H^2(X_{i},X_j)|\cF_{j-1}] 
    \\&+ \frac{2\alpha_1^2\zeta^2n^2}{V_1\nu^2(\nu-1)^2}\sum_{i_1,i_2 = 1, i_1 \neq i_2}^{j-1}\E[H(X_{i_1},X_j)H(X_{i_2},X_j)|\cF_{j-1}].
\end{align*}
By similar techniques, we can show
that
\begin{align*}
    \sum_{j = 2}^n\E[\xi_{n,j}^2|\cF_{j-1}] &= \sum_{j = 2}^{\nu}\frac{2\alpha_1^2\zeta^2n^2}{V_1\nu^2(\nu-1)^2}\sum_{i = 1}^{j-1}\E[H^2(X_{i},X_j)|\cF_{j-1}] \\
    & + \sum_{j = \nu + 2}^n\frac{2\alpha_2^2(1-\zeta)^2n^2}{V_2(n-\nu)^2(n-\nu-1)^2}\sum_{i = \nu+1}^{j-1}\E[H^2(X_i,X_j)|\cF_{j-1}]\\
    & + \sum_{j = \nu+1}^n\frac{\alpha_3^2n^2\zeta(1-\zeta)}{V_3\nu^2(1-\nu)^2}\sum_{i = 1}^{\nu}\E[H^2(X_i,X_j)|\cF_{j-1}]\\
    & + \sum_{j = 2}^{\nu}\frac{2\alpha_1^2\zeta^2n^2}{V_1\nu^2(\nu-1)^2}\sum_{i_1,i_2 = 1,i_1 \neq i_2}^{j-1}\E[H(X_{i_1},X_j)H(X_{i_2},X_j)|\cF_{j-1}] \\
    & + \sum_{j = \nu+2}^{n}\frac{2\alpha_2^2(1-\zeta)^2n^2}{V_2(n-\nu)^2(n-\nu-1)^2}\sum_{i_1,i_2 = \nu+1,i_1 \neq i_2}^{n}\E[H(X_{i_1},X_j)H(X_{i_2},X_j)|\cF_{j-1}] \\
    & + \sum_{j = \nu+1}^n\frac{\alpha_3^2n^2\zeta(1-\zeta)}{V_3\nu^2(1-\nu)^2}\sum_{i_1,i_2 = 1,i_1 \neq i_2}^{\nu}\E[H(X_{i_1},X_j)H(X_{i_2},X_j)|\cF_{j-1}]\\
    & + \sum_{j = \nu + 2}^{n}\frac{2\sqrt{2}\alpha_2\alpha_3\sqrt{\zeta(1-\zeta)^3}}{\sqrt{V_2V_3}\nu(n-\nu)^2(n-\nu-1)^2}\sum_{i_1 = \nu+1}^{j-1}\sum_{i_2 = 1}^{\nu}\E[H(X_{i_1},X_j)H(X_{i_2},X_j)|\cF_{j-1}]\\
    & := \mathcal{I}_1 + \mathcal{I}_2 + \mathcal{I}_3+\mathcal{I}_4+\mathcal{I}_5+\mathcal{I}_6+\mathcal{I}_7.
\end{align*}

If we can show that $\mathcal{I}_1$, $\mathcal{I}_2$, $\mathcal{I}_3$ converge to $\alpha_1^2$, $\alpha_2^2$ and $\alpha_3^2$ in probability, respectively, and $\mathcal{I}_4$,\dots,$\mathcal{I}_7$ converge to zero in probability, the proof is then complete. To see this, consider $\mathcal{I}_1$ first. Note that
\begin{align*}
    \E[\mathcal{I}_1] = \E\left[\sum_{j = 2}^{\nu}\frac{2\alpha_1^2\zeta^2n^2}{V_1\nu^2(\nu-1)^2}\sum_{i = 1}^{j-1}\E[H^2(X_{i},X_j)|\cF_{j-1}] \right]  = \frac{2\alpha_1^2\zeta^2n^2}{\nu^2(\nu-1)^2}\sum_{j = 2}^{\nu}(j-1) \rightarrow \alpha_1^2.
\end{align*}
Hence to prove $\mathcal{I}_1 \overset{P}{\rightarrow} \alpha_1^2$, by Chebyshev's inequality, it suffices to show that $\E[\mathcal{I}_1^2] \rightarrow \alpha_1^4$. Notice that
\begin{align*}
 \E[\mathcal{I}_1^2] &= \left(\frac{2\alpha_1^2\zeta^2n^2}{V_1\nu^2(\nu-1)^2}\right)^2 \sum_{j_1,j_2 = 2}^{\nu}\sum_{i_1 = 1}^{j_1-1}\sum_{i_2 = 1}^{j_2 - 1}\E\left[\E[H^2(X_{i_1},X_{j_1})|\cF_{j_1-1}]\E[H^2(X_{i_2},X_{j_2})|\cF_{j_2-1}]\right]\\
 &= \left(\frac{2\alpha_1^2\zeta^2n^2}{V_1\nu^2(\nu-1)^2}\right)^2 \sum_{j_1,j_2 = 2}^{\nu}\sum_{i_1 = 1}^{j_1-1}\sum_{i_2 = 1}^{j_2 - 1}\E\left[\E[H^2(X_{i_1},X_{j_1})|X_{i_1}]\E[H^2(X_{i_2},X_{j_2})|X_{i_2}]\right]\\
 &= \left(\frac{2\alpha_1^2\zeta^2n^2}{V_1\nu^2(\nu-1)^2}\right)^2 \sum_{j_1,j_2 = 2}^{\nu}\sum_{i_1 = 1}^{j_1-1}\sum_{i_2 = 1,i_2 \neq i_1}^{j_2 - 1}\E\left[\E[H^2(X_{i_1},X_{j_1})|X_{i_1}]\right]\E\left[\E[H^2(X_{i_2},X_{j_2})|X_{i_2}]\right]\\
 &+\left(\frac{2\alpha_1^2\zeta^2n^2}{V_1\nu^2(\nu-1)^2}\right)^2 \sum_{j_1,j_2 = 2}^{\nu}\sum_{i = 1}^{j_1\wedge j_2-1}\E\left[\E[H^2(X_{i},X_{j_1})|X_{i}]\E[H^2(X_{i},X_{j_2})|X_{i}]\right]\\
 &= \left(\frac{2\alpha_1^2\zeta^2n^2}{V_1\nu^2(\nu-1)^2}\right)^2\left\{\left(\frac{\nu(\nu-1)}{2}\right)^2V_1^2(1 + o(1)) +  \sum_{j_1,j_2 = 2}^{\nu}\sum_{i = 1}^{j_1\wedge j_2-1}\E\left[H^2(X_{i},X_{j_1})H^2(X_{i},X_{j_1'})\right]\right\}.
\end{align*}
By Assumption \ref{ass1_new}, $\E\left[H^2(X_{i},X_{j_1})H^2(X_{i},X_{j_1'})\right] \leq \E\left[H^4(X_{i},X_{j_1})\right] = o(nV_1^2)$. Therefore
\begin{align*}
    \E[\mathcal{I}_1^2]&= \left(\frac{2\alpha_1^2\zeta^2n^2}{V_1\nu^2(\nu-1)^2}\right)^2\left\{\left(\frac{\nu(\nu-1)}{2}\right)^2V_1^2(1 + o(1)) +  \sum_{j_1,j_2 = 2}^{\nu}\sum_{i = 1}^{j_1\wedge j_2-1}\E\left[H^2(X_{i},X_{j_1})H^2(X_{i},X_{j_1'})\right]\right\}\\
    &= \frac{\alpha_1^4\zeta^4n^4}{\nu^2(\nu-1)^2}(1 + o(1)) \rightarrow \alpha_1^4,
\end{align*}
which shows $\mathcal{I}_1 \overset{P}{\rightarrow} \alpha_1^4$. By similar arguments, we can also prove that $\mathcal{I}_2 \overset{P}{\rightarrow} \alpha_2^4$ and $\mathcal{I}_3 \overset{P}{\rightarrow} \alpha_3^4$. Now let us consider $\mathcal{I}_4$. Note that
\begin{align*}
    &\E\left[\mathcal{I}_4^2\right] 
    \\=&  \left(\frac{2\alpha_1^2\zeta^2n^2}{V_1\nu^2(\nu-1)^2}\right)^2\sum_{j_1,j_2 = 2}^{\nu}\sum_{i_1,i_2 = 1,i_1 \neq i_2}^{j_1-1}\sum_{i_3,i_4 = 1,i_3 \neq i_4}^{j_2-1} \E\Bigg[\E[H(X_{i_1},X_{j_1})H(X_{i_2},X_{j_1})|\cF_{j_1-1}]
    \\ & \hspace{3cm} \times \E[H(X_{i_3},X_{j_2})H(X_{i_4},X_{j_2})|\cF_{j_2-1}]\Bigg]\\
    =&  \left(\frac{4\alpha_1^2\zeta^2n^2}{V_1\nu^2(\nu-1)^2}\right)^2\sum_{j_1,j_2 = 2}^{\nu}\sum_{1 \leq i_1 < i_2 \leq j_1 -1}\sum_{1 \leq i_3 < i_4 \leq j_2 - 1}\E\Bigg[\E[H(X_{i_1},X_{j_1})H(X_{i_2},X_{j_1})|X_{i_1},X_{i_2}]\\ 
    & \hspace{3cm}\times \E[H(X_{i_3},X_{j_2})H(X_{i_4},X_{j_2})|X_{i_3},X_{i_4}]\Bigg].
\end{align*}
Denote $\widetilde{H}(X_{i_1},X_{i_2}) = \E[H(X_{i_1},X_{j_1})H(X_{i_2},X_{j_1})|X_{i_1},X_{i_2}]$. It is easy to see that \[\E\left[\E[H(X_{i_1},X_{j_1})H(X_{i_2},X_{j_1})|X_{i_1},X_{i_2}]\E[H(X_{i_3},X_{j_2})H(X_{i_4},X_{j_2})|X_{i_3},X_{i_4}]\right] = 0,\]
if $i_1 \neq i_3, i_1 \neq i_4, i_2 \neq i_3$ and $i_2 \neq i_4$. In addition, if $i_1 = i_3$, $i_2 \neq i_4$, then  
\begin{align*}
    &\E\left[\widetilde{H}(X_{i_1,i_2})\widetilde{H}(X_{i_1,i_4})\right] 
    \\ =& \E\left[\E[\widetilde{H}(X_{i_1,i_2})\widetilde{H}(X_{i_1,i_4})|X_{i_1}]\right] = \E\left[\E[\widetilde{H}(X_{i_1,i_2})|X_{i_1}]\E[\widetilde{H}(X_{i_1,i_4})|X_{i_1}]\right]\\
    =&\E\left[\left\{\E[\E[H(X_{i_1},X_{j_1})H(X_{i_2},X_{j_1})|X_{i_1},X_{i_2}]|X_{i_1}]\right\}^2\right]\\
    =&\E\left[\left\{\E[H(X_{i_1},X_{j_1})H(X_{i_2},X_{j_1})|X_{i_1}]\right\}^2\right]\\
    =&\E\left[\left\{\E[\E[H(X_{i_1},X_{j_1})H(X_{i_2},X_{j_1})|X_{i_1},X_{j_1}]|X_{i_1}]\right\}^2\right]\\
    =&\E\left[\left\{\E[H(X_{i_1},X_{j_1})\E[H(X_{i_2},X_{j_1})|X_{j_1}]|X_{i_1}]\right\}^2\right] = 0,
\end{align*}
since $\E[H(X_{i_2},X_{j_1})|X_{j_1}] = 0$. Similarly, we can show that $\E\left[\widetilde{H}(X_{i_1,i_2})\widetilde{H}(X_{i_1,i_4})\right] = 0$ if $i_1 < i_2 = i_3 < i_4$, $i_3 < i_4 = i_1 < i_2$, or $i_1 \neq i_3$ and $i_2= i_4$. Therefore, $\E[H(X_{i_2},X_{j_1})|X_{j_1}]$ can be nonzero only if $i_1 = i_3$ and $i_2 = i_4$, and 
\begin{align*}
    &\E\left[\mathcal{I}_4^2\right] 
    \\=&  \left(\frac{4\alpha_1^2\zeta^2n^2}{V_1\nu^2(\nu-1)^2}\right)^2\sum_{j_1,j_2 = 2}^{\nu}\sum_{1 \leq i_1 < i_2 \leq j_1 -1}\sum_{1 \leq i_3 < i_4 \leq j_2 - 1}\E\Bigg[\E[H(X_{i_1},X_{j_1})H(X_{i_2},X_{j_1})|X_{i_1},X_{i_2}]
    \\&\hspace{3cm} \times \E[H(X_{i_3},X_{j_2})H(X_{i_4},X_{j_2})|X_{i_3},X_{i_4}]\Bigg]\\
    =&  \left(\frac{4\alpha_1^2\zeta^2n^2}{V_1\nu^2(\nu-1)^2}\right)^2\sum_{j_1,j_2 = 2}^{\nu}\sum_{1 \leq i_1 < i_2 \leq j_1\wedge j_2 -1}\E\Bigg[\E[H(X_{i_1},X_{j_1})H(X_{i_2},X_{j_1})|X_{i_1},X_{i_2}]
    \\&\hspace{3cm} \times \E[H(X_{i_1},X_{j_2})H(X_{i_2},X_{j_2})|X_{i_1},X_{i_2}]\Bigg]\\
    =&\left(\frac{4\alpha_1^2\zeta^2n^2}{V_1\nu^2(\nu-1)^2}\right)^2\sum_{j_1,j_2 = 2}^{\nu}\sum_{1 \leq i_1 < i_2 \leq j_1\wedge j_2 -1}\E[H(X_{i_1},X_{j_1})H(X_{i_2},X_{j_1})H(X_{i_1},X_{j_2})H(X_{i_2},X_{j_2})]\\
    =& \left(\frac{4\alpha_1^2\zeta^2n^2}{V_1\nu^2(\nu-1)^2}\right)^2o(n^4V_1^2) = o(1) \rightarrow 0,
\end{align*}
by Assumption \ref{ass1_new}. This implies that $\mathcal{I}_4 \overset{P}{\rightarrow} 0$. Similarly, we can show that $\mathcal{I}_5, \mathcal{I}_6$ and $\mathcal{I}_7$ all converge to zero in probability. Hence, the proof is complete by combining all the results above.
\end{proof}

\begin{proof}[Proof of Lemma \ref{alt:jointnormal2}]
    Recall that 
    \begin{align*}
    U_{n,\nu} =& \frac{1}{\nu}\sum_{i = 1}^{\nu} (\tau_3\E[L(X_i,Y)|X_i] - \tau_1\E[L(X_i,X)|X_i]) 
    \\&+ \frac{1}{n-\nu}\sum_{i = \nu+1}^n(\tau_3\E[L(X_i,X)|X_i] - \tau_2\E[L(X_i,Y)|X_i]),    
    \end{align*}
    and note that $\tau_3\E[L(X_i,Y)|X_i] - \tau_1\E[L(X_i,X)|X_i]$ are i.i.d. for $i = 1,...,\nu$, and $\tau_3\E[L(X_i,X)|X_i] - \tau_2\E[L(X_i,Y)|X_i]$ are also i.i.d. for $i = \nu+1,...,n$. Therefore, by classical CLT, we have
    \[\frac{1}{\sqrt{\nu\Gamma_1}}\sum_{i = 1}^{\nu} (\tau_3\E[L(X_i,Y)|X_i] - \tau_1\E[L(X_i,X)|X_i]) \overset{\mathcal{D}}{\rightarrow} N(0,1),\]
    and
    \[\frac{1}{\sqrt{(n-\nu)\Gamma_2}}\sum_{i = \nu+1}^{n} (\tau_3\E[L(X_i,X)|X_i] - \tau_2\E[L(X_i,Y)|X_i]) \overset{\mathcal{D}}{\rightarrow} N(0,1).\]
    This concludes the proof.
\end{proof}

\begin{proof}[Proof of Lemma \ref{alt:denominator}]
    The first three results can be proved similarly using the arguments in the proof of Lemma D.4 in Chakraborty and Zhang (2021). Therefore, we omit the details here. 
    By the definition of $\widehat{S}^2$, it is straightforward to see that $\widehat{S}^2(\mathbf{X}_{1:\nu},\mathbf{X}_{(\nu+1):n}) = O_p(\max\{V_1,V_2,V_3\})$. And by the definition of $a_{\nu,n-\nu}$,
    \[n^2a_{\nu,n-\nu}^2 = n^2\left\{\frac{1}{\nu(n - \nu)} + \frac{1}{2\nu(\nu-1)} + \frac{1}{2(n-\nu)(n-\nu-1)}\right\} \rightarrow (2\zeta^2(1-\zeta)^2)^{-1},\]
    which completes the proof.
\end{proof}

\begin{proof}[Proof of Lemma \ref{alt:negligible}]
    By the definition of $R_{n,\nu}$ in (\ref{eq_sup_2}), it can be expressed as a linear combination of three terms. If we can prove that each term is of the order $o_p(n^{-1}\max(\sqrt{V_1},\sqrt{V_2},\sqrt{V_3}))$, then the proof is complete. 
    For the second term in $R_{n,\nu}$, we have
    \begin{align*}
        &\E\left[\left(\frac{1}{\nu(\nu-1)}\sum_{1\leq i_1 \neq i_2 \leq \nu}\tau_1R(X_{i_1},X_{i_2})\right)^2\right] 
        \\=& \frac{\tau_1^2}{\nu^2(\nu-1)^2}\sum_{1\leq i_1 \neq i_2 \leq \nu}\sum_{1\leq j_1 \neq j_2 \leq \nu}\E\left[R(X_{i_1},X_{i_2})R(X_{j_1},X_{j_2})\right]\\
        \leq&\frac{\tau_1^2}{\nu^2(\nu-1)^2}\sum_{1\leq i_1 \neq i_2 \leq \nu}\sum_{1\leq j_1 \neq j_2 \leq \nu}\sqrt{\E[R(X_{i_1},X_{i_2})^2]\E[R(X_{j_1},X_{j_2})^2]}\\
        =&\tau_1^2\E[R(X_1,X_2)^2]\leq \sqrt{\tau_1^4\E[R(X_1,X_2)^4]} = o(V_1/n^2),
    \end{align*}
    according to Assumption \ref{ass2_new}. This indicates, by the Chebyshev's inequality, \[\frac{1}{\nu(\nu-1)}\sum_{1\leq i_1 \neq i_2 \leq \nu}\tau_1R(X_{i_1},X_{i_2}) = o_p(n^{-1}\sqrt{V_1}) = o_p(n^{-1}\max
    \{\sqrt{V_1},\sqrt{V_2},\sqrt{V_3}\}).\]
    Similarly, we can prove using the same technique that the first term and third term are both $o_p(n^{-1}\sqrt{V_1}) = o_p(n^{-1}\max\{\sqrt{V_1},\sqrt{V_2},\sqrt{V_3}\})$. Hence, the proof is complete.
\end{proof}

\begin{proof}[Proof of Remark \ref{alt:illustration1}]
We have established that under Assumptions \ref{ass0_new}-\ref{ass2_new}, \( L_{n,\nu} = O_p(n^{-1} \max\{V_1, V_2, V_3\}) \) and \( U_{n,\nu} = O_p(n^{-1/2} \max\{\Gamma_1, \Gamma_2\}) \) (for further details, refer to the proofs of Lemmas \ref{alt:jointnormal}-\ref{alt:jointnormal2}). Moreover, either \( L_{n,\nu} \) or \( U_{n,\nu} \) will be the leading term in \( M_{n} \) (in addition to \( 2\tau_3 - \tau_1 - \tau_2 \)), while \( R_{n,\nu} \) remains asymptotically negligible. Assume that a random vector $X \sim F_1$, $Y \sim F_2$ and $X \bigCI Y$.
When considering \( \gamma \) as the Euclidean distance, we can calculate that:
\begin{align*}
& H(X, X') = -\frac{2(X - \mu_1)^\top(X' - \mu_1)}{\tau_1},\\
& H(Y, Y') = -\frac{2(Y - \mu_2)^\top(Y' - \mu_2)}{\tau_2},\\
& H(X, Y) = -\frac{2(X - \mu_1)^\top(Y - \mu_2)}{\tau_3}.
\end{align*}
This indicates that \( V_1, V_2, \) and \( V_3 \) are all of the order 
$$O_p\left(\frac{\max\{\text{tr}(\Sigma_1^2), \text{tr}(\Sigma_2^2), \text{tr}(\Sigma_1\Sigma_2)\}}{p}\right).$$ 
We also have the following expectations:
\[
\E[\tau_3 L(X,Y) - \tau_1 L(X,X') | X]=\frac{\tau_1 - \tau_3}{\tau_1 \tau_3} \left[\|X - \mu_1\|^2 - \text{tr}(\Sigma_1)\right] + 2\tau_3^{-1}(\mu_1 - \mu_2)^\top (X - \mu_1),
\]
and 
\[
\E[\tau_3 L(X,Y) - \tau_2 L(Y,Y') | Y]=\frac{\tau_2 - \tau_3}{\tau_2 \tau_3} \left[\|Y - \mu_2\|^2 - \text{tr}(\Sigma_2)\right] + 2\tau_3^{-1}(\mu_2 - \mu_1)^\top (Y - \mu_2).
\]
Consequently, we derive:
\[
\Gamma_1 = O\left((\tau_3 - \tau_1)^2 p^{-2} \text{var}(\|X - \mu_1\|^2) + p^{-1} (\mu_1 - \mu_2)^\top \Sigma_1 (\mu_1 - \mu_2)\right),
\]
and 
\[
\Gamma_2 = O\left((\tau_3 - \tau_2)^2 p^{-2} \text{var}(\|Y - \mu_2\|^2) + p^{-1} (\mu_2 - \mu_1)^\top \Sigma_2 (\mu_2 - \mu_1)\right).
\]
Therefore, if we additionally assume that the components of \( X_i \) are independent with a finite fourth moment, it can be shown that \( V_1, V_2, V_3 \asymp 1 \), 
$$\Gamma_1 = O_p\left(\frac{(\tau_3 - \tau_1)^2}{p} + \frac{\|\mu_1 - \mu_2\|^2}{p}\right) = O(1),$$ 
and 
$$\Gamma_2 = O\left(\frac{(\tau_3 - \tau_2)^2}{p} + \frac{\|\mu_1 - \mu_2\|^2}{p}\right) = O(1).$$ 
This results from the fact that \( \|\mu_1 - \mu_2\|^2 = O(p) \), which is implied by Assumption \ref{ass0_new}. Consequently, the desired results follow naturally.
\end{proof}

\begin{proof}[Proof of Lemma \ref{lem:alt_process}]
The lemma can be proved similarly to the proof of Theorem \ref{sup_theorem1}. The only difference between these two results is that in Theorem \ref{sup_theorem1}, $H(X_i,X_j)$'s are identically distributed for all $i$ and $j$, while in the current setup, $H(X_i,X_j)$'s are no longer identically distributed because of the presence of the change-point. We skip the details here.
\end{proof}

\begin{proof}[Proof of Lemma \ref{lem:alt_U}]
For the ease of notations, we let $X$ and $Y$ be two independent random vectors that are also independent of any $X_i$ for all $i$ and $X \sim F_1$ and $Y \sim F_2$. Consider the case of $k < \nu$ first. Recall that 
    \begin{align*}
      U_{n,k} =& \frac{1}{k}\sum_{i = 1}^k\frac{n-\nu}{n-k}(\tau_3\E[L(X_i,Y)|X_i] - \tau_1\E[L(X_i,X)|X_i])
      \\&-\frac{1}{n-k}\sum_{i = k+1}^{\nu}\frac{n-\nu}{n-k-1}(\tau_3\E[L(X_i,Y)|X_i] 
      - \tau_1\E[L(X_i,X)|X_i])\\
        &+\frac{1}{n-k}\sum_{i = \nu+1}^{n}\frac{n-\nu-1}{n-k-1}(\tau_3\E[L(X_i,X)|X_i] - \tau_2\E[L(X_i,Y)|X_i]).
  \end{align*}
  Therefore, we have
  \begin{align*}
      \sup_{k < \nu}\frac{k(n-k)}{n^2}|U_{n,k}| 
      \leq &\sup_{k < \nu}\frac{(n-k)}{n^2}\left|\sum_{i = 1}^k\tau_3\E[L(X_i,Y)|X_i] - \tau_1\E[L(X_i,X)|X_i]\right|\\
      &+ \sup_{k < \nu}\frac{k}{n^2}\left|\sum_{i = k+1}^{\nu}\tau_3\E[L(X_i,Y)|X_i] - \tau_1\E[L(X_i,X)|X_i]\right|\\
      &+\sup_{k < \nu}\frac{k}{n^2}\left|\sum_{i = \nu+1}^{n}\tau_3\E[L(X_i,X)|X_i] - \tau_2\E[L(X_i,Y)|X_i]\right|\\
      =& O_p(\sqrt{\Gamma_1/n}) + O_p(\sqrt{\Gamma_2/n}),
  \end{align*}
  where we have used the facts that
  \begin{align*}
    &\sup_{1 \leq a < b \leq \nu}\left|\sum_{i = a}^b\tau_3\E[L(X_i,Y)|X_i] - \tau_1\E[L(X_i,X)|X_i]\right| \\
    \leq& 2\sup_{1 \leq b \leq \nu}\left|\sum_{i = 1}^b\tau_3\E[L(X_i,Y)|X_i] - \tau_1\E[L(X_i,X)|X_i]\right|\\
     =& O_p(\sqrt{n\Gamma_1}),
  \end{align*}
    $\sup_{\nu \leq a < b \leq n}\left|\sum_{i = a}^b\tau_3\E[L(X_i,X)|X_i] - \tau_2\E[L(X_i,Y)|X_i]\right| = O_p(\sqrt{n\Gamma_2})$ and the Kolmogorov's inequality.
    It can be shown similarly for the case of $k \geq \nu$, which completes the proof.
\end{proof}
\begin{proof}[Proof of Lemma \ref{lem:alt_rem}]
We can adapt the same arguments as in the proofs of Theorem \ref{theorem1} and Theorem \ref{sup_theorem2}. Specifically, we define $\widetilde{R}_{n}(k,m) = \sum_{i_2 = k+1}^m\sum_{i_1 = k}^{i_2-1}\tau_{i_1,i_2}R{(X_{i_1},X_{i_2})}$, and follow the same steps therein by using Assumptions \ref{ass1_new}-\ref{ass2_new} instead of the assumptions under the null hypothesis. 
\end{proof}

\begin{proof}[Proof of Lemma \ref{lem:isolating_seeded}]
Fix $\ell \in \{1, \dots, N\}$ and let $t = \nu_\ell$. By the spacing condition, the interval $(t - \Delta_n, t + \Delta_n)$ contains no other change-points.
We select the scale $m$ to be the largest power of 2 satisfying $m \le \Delta_n$. Specifically, choose integer $j$ such that $2^{j-1} \le \Delta_n < 2^j$ and set $m = 2^{j-1}$. Note that this implies $\Delta_n/2 < m \le \Delta_n$.

The seeded family $\mathcal{I}_n$ contains intervals of length $m$ starting at $s_r = 1 + r(m/2)$ for $r=0,1,\dots$. These intervals form a grid where consecutive start points are spaced by $m/2$. Consequently, the collection of intervals $[s_r, s_r + m/2)$ covers the domain.
There exists an index $r$ such that $t \in [s_r, s_r + m/2)$. Let $I_\ell = [s_r, s_r + m - 1]$.

We verify the isolation property:
\begin{itemize}
    \item \textbf{Left Boundary:} Since $t < s_r + m/2$, we have $s_r > t - m/2$. Since $m \le \Delta_n$, it follows that $s_r > t - \Delta_n/2 > t - \Delta_n$. Thus, $I_\ell$ starts after $\nu_{\ell-1}$.
    \item \textbf{Right Boundary:} Since $s_r \le t$, the endpoint is $e_r = s_r + m - 1 \le t + m - 1$. Since $m \le \Delta_n$, we have $e_r < t + \Delta_n$. Thus, $I_\ell$ ends before $\nu_{\ell+1}$.
\end{itemize}
Therefore, $I_\ell \subset (\nu_{\ell-1}, \nu_{\ell+1})$, meaning it isolates $\nu_\ell$. The length constraint $|I_\ell| = m \in (\Delta_n/2, \Delta_n]$ is satisfied by construction.
\end{proof}

\begin{proof}[Proof of Lemma \ref{lem:SBS}]
    For the interval $I_{\ell} = [s_\ell, e_\ell]$ that contains only one change point and $(\nu_{\ell} - s_{\ell}+1)(e_{\ell} - \nu_{\ell}) \asymp |I_{\ell}|^2$, Theorem \ref{alt:fix_alternative} shows that when $b_{|I|,\ell} \rightarrow \infty$, $P(M(I_{\ell}) > \lambda_{I_\ell}) \rightarrow 1$ if $\lambda_{I_\ell} \rightarrow \infty$ and $\lambda_{I_\ell} = o(b_{|I_\ell|,\ell})$. 
    
    In the proof of Theorem \ref{alt:consistency}, we have shown that if $I_\ell$ contains only one change point at $\nu_\ell$, then 
    \begin{align*}
        &\frac{(b-s_\ell+1)(e_\ell-b)}{(e_\ell-s_\ell+1)^2}\widehat E_{\gamma}(\mathbf{X}_{s_\ell:b},\mathbf{X}_{(b+1):e_\ell})\\
        = & \frac{(b-e_\ell+1)(s_\ell-b)}{(e_\ell-s_\ell+1)^2}\left\{\delta_{\ell,\ell+1}r_b + L_{s_\ell,e_\ell,b} +U_{s_\ell,e_\ell,b}+R_{s_\ell,e_\ell,b} \right\},
    \end{align*}
    where \[r_b = \frac{(e_{\ell} - \nu_\ell)(e_\ell - \nu_\ell - 1)}{(e_\ell - b)(e_\ell - b - 1)}\]
    if $b \leq \nu_\ell$ and \[r_b = \frac{(\nu_\ell - s_\ell + 1)(\nu_\ell - s_\ell)}{(b - s_\ell + 1)(b - s_\ell)}\] if $b > \nu_\ell$. Let $L_{s_\ell,e_\ell,b}$, $U_{s_\ell,e_\ell,b}$ and $R_{s_\ell,e_\ell,b}$ be defined in the proof of Theorem 3.3 for the interval $[s_\ell,e_\ell]$ and split point $b$. 

    Lemmas \ref{lem:alt_process}-\ref{lem:alt_rem} in the proof of Theorem \ref{alt:consistency} show that under Assumptions \ref{ass0_new}-\ref{ass2_new}, $$\sup_{b = s_\ell+3,\dots,e_\ell-4}\frac{(b-e_\ell+1)(s_\ell - b)}{(e_\ell - s_\ell + 1)^2}\left|L_{s_\ell,e_\ell,b} +U_{s_\ell,e_\ell,b}+R_{s_\ell,e_\ell,b}\right| = O_p(\sqrt{V_{|I_\ell|,\ell}}/|I_{\ell}|).$$
    
     By the construction of the seeded intervals, $r_b \asymp 1$. Therefore, applying Theorem 3.5 to $I_\ell$, according to the proof of Theorem 3.5, \[P\left(\delta_{\ell,\ell+1}|\widehat 
     \nu(I_{\ell}) - \nu_{\ell}| \leq \sqrt{V_{|I_{\ell}|,\ell}}\right)\rightarrow 1.\]
   We can let $b_{n,\ell} = n\delta_{\ell,\ell+1}/\sqrt{V_{n,\ell}}$ and under Assumptions 3.1-3.3, we have shown $$P\!\left(b_{n,\ell}\big|\widehat\zeta(I_\ell)-\zeta_\ell\big|\le C\right)\to 1.$$

\end{proof}

\end{alphasection}

\end{document}